%% file: main.tex
\pdfoutput=1

\documentclass{article}

\PassOptionsToPackage{numbers, compress}{natbib}

\usepackage[preprint]{neurips_2024}

\usepackage[utf8]{inputenc} %
\usepackage[T1]{fontenc}    %
\usepackage{hyperref}       %
\usepackage{url}            %
\usepackage{booktabs}       %
\usepackage{amsfonts}       %
\usepackage{nicefrac}       %
\usepackage{microtype}      %
\usepackage{xcolor}         %

\usepackage{amsmath,amsthm,amsfonts,amssymb}
\usepackage{hyperref}
\usepackage{color}
\usepackage{mathrsfs}
\usepackage{multirow}
\usepackage{booktabs}
\usepackage{makecell}
\usepackage{graphicx}
\usepackage{subfigure}
\usepackage{caption}
\usepackage{thmtools}
\usepackage{thm-restate}
\usepackage{hhline}
\usepackage{times}
\usepackage{rotating} 
\usepackage{bm}

\usepackage{enumerate}
\usepackage{enumitem}
\setlist[enumerate]{itemsep=-1pt,topsep=-3pt}

\usepackage{nicefrac, xfrac}
\usepackage{mathrsfs}
\usepackage{xspace}
\usepackage{bm}
\usepackage{upgreek}
\usepackage{bbm}
\usepackage{mathabx}

\usepackage{cleveref}

\usepackage{siunitx}
\usepackage{pifont}%
\newcommand{\cmark}{\ding{51}}%
\newcommand{\xmark}{\ding{55}}%

\usepackage{algorithm}[l]
\usepackage[noend]{algpseudocode}

\newtheorem{theorem}{Theorem}

\newtheorem{definition}{Definition}

\newtheorem{lemma}{Lemma}
\newtheorem{example}{Example}
\newtheorem{remark}{Remark}

\input{definitions}

\title{Exploiting Approximate Symmetry for Efficient Multi-Agent Reinforcement Learning }

\author{%
  Batuhan Yardim \\
  Department of Computer Science\\
  ETH Z\"urich\\
  \texttt{alibatuhan.yardim@inf.ethz.ch} \\
  \And
  Niao He \\
  Department of Computer Science\\ 
  ETH Z\"urich\\
  \texttt{niao.he@inf.ethz.ch} \\
}

\begin{document}

\maketitle

\begin{abstract}
Mean-field games (MFG) have become significant tools for solving large-scale multi-agent reinforcement learning problems under symmetry.
However, the assumption of exact symmetry limits the applicability of MFGs, as real-world scenarios often feature inherent heterogeneity.
Furthermore, most works on MFG assume access to a known MFG model, which might not be readily available for real-world finite-agent games.
In this work, we broaden the applicability of MFGs by providing a methodology to extend any finite-player, possibly asymmetric, game to an ``induced MFG''.
First, we prove that $N$-player dynamic games can be symmetrized and smoothly extended to the infinite-player continuum via explicit Kirszbraun extensions.
Next, we propose the notion of $\alpha,\beta$-symmetric games, a new class of dynamic population games that incorporate approximate permutation invariance.
For $\alpha,\beta$-symmetric games, we establish explicit approximation bounds, demonstrating that a Nash policy of the induced MFG is an approximate Nash of the $N$-player dynamic game.
We show that TD learning converges up to a small bias using trajectories of the $N$-player game with finite-sample guarantees, permitting symmetrized learning without building an explicit MFG model.
Finally, for certain games satisfying monotonicity, we prove a sample complexity of $\widetilde{\mathcal{O}}(\varepsilon^{-6})$ for the $N$-agent game to learn an $\varepsilon$-Nash up to symmetrization bias. 
Our theory is supported by evaluations on MARL benchmarks with thousands of agents.
\end{abstract}

\input{intro}

\input{results}

\input{experiments}

\input{discussion}

\begin{ack}
This project is supported by Swiss National Science Foundation (SNSF) under the framework of NCCR Automation and SNSF Starting Grant.
\end{ack}

\bibliographystyle{abbrvnat}
\bibliography{references}

\appendix

\input{appendix/preliminaries}

\input{appendix/approximation}
\input{appendix/learning}

\input{appendix/monotonicity}

\input{appendix/experiment_details}

\end{document}

%% file: definitions.tex
\newcommand{\safemath}[2]{\newcommand{#1}{\ensuremath{#2}\xspace}}

\safemath{\bma}{\mathbf{a}}
\safemath{\bmb}{\mathbf{b}}
\safemath{\bmc}{\mathbf{c}}
\safemath{\bmd}{\mathbf{d}}
\safemath{\bme}{\mathbf{e}}
\safemath{\bmf}{\mathbf{f}}
\safemath{\bmg}{\mathbf{g}}
\safemath{\bmh}{\mathbf{h}}
\safemath{\bmi}{\mathbf{i}}
\safemath{\bmj}{\mathbf{j}}
\safemath{\bmk}{\mathbf{k}}
\safemath{\bml}{\mathbf{l}}
\safemath{\bmm}{\mathbf{m}}
\safemath{\bmn}{\mathbf{n}}
\safemath{\bmo}{\mathbf{o}}
\safemath{\bmp}{\mathbf{p}}
\safemath{\bmq}{\mathbf{q}}
\safemath{\bmr}{\mathbf{r}}
\safemath{\bms}{\mathbf{s}}
\safemath{\bmt}{\mathbf{t}}
\safemath{\bmu}{\mathbf{u}}
\safemath{\bmv}{\mathbf{v}}
\safemath{\bmw}{\mathbf{w}}
\safemath{\bmx}{\mathbf{x}}
\safemath{\bmy}{\mathbf{y}}
\safemath{\bmz}{\mathbf{z}}
\safemath{\bmzero}{\mathbf{0}}
\safemath{\bmone}{\mathbf{1}}
\safemath{\bmpi}{\pmb{\pi}}
\safemath{\bmalpha}{\pmb{\alpha}}
\safemath{\bmrho}{\pmb{\rho}}

\bmdefine{\biad}{a}
\bmdefine{\bibd}{b}
\bmdefine{\bicd}{c}
\bmdefine{\bidd}{d}
\bmdefine{\bied}{e}
\bmdefine{\bifd}{f}
\bmdefine{\bigd}{g}
\bmdefine{\bihd}{h}
\bmdefine{\biid}{i}
\bmdefine{\bijd}{j}
\bmdefine{\bikd}{k}
\bmdefine{\bild}{l}
\bmdefine{\bimd}{m}
\bmdefine{\bind}{n}
\bmdefine{\biod}{o}
\bmdefine{\bipd}{p}
\bmdefine{\biqd}{q}
\bmdefine{\bird}{r}
\bmdefine{\bisd}{s}
\bmdefine{\bitd}{t}
\bmdefine{\biud}{u}
\bmdefine{\bivd}{v}
\bmdefine{\biwd}{w}
\bmdefine{\bixd}{x}
\bmdefine{\biyd}{y}
\bmdefine{\bizd}{z}

\bmdefine{\bixid}{\xi}
\bmdefine{\bilambdad}{\lambda}
\bmdefine{\bimud}{\mu}
\bmdefine{\binud}{\nu}
\bmdefine{\bithetad}{\theta}
\bmdefine{\biomegad}{\omega}
\bmdefine{\biphid}{\phi}

\safemath{\bmia}{\biad}
\safemath{\bmib}{\bibd}
\safemath{\bmic}{\bicd}
\safemath{\bmid}{\bidd}
\safemath{\bmie}{\bied}
\safemath{\bmif}{\bifd}
\safemath{\bmig}{\bigd}
\safemath{\bmih}{\bihd}
\safemath{\bmii}{\biid}
\safemath{\bmij}{\bijd}
\safemath{\bmik}{\bikd}
\safemath{\bmil}{\bild}
\safemath{\bmim}{\bimd}
\safemath{\bmin}{\bind}
\safemath{\bmio}{\biod}
\safemath{\bmip}{\bipd}
\safemath{\bmiq}{\biqd}
\safemath{\bmir}{\bird}
\safemath{\bmis}{\bisd}
\safemath{\bmit}{\bitd}
\safemath{\bmiu}{\biud}
\safemath{\bmiv}{\bivd}
\safemath{\bmiw}{\biwd}
\safemath{\bmix}{\bixd}
\safemath{\bmiy}{\biyd}
\safemath{\bmiz}{\bizd}

\safemath{\bmxi}{\bixid}
\safemath{\bmlambda}{\bilambdad}
\safemath{\bmmu}{\bimud}
\safemath{\bmnu}{\binud}
\safemath{\bmtheta}{\bithetad}
\safemath{\bmomega}{\biomegad}
\safemath{\bmphi}{\biphid}

\safemath{\bA}{\mathbf{A}}
\safemath{\bB}{\mathbf{B}}
\safemath{\bC}{\mathbf{C}}
\safemath{\bD}{\mathbf{D}}
\safemath{\bE}{\mathbf{E}}
\safemath{\bF}{\mathbf{F}}
\safemath{\bG}{\mathbf{G}}
\safemath{\bH}{\mathbf{H}}
\safemath{\bI}{\mathbf{I}}
\safemath{\bJ}{\mathbf{J}}
\safemath{\bK}{\mathbf{K}}
\safemath{\bL}{\mathbf{L}}
\safemath{\bM}{\mathbf{M}}
\safemath{\bN}{\mathbf{N}}
\safemath{\bO}{\mathbf{O}}
\safemath{\bP}{\mathbf{P}}
\safemath{\bQ}{\mathbf{Q}}
\safemath{\bR}{\mathbf{R}}
\safemath{\bS}{\mathbf{S}}
\safemath{\bT}{\mathbf{T}}
\safemath{\bU}{\mathbf{U}}
\safemath{\bV}{\mathbf{V}}
\safemath{\bW}{\mathbf{W}}
\safemath{\bX}{\mathbf{X}}
\safemath{\bY}{\mathbf{Y}}
\safemath{\bZ}{\mathbf{Z}}

\safemath{\bZero}{\mathbf{0}}
\safemath{\bOne}{\mathbf{1}}
\safemath{\bDelta}{\mathbf{\Delta}}
\safemath{\bLambda}{\mathbf{\UpLambda}}
\safemath{\bPhi}{\mathbf{\Upphi}}
\safemath{\bSigma}{\mathbf{\Upsigma}}
\safemath{\bOmega}{\mathbf{\Upomega}}
\safemath{\bTheta}{\mathbf{\Uptheta}}

\bmdefine{\biAd}{A}
\bmdefine{\biBd}{B}
\bmdefine{\biCd}{C}
\bmdefine{\biDd}{D}
\bmdefine{\biEd}{E}
\bmdefine{\biFd}{F}
\bmdefine{\biGd}{G}
\bmdefine{\biHd}{H}
\bmdefine{\biId}{I}
\bmdefine{\biJd}{J}
\bmdefine{\biKd}{K}
\bmdefine{\biLd}{L}
\bmdefine{\biMd}{M}
\bmdefine{\biOd}{N}
\bmdefine{\biPd}{O}
\bmdefine{\biQd}{P}
\bmdefine{\biRd}{R}
\bmdefine{\biSd}{S}
\bmdefine{\biTd}{T}
\bmdefine{\biUd}{U}
\bmdefine{\biVd}{V}
\bmdefine{\biWd}{W}
\bmdefine{\biXd}{X}
\bmdefine{\biYd}{Y}
\bmdefine{\biZd}{Z}

\bmdefine{\biDelta}{\Delta}
\bmdefine{\biLambda}{\Lambda}
\bmdefine{\biPhi}{\Phi}
\bmdefine{\biSigma}{\Sigma}
\bmdefine{\biOmega}{\Omega}
\bmdefine{\biTheta}{\Theta}

\safemath{\bimA}{\biAd}
\safemath{\bimB}{\biBd}
\safemath{\bimC}{\biCd}
\safemath{\bimD}{\biDd}
\safemath{\bimE}{\biEd}
\safemath{\bimF}{\biFd}
\safemath{\bimG}{\biGd}
\safemath{\bimH}{\biHd}
\safemath{\bimI}{\biId}
\safemath{\bimJ}{\biJd}
\safemath{\bimK}{\biKd}
\safemath{\bimL}{\biLd}
\safemath{\bimM}{\biMd}
\safemath{\bimN}{\biNd}
\safemath{\bimO}{\biOd}
\safemath{\bimP}{\biPd}
\safemath{\bimQ}{\biQd}
\safemath{\bimR}{\biRd}
\safemath{\bimS}{\biSd}
\safemath{\bimT}{\biTd}
\safemath{\bimU}{\biUd}
\safemath{\bimV}{\biVd}
\safemath{\bimW}{\biWd}
\safemath{\bimX}{\biXd}
\safemath{\bimY}{\biYd}
\safemath{\bimZ}{\biZd}

\safemath{\bimDelta}{\biDelta}
\safemath{\bimLambda}{\biLambda}
\safemath{\bimPhi}{\biPhi}
\safemath{\bimSigma}{\biSigma}
\safemath{\bimOmega}{\biOmega}
\safemath{\bimTheta}{\biTheta}

\safemath{\setA}{\mathcal{A}}
\safemath{\setB}{\mathcal{B}}
\safemath{\setC}{\mathcal{C}}
\safemath{\setD}{\mathcal{D}}
\safemath{\setE}{\mathcal{E}}
\safemath{\setF}{\mathcal{F}}
\safemath{\setG}{\mathcal{G}}
\safemath{\setH}{\mathcal{H}}
\safemath{\setI}{\mathcal{I}}
\safemath{\setJ}{\mathcal{J}}
\safemath{\setK}{\mathcal{K}}
\safemath{\setL}{\mathcal{L}}
\safemath{\setM}{\mathcal{M}}
\safemath{\setN}{\mathcal{N}}
\safemath{\setO}{\mathcal{O}}
\safemath{\setP}{\mathcal{P}}
\safemath{\setQ}{\mathcal{Q}}
\safemath{\setR}{\mathcal{R}}
\safemath{\setS}{\mathcal{S}}
\safemath{\setT}{\mathcal{T}}
\safemath{\setU}{\mathcal{U}}
\safemath{\setV}{\mathcal{V}}
\safemath{\setW}{\mathcal{W}}
\safemath{\setX}{\mathcal{X}}
\safemath{\setY}{\mathcal{Y}}
\safemath{\setZ}{\mathcal{Z}}
\safemath{\emptySet}{\varnothing}

\safemath{\colA}{\mathscr{A}}
\safemath{\colB}{\mathscr{B}}
\safemath{\colC}{\mathscr{C}}
\safemath{\colD}{\mathscr{D}}
\safemath{\colE}{\mathscr{E}}
\safemath{\colF}{\mathscr{F}}
\safemath{\colG}{\mathscr{G}}
\safemath{\colH}{\mathscr{H}}
\safemath{\colI}{\mathscr{I}}
\safemath{\colJ}{\mathscr{J}}
\safemath{\colK}{\mathscr{K}}
\safemath{\colL}{\mathscr{L}}
\safemath{\colM}{\mathscr{M}}
\safemath{\colN}{\mathscr{N}}
\safemath{\colO}{\mathscr{O}}
\safemath{\colP}{\mathscr{P}}
\safemath{\colQ}{\mathscr{Q}}
\safemath{\colR}{\mathscr{R}}
\safemath{\colS}{\mathscr{S}}
\safemath{\colT}{\mathscr{T}}
\safemath{\colU}{\mathscr{U}}
\safemath{\colV}{\mathscr{V}}
\safemath{\colW}{\mathscr{W}}
\safemath{\colX}{\mathscr{X}}
\safemath{\colY}{\mathscr{Y}}
\safemath{\colZ}{\mathscr{Z}}

\safemath{\opA}{\mathbb{A}}
\safemath{\opB}{\mathbb{B}}
\safemath{\opC}{\mathbb{C}}
\safemath{\opD}{\mathbb{D}}
\safemath{\opE}{\mathbb{E}}
\safemath{\opF}{\mathbb{F}}
\safemath{\opG}{\mathbb{G}}
\safemath{\opH}{\mathbb{H}}
\safemath{\opI}{\mathbb{I}}
\safemath{\opJ}{\mathbb{J}}
\safemath{\opK}{\mathbb{K}}
\safemath{\opL}{\mathbb{L}}
\safemath{\opM}{\mathbb{M}}
\safemath{\opN}{\mathbb{N}}
\safemath{\opO}{\mathbb{O}}
\safemath{\opP}{\mathbb{P}}
\safemath{\opQ}{\mathbb{Q}}
\safemath{\opR}{\mathbb{R}}
\safemath{\opS}{\mathbb{S}}
\safemath{\opT}{\mathbb{T}}
\safemath{\opU}{\mathbb{U}}
\safemath{\opV}{\mathbb{V}}
\safemath{\opW}{\mathbb{W}}
\safemath{\opX}{\mathbb{X}}
\safemath{\opY}{\mathbb{Y}}
\safemath{\opZ}{\mathbb{Z}}
\safemath{\opZero}{\mathbb{O}}
\safemath{\identityop}{\opI}

\safemath{\veca}{\bma}
\safemath{\vecb}{\bmb}
\safemath{\vecc}{\bmc}
\safemath{\vecd}{\bmd}
\safemath{\vece}{\bme}
\safemath{\vecf}{\bmf}
\safemath{\vecg}{\bmg}
\safemath{\vech}{\bmh}
\safemath{\veci}{\bmi}
\safemath{\vecj}{\bmj}
\safemath{\veck}{\bmk}
\safemath{\vecl}{\bml}
\safemath{\vecm}{\bmm}
\safemath{\vecn}{\bmn}
\safemath{\veco}{\bmo}
\safemath{\vecp}{\bmmp}
\safemath{\vecq}{\bmq}
\safemath{\vecr}{\bmr}
\safemath{\vecs}{\bms}
\safemath{\vect}{\bmt}
\safemath{\vecu}{\bmu}
\safemath{\vecv}{\bmv}
\safemath{\vecw}{\bmw}
\safemath{\vecx}{\bmx}
\safemath{\vecy}{\bmy}
\safemath{\vecz}{\bmz}

\safemath{\veczero}{\bmzero}
\safemath{\vecone}{\bmone}
\safemath{\vecxi}{\bmxi}
\safemath{\veclambda}{\bmlambda}
\safemath{\vecmu}{\bmmu}
\safemath{\vecnu}{\bmnu}

\safemath{\vecomega}{\bmomega}
\safemath{\vectheta}{\bmtheta}
\safemath{\vecphi}{\bmphi}
\safemath{\vecpi}{\bmpi}
\safemath{\vecalpha}{\bmalpha}
\safemath{\vecrho}{\bmrho}

\safemath{\matA}{\bA}
\safemath{\matB}{\bB}
\safemath{\matC}{\bC}
\safemath{\matD}{\bD}
\safemath{\matE}{\bE}
\safemath{\matF}{\bF}
\safemath{\matG}{\bG}
\safemath{\matH}{\bH}
\safemath{\matI}{\bI}
\safemath{\matJ}{\bJ}
\safemath{\matK}{\bK}
\safemath{\matL}{\bL}
\safemath{\matM}{\bM}
\safemath{\matN}{\bN}
\safemath{\matO}{\bO}
\safemath{\matP}{\bP}
\safemath{\matQ}{\bQ}
\safemath{\matR}{\bR}
\safemath{\matS}{\bS}
\safemath{\matT}{\bT}
\safemath{\matU}{\bU}
\safemath{\matV}{\bV}
\safemath{\matW}{\bW}
\safemath{\matX}{\bX}
\safemath{\matY}{\bY}
\safemath{\matZ}{\bZ}
\safemath{\matzero}{\bmzero}

\safemath{\matDelta}{\bDelta}
\safemath{\matLambda}{\bLambda}
\safemath{\matPhi}{\bPhi}
\safemath{\matSigma}{\bSigma}
\safemath{\matOmega}{\bOmega}
\safemath{\matTheta}{\bTheta}

\safemath{\matidentity}{\matI}
\safemath{\matone}{\matO}

\safemath{\rnda}{A}
\safemath{\rndb}{B}
\safemath{\rndc}{C}
\safemath{\rndd}{D}
\safemath{\rnde}{E}
\safemath{\rndf}{F}
\safemath{\rndg}{G}
\safemath{\rndh}{H}
\safemath{\rndi}{I}
\safemath{\rndj}{J}
\safemath{\rndk}{K}
\safemath{\rndl}{L}
\safemath{\rndm}{M}
\safemath{\rndn}{N}
\safemath{\rndo}{O}
\safemath{\rndp}{P}
\safemath{\rndq}{Q}
\safemath{\rndr}{R}
\safemath{\rnds}{S}
\safemath{\rndt}{T}
\safemath{\rndu}{U}
\safemath{\rndv}{V}
\safemath{\rndw}{W}
\safemath{\rndx}{X}
\safemath{\rndy}{Y}
\safemath{\rndz}{Z}

\safemath{\rveca}{\bimA}
\safemath{\rvecb}{\bimB}
\safemath{\rvecc}{\bimC}
\safemath{\rvecd}{\bimD}
\safemath{\rvece}{\bimE}
\safemath{\rvecf}{\bimF}
\safemath{\rvecg}{\bimG}
\safemath{\rvech}{\bimH}
\safemath{\rveci}{\bimI}
\safemath{\rvecj}{\bimJ}
\safemath{\rveck}{\bimK}
\safemath{\rvecl}{\bimL}
\safemath{\rvecm}{\bimM}
\safemath{\rvecn}{\bimN}
\safemath{\rveco}{\bomO}
\safemath{\rvecp}{\bimP}
\safemath{\rvecq}{\bimQ}
\safemath{\rvecr}{\bimR}
\safemath{\rvecs}{\bimS}
\safemath{\rvect}{\bimT}
\safemath{\rvecu}{\bimU}
\safemath{\rvecv}{\bimV}
\safemath{\rvecw}{\bimW}
\safemath{\rvecx}{\bimX}
\safemath{\rvecy}{\bimY}
\safemath{\rvecz}{\bimZ}

\safemath{\rvecxi}{\bmxi}
\safemath{\rveclambda}{\bmlambda}
\safemath{\rvecmu}{\bmmu}
\safemath{\rvectheta}{\bmtheta}
\safemath{\rvecphi}{\bmphi}

\safemath{\rmatA}{\bimA}
\safemath{\rmatB}{\bimB}
\safemath{\rmatC}{\bimC}
\safemath{\rmatD}{\bimD}
\safemath{\rmatE}{\bimE}
\safemath{\rmatF}{\bimF}
\safemath{\rmatG}{\bimG}
\safemath{\rmatH}{\bimH}
\safemath{\rmatI}{\bimI}
\safemath{\rmatJ}{\bimJ}
\safemath{\rmatK}{\bimK}
\safemath{\rmatL}{\bimL}
\safemath{\rmatM}{\bimM}
\safemath{\rmatN}{\bimN}
\safemath{\rmatO}{\bimO}
\safemath{\rmatP}{\bimP}
\safemath{\rmatQ}{\bimQ}
\safemath{\rmatR}{\bimR}
\safemath{\rmatS}{\bimS}
\safemath{\rmatT}{\bimT}
\safemath{\rmatU}{\bimU}
\safemath{\rmatV}{\bimV}
\safemath{\rmatW}{\bimW}
\safemath{\rmatX}{\bimX}
\safemath{\rmatY}{\bimY}
\safemath{\rmatZ}{\bimZ}

\safemath{\rmatDelta}{\bimDelta}
\safemath{\rmatLambda}{\bimLambda}
\safemath{\rmatPhi}{\bimPhi}
\safemath{\rmatSigma}{\bimSigma}
\safemath{\rmatOmega}{\bimOmega}
\safemath{\rmatTheta}{\bimTheta}

\usepackage{mathrsfs}
\usepackage{xspace}
\usepackage{bm}

\newenvironment{textbmatrix}{	\setlength{\arraycolsep}{2.5pt}%
								\big[\begin{matrix}}{\end{matrix}\big]%
								\raisebox{0.08ex}{\vphantom{M}}}

\def\be{\begin{equation}}
\def\ee{\end{equation}}
\def\een{\nonumber \end{equation}}
\def\mat{\begin{bmatrix}}
\def\emat{\end{bmatrix}}
\def\btm{\begin{textbmatrix}}
\def\etm{\end{textbmatrix}}

\def\ba#1\ea{\begin{align}#1\end{align}}
\def\bas#1\eas{\begin{align*}#1\end{align*}}
\def\bs#1\es{\begin{split}#1\end{split}} 
\def\bg#1\eg{\begin{gather}#1\end{gather}} 
\def\bi#1\ei{\begin{itemize}#1\end{itemize}}

\DeclareMathOperator*{\argmax}{arg\;max}		%
\DeclareMathOperator{\Prob}{\opP}			%
\DeclareMathOperator{\Exop}{\opE}			%

\newcommand{\ind}[1]{\mathbbm{1}_{#1}}

\safemath{\dirac}{\delta}					%
\safemath{\krond}{\dirac}					%
\safemath{\upto}{\uparrow}
\safemath{\downto}{\downarrow}
\safemath{\iu}{j}							%
\safemath{\ev}{\lambda}						%
\safemath{\hilseqspace}{l^{2}}				%
\newcommand{\banachfunspace}[1]{\setL^{#1}}	%
\safemath{\hilfunspace}{\banachfunspace{2}}	%

\safemath{\SNR}{\text{\sc snr}} 				%
\safemath{\No}{N_0}							%
\safemath{\Es}{E_s}							%
\safemath{\Eb}{E_b}							%
\safemath{\EbNo}{\frac{\Eb}{\No}}
\safemath{\EsNo}{\frac{\Es}{\No}}

\DeclareMathOperator{\CHop}{\ensuremath{\opH}} %
\safemath{\tvir}{\rndh_{\CHop}}				%
\safemath{\tvtf}{\rndl_{\CHop}}				%
\safemath{\spf}{\rnds_{\CHop}}				%
\safemath{\bff}{H_{\CHop}}					%

\safemath{\ircf}{r_{h}}						%
\safemath{\tftvcf}{r_{s}}					%
\safemath{\tfcf}{r_{l}}						%
\safemath{\bfcf}{r_{H}}						%

\safemath{\tcorr}{c_h}						%
\safemath{\scf}{c_{s}}						%
\safemath{\tfcorr}{c_{l}}					%
\safemath{\fcorr}{c_{H}}						%

\safemath{\mi}{I}							%
\safemath{\capacity}{C}						%

\safemath{\normal}{\mathcal{N}}			%
\safemath{\jpg}{\mathcal{CN}}			%
\safemath{\mchain}{\leftrightarrow}		%

\safemath{\dB}{\,\mathrm{dB}}
\safemath{\dBm}{\,\mathrm{dBm}}
\safemath{\Hz}{\,\mathrm{Hz}}
\safemath{\kHz}{\,\mathrm{kHz}}
\safemath{\MHz}{\,\mathrm{MHz}}
\safemath{\GHz}{\,\mathrm{GHz}}
\safemath{\s}{\,\mathrm{s}}
\safemath{\ms}{\,\mathrm{ms}}
\safemath{\mus}{\,\mathrm{\mu s}}
\safemath{\ns}{\,\mathrm{ns}}
\safemath{\meter}{\,\mathrm{m}}
\safemath{\mm}{\,\mathrm{mm}}
\safemath{\cm}{\,\mathrm{cm}}
\safemath{\m}{\,\mathrm{m}}
\safemath{\W}{\,\mathrm{W}}
\safemath{\J}{\,\mathrm{J}}
\safemath{\K}{\,\mathrm{K}}
\safemath{\bit}{\,\mathrm{bit}}

\safemath{\define}{=}			%

\safemath{\equivalent}{\sim}
\safemath{\distas}{\sim}					%
\safemath{\sdiff}{\Delta}				%

\safemath{\reals}{\mathbb{R}}
\safemath{\positivereals}{\reals_{+}}
\safemath{\integers}{\mathbb{Z}}
\safemath{\posint}{\integers_{+}}
\safemath{\naturals}{\mathbb{N}}
\safemath{\posnaturals}{\naturals_{+}}
\safemath{\complexset}{\mathbb{C}}
\safemath{\rationals}{\mathbb{Q}}

\newcommand{\Sym}{\mathbb{S}}
\newcommand{\lift}[1]{\widebar{#1}}
\newcommand{\lipext}[1]{\operatorname{Ext}\left(#1\right)}
\newcommand{\empc}[1]{\sigma(#1)}
\newcommand{\symmetrization}[1]{\operatorname{Sym}\left( #1\right)}
\newcommand{\symmbar}[1]{\widebar{\operatorname{Sym}}\left( #1\right)}
\newcommand{\inducedmfg}[1]{\operatorname{MFG}\left( #1\right)}

\newcommand{\Vfin}{V}
\newcommand{\Expfin}{\setE}
\newcommand{\gpop}{\Gamma}
\newcommand{\lpop}{\Lambda}
\newcommand{\JfinNi}[1]{J^{(#1)}}
\newcommand{\ExpfinNi}[1]{\mathcal{E}^{(#1)}}
\newcommand{\initpop}{\rho_0}
\newcommand{\lpopmu}{L_{pop,\mu}}

\newcommand{\refpolvec}{\widebar{\pi}}
\newcommand{\refpol}{\widebar{\pi}}
\newcommand{\cF}{\mathcal{F}}
\newcommand{\cO}{\mathcal{O}}

\newcommand{\empdist}[1]{\widehat{\mu}_{#1}} 
\newcommand{\EE}[1]{\mathbb{E}\left[#1\right]}

\newcommand{\EEc}[2]{\mathbb{E}\left[#1\left|#2\right.\right]}

\newcommand{\Qpi}[1]{Q^{\tau,\pi}_{#1}} 
\newcommand{\qpi}[1]{q^{\tau,\pi}_{#1}} 
\newcommand{\Qest}[2]{\widehat{Q}^{#1}_{#2}}
\newcommand{\qest}[2]{\widehat{q}^{#1}_{#2}}
\newcommand{\Qmax}{Q_{\text{max}}}

\newcommand{\lr}[2]{\eta_{#1}} 
\newcommand{\plr}[1]{\xi_{#1}} 
\newcommand{\entropy}{\setH}
\newcommand{\kl}{D_{\text{KL}}}

%% file: intro.tex
\section{Introduction}

Competitive multi-agent reinforcement learning (MARL) has found a wide range of applications in the recent years \cite{shavandi2022multi, wiering2000multi,samvelyan2019starcraft, rashedi2016markov, matignon2007hysteretic, mao2022mean}.
Simultaneously, MARL is fundamentally challenging at the regime with many agents due to an exponentially growing search space \citep{wang2020breaking}, also known as the \emph{curse-of-many-agents}.
Even finding an \emph{approximate} solution (i.e. approximate \emph{Nash}) is \textsc{PPAD}-hard \cite{daskalakis2023complexity}, thus potentially intractable.
For these reasons, it has been an active area of research to identify ``islands of tractability'', where MARL can be solved efficiently (see e.g. \cite{leonardos2021global, perrin2020fictitious}).

In this work, we develop a theory of efficient learning for MARL problems that exhibit \emph{approximate symmetry} building upon the theory of mean-field games (MFG).
MFG is a common theoretical framework for breaking the curse of many agents under perfect symmetry.
Initially proposed by \cite{lasry2007mean} and \cite{huang2006large}, MFG analyzes $N$-player games with symmetric agents when $N$ is large.
In this setting, the so-called propagation of chaos permits the reduction of the $N$-player game to a game between a representative agent and a population distribution.
This theoretical framework has been widely studied in many recent works \cite{anahtarci2022q, guo2022general, perolat2022scaling, perrin2020fictitious, xie2021learning, yardim2023policy}.

However, works on MFG exhibit two major bottlenecks preventing wider applicability in MARL.
First and foremost, the aforementioned works on MFG all assume some form of exact symmetry between agents.
Namely, in the MFGs, all agents must have the same reward function and dynamics must be homogeneous (or permutation invariant) among agents.
Such perfect symmetry between agents in MARL is theoretically convenient yet practically infeasible:
Even in applications where symmetry is presumed, usually, there are imperfections in dynamics that break invariance.
Little research has studied whether MFGs could offer tractable approximations to otherwise intractable games that might exhibit approximate symmetries.
Secondly, many works on MFG (such as \cite{guo2019learning, perolat2022scaling}) implicitly assume that an exact model of the MFG is known to the algorithm akin to solving a known MDP.
In real-world applications, an exact MFG model might not be readily available.
MFGs can potentially address settings where only $N$-player dynamics (possibly incorporating imperfections and heterogeneity) can be simulated;
however, such a theory of MFGs has yet to be developed.

We address these shortcomings by developing a theoretically sound MARL framework for scenarios when permutation invariance holds only approximately.
Unlike previous work on MFG, our theoretical approach is \emph{constructive}: we show that given any MARL problem, one can construct an MFG approximation that permits efficient learning.
We define a new, broad class of games with approximate permutation invariance, dubbed \emph{$\alpha,\beta$-symmetric} games, for which approximate Nash equilibria can be learned efficiently.
Our theoretical framework provides \emph{end-to-end} learning guarantees for policy mirror descent combined with TD learning. 
Our experimental findings further demonstrate strong performance improvements in MARL problems with thousands of agents.

\subsection{Related Work}

\begin{table*}[t]
  \begin{tabular}{llccc} \toprule
    \textbf{Work} & \textbf{Symmetry} & \textbf{Approximation} & \textbf{Learning} & \textbf{Learn w/o model} \\ \midrule
    \citeauthor{saldi2018markov}, \citeyear{saldi2018markov} & Exact & \cmark \emph{(asymptotic)} & \xmark & - \\
    \citeauthor{yardim2024meanfield}, \citeyear{yardim2024meanfield} & Exact & \cmark \emph{(explicit)} & \xmark & - \\
    \citeauthor{cui2021approximately}, \citeyear{cui2021approximately} & Exact & \cmark \emph{(asymptotic)} & \cmark \emph{(reg.)} & \xmark \\
    \citeauthor{zaman2023oracle}, \citeyear{zaman2023oracle} & Exact & \xmark & \cmark \emph{(reg.)} & \xmark \\
    \citeauthor{yardim2023policy}, \citeyear{yardim2023policy} & Exact & \xmark & \cmark \emph{(reg.)} & \cmark \\
    \citeauthor{parise2019graphon}, \citeyear{parise2019graphon} & Graphon  & \cmark \emph{(explicit)} & \xmark & - \\
    \citeauthor{zhang2023learning}, \citeyear{zhang2023learning} & Graphon  & \xmark & \cmark \emph{(mon.)} & \xmark \\
    \citeauthor{perolat2022scaling}, \citeyear{perolat2022scaling} & Multi-pop. & \xmark & \cmark \emph{(mon.)} & \xmark \\
     \midrule
    \textbf{Our work} & $\alpha,\beta$-symm. & \cmark (\emph{explicit}) & \cmark (\emph{mon.}) & \cmark \\
    \bottomrule
  \end{tabular}
  \caption{Selected models of symmetric games studied in MF-RL works. 
  (\emph{reg.}: only Nash with regularization strictly bounded away from zero, \emph{mon.}: monotonicity assumption)
  }
  \label{tab:selected_works_approx}
\end{table*}

We compare our work with selected past MFG results in Table~\ref{tab:selected_works_approx}, and also provide a detailed commentary in this section.

\textbf{Mean-field games and RL.}
MFGs represent a particular type of competitive game where players exhibit strong symmetries.
Past work has studied the existence of MFG Nash equilibrium as well as its approximation of finite-player Nash \cite{carmona2013probabilistic,carmona2018probabilistic, saldi2018markov}.
The convergence of RL algorithms has also been widely studied in discrete-time MFG assuming either contractivity in the stationary equilibrium setting \cite{zaman2023oracle, yardim2023policy, guo2019learning, xie2021learning, cui2021approximately} or monotonicity in the finite-horizon setting \cite{perrin2020fictitious, perolat2022scaling,yardim2023stateless,perolat2015approximate,perrin2022generalization}.
These models however assume exact homogeneity between all participants.
Furthermore, existing algorithms typically assume knowledge of the exact MFG model \cite{guo2019learning, zaman2023oracle}, hindering their real-world applicability.
Multi-population MFG (MP-MFG) can incorporate multiple types of populations exposed to different dynamics \cite{huang2006large,perolat2022scaling, subramanian2020multi,dayanikli2023multi, bensoussan2018mean, carmona2018probabilistic, huang2024modelbased}. 
However, within each population exact symmetry must hold and the number of types must be much smaller than the number of agents.
Moreover, MP-MFG can be lifted to an equivalent single-population MFG \cite{huang2024modelbased} under certain constraints.
Overall, all of these works require variations of the same stringent symmetry assumptions, restricting their applicability.
A detailed survey of learning in MFGs can be found at \cite{lauriere2024learning}.

\textbf{Graphon MFG.}
Graphon games, proposed initially by \cite{parise2019graphon}, can incorporate heterogeneity between MFG agents by assuming graphon-based interactions.
The setting has been analyzed in discrete-time \cite{cui2021learning, vasal2020master} as well as in the continuous time setting \cite{aurell2022finite, caines2019graphon, aurell2022stochastic}. 
Recently, policy mirror descent has been analyzed in this setting to produce convergence guarantees under monotonicity conditions \cite{zhang2023learning}.
However, these works on graphon mean-field games still incorporate exact symmetry in the form of the graphon: namely, the types of agents must follow a symmetric distribution and interactions must be through a symmetric graphon.
In fact, graphon MFGs can still be reduced to regular MFGs \cite{zhang2023learning}.

\textbf{Other related work.}
Another class of games where a large number of agents can be tackled tractably are the so-called potential games \cite{rosenthal1973class}, generalized to Markov potential games incorporating dynamics \cite{leonardos2021global}.
Approximate potentials have been studied in a similar spirit on Markov $\alpha$-potential games \cite{guo2023markov} and near potential games \cite{candogan2013near}.
However, to the best of our knowledge, approximate symmetry has not been studied in the literature of MFGs.

\subsection{Our Contributions}

We list the following as our contributions, compared to past work summarized in the previous section.
\begin{enumerate}[leftmargin=*]
    \item We first tackle the foundational but understudied question for MFGs: \emph{when can a given $N$-agent game be meaningfully extended to an infinite-player MFG?} 
    We construct a well-defined MFG approximation to an arbitrary (possibly non-symmetric) finite-player dynamical game (DG) using the idea of function symmetrization and via Kirszbraun Lipschitz extensions.
    \item Using our extension, we define a new class of \emph{$\alpha,\beta$-symmetric DGs} for which it is tractable to find approximate Nash.
    $\alpha,\beta$-symmetry generalizes permutation invariance in dynamic games to arbitrary MARL problems, where parameters $\alpha,\beta$ quantify degrees of heterogeneity in dynamics and player rewards respectively.
    \item We prove that the solution of the induced MFG is indeed an approximate Nash to the original $\alpha,\beta$-symmetric DG up to a bias of $\mathcal{O}(\sfrac{1}{\sqrt{N}} + \alpha + \beta)$, demonstrating that MFG approximation is robust to heterogeneity and finite-agent errors in the DG.
    \item We analyze TD learning on the trajectories of the finite-agent DG.
    We show that by only using $\mathcal{O}(\varepsilon^{-2})$ samples from the $N$-player dynamic game, policies can be approximately evaluated \emph{on the abstract MFG} up to symmetrization error.
    \item Finally, we show that under monotonicity conditions, policy mirror descent (PMD) combined with TD learning converges to an approximate Nash equilibrium using $\widetilde{\mathcal{O}}(\varepsilon^{-6})$ sample trajectories of the $N$-player DG.
    This provides an end-to-end learning guarantee for MARL under $\alpha,\beta$-symmetry, characterizing a novel class of problems that can be solved efficiently with MARL.
\end{enumerate}

%% file: results.tex
\section{Main Results}

\emph{Notation.}
For $K\in\mathbb{N}_{>0}$, let $[K] := \{ 1,\ldots, K\}$.
Let $\Delta_\setX$ be the probability simplex over $\setX$.
For any $N \in \mathbb{N}_{>0}$ define $\Delta_{\setX, N} := \{ \mu \in \Delta_\setX \, | \, N\mu(x) \in \mathbb{N}_{\geq 0} , \forall x\in\setX\}$.
For $\vecx \in \setX^N$, define the empirical distribution $\empc{\vecx} \in \Delta_{\setX, N}$ as $\empc{\vecx}(x') = \sfrac{1}{N}\sum_{i=1}^N \ind{x_i = x'}$.
Let $\Sym_K$ be the set of permutations over the set $[K]$, so $\Sym_K := \{g:[K] \rightarrow [K] \, | \, g \text{ bijective}\}$.
For $\vecx = (x_1, \ldots, x_K) \in \setX^K$ and $g \in \Sym_K$, define $g(\vecx) := (x_{g(1)}, \ldots, x_{g(K)}) \in \setX^K$.
Define $\vecx^{-i} \in \setX^{K-1}$ as the vector with $i$-th entry of $\veca$ removed, and $(x, \vecx^{-i}) \in \setX^K$ as the vector where $i$-th coordinate of $\vecx$ has been replaced by $x\in\setX$.

We consider discrete state-action sets $\setS,\setA$.
We denote the set of time-dependent policies on $\setS,\setA$ as $\Pi := \{ \pi:\setS\times[H] \rightarrow \Delta_\setA\}$.
We abbreviate $\pi_h(a|s):=\pi(s,h)(a)$.
For $p:\setS \rightarrow \Delta_\setA$ and $\rho \in \Delta_\setS$, we define $(\rho \cdot p) \in\Delta_\setS$ as $(\rho \cdot p)(s,a) := \rho(s) p(s)(a)$ for all $s,a\in\setS\times\setA$.
Finally, we define entropy $\entropy(u) := -\sum_a u(a) \log u(a)$ for $u\in\Delta_\setA$.
We denote $\kl(u|v) := \sum_{a} u(a)\log\frac{u(a)}{v(a)}$ for $u,v\in\Delta_\setA$.

\subsection{Finite-Horizon Dynamic Games}

Firstly, we define finite-horizon dynamic games, the main object of interest of this work.

\begin{definition}[N-player FH-DG]\label{def:dg}
    An $N$-player finite-horizon dynamic game (FH-DG) is a tuple $(\setS, \setA, \initpop, N, H, \{ P^i\}_{i=1}^N, \{R^i\}_{i=1}^N)$ where the state and actions sets $\setS, \setA$ are discrete, $\initpop\in\Delta_\setS$, the number of players $N\in\mathbb{N}_{>1}$, horizon length $H\in\mathbb{N}_{>0}$, and transition dynamics and rewards are functions such that 
    $P^i: \setS \times \setA \times (\setS \times \setA) ^ {N-1} \rightarrow \Delta_\setS$ and $R^i: \setS \times \setA \times (\setS \times \setA) ^ {N-1} \rightarrow [0,1]$.
\end{definition}

The above definition differs from Markov games \cite{shapley1953stochastic}, where a common state is shared by all agents.
In FH-DG, each agent only observes their own state and the dynamics depend on the state vector of all $N$ agents.
Such a model can be especially realistic in cases where games have natural \emph{locality}, that is, the game state is not globally available to agents.
Next, we define the Nash equilibrium.

\begin{definition}[FH-DG Nash equilibrium]\label{def:fhdg_ne}
    For a FH-DG $\setG = (\setS, \setA, \initpop, N, H, \{ P^i\}_{i=1}^N, \{R^i\}_{i=1}^N)$, policy tuple $\vecpi=(\pi^1,\ldots,\pi^N)\in\Pi^N$ the expected total reward of agent $i\in [N]$ is defined as 
    \begin{align*}
        \JfinNi{i} \left( \vecpi \right) & := \Exop \left[ \sum_{h=0}^{H-1} R^i(s_h^i, a_h^i, \vecrho_h^{-i}) \middle| \substack{\forall j : s_0^j \sim \initpop , \quad a_h^j \sim \pi_h^j(s_h^j)\\ s_{h+1}^j \sim P^j(\cdot|s_h^j, a_h^j, \vecrho_h^{-j})} \right]
    \end{align*}
    where $\vecrho_h := (s_h^i, a_h^i)_{i=1}^N$.
    The \emph{exploitability} of agent $i$ for policies $\vecpi$ is then defined as $\ExpfinNi{i}(\vecpi) := \max_{\pi\in\Pi} \JfinNi{i} \left(\pi, \vecpi^{-i} \right) - \JfinNi{i} \left( \vecpi \right)$.
    If $\max_i \ExpfinNi{i}(\vecpi) = 0$, $\vecpi$ is called a Nash equilibrium (NE) of the FH-DG.
    If $\max_i \ExpfinNi{i}(\vecpi) \leq \delta$, $\vecpi$ is called a $\delta$-Nash equilibrium ($\delta$-NE) of the FH-DG.
\end{definition}

At a $\delta$-NE, the incentive for any selfish agent to deviate is small: hence, approximate NE is a natural solution concept for FH-DG.
However, the problem of finding $\delta$-NE is challenging: not only is it computationally intractable in general \cite{daskalakis2009complexity}, but for large $N$ the search space of policies grows exponentially.
This motivates the approximation via symmetrization in the remainder of the work. 

\subsection{Symmetrization and Lipschitz Extension}

In order to define approximate symmetry, we first show that finite-agent dynamics of Definition~\ref{def:dg} can be extended to infinitely many players.
In the process, we tackle a question that is relevant for MFGs beyond our work:\emph{When and how can we build an MFG model on the continuum, given dynamics on finitely many players?}
We will use the notions of symmetrization and Lipschitz extension.

\begin{definition}[Symmetric function, symmetrization]
A function $f: \setX^K \rightarrow \setY$ is called \emph{symmetric} if $f(g(\vecx)) = f(\vecx)$, $\forall\vecx \in \setX^K$, $g\in\Sym_{K}$.
For a symmetric $f: \setX^K \rightarrow \setY$, we define its \emph{population lifted version} $\lift{f}: \Delta_{\setX, K} \rightarrow \setY$ as the well-defined function such that $\widebar{f}(\mu) = f(\vecx)$ for $\forall \vecx \in \setX^K$ satisfying $\empc{\vecx} = \mu.$
Given $f: \setX^K \rightarrow \mathbb{R}^D$, we define the \emph{symmetrization} $\symmetrization{f}: \setX^K \rightarrow \setY$ as
\begin{align*}
    \symmetrization{f}(\vecx) = \frac{1}{K!} \sum_{g \in \Sym_K} f(g(\veca)), \quad \forall \veca\in\setX^K.
\end{align*}
We also denote $\symmbar{f} := \widebar{\symmetrization{f}}$.
\end{definition}

We note that the terminology ``symmetrization'' is consistent as $\symmetrization{f}$ is indeed a symmetric function (as verified in Section~\ref{sec:appendix_prelim}).
Furthermore, if $f$ is symmetric then $\symmetrization{f} = f$ as expected.

Finally, to extend DG to the infinite-player continuum, we will use the following special case of the Kirszbraun-Valentine theorem, which concerns Lipschitz extensions of functions from strict subsets of the Euclidean space to the entirety of the space preserving their Lipschitz modulus.

\begin{lemma}[Kirszbraun-Valentine \cite{kirszbraun1934zusammenziehende, valentine1945lipschitz}]\label{lemma:kirszbraun}
    Let $d_1, d_2 \in \mathbb{N}_{> 0}$, and $U\subset\mathbb{R}^{d_1}$.
    Let $f: U \rightarrow \mathbb{R}^{d_2}$ be an $L$-Lipschitz function with respect to the Euclidean norm $\|\cdot\|_2$.
    Then, there exists $\lipext{f}: \mathbb{R}^{d_1} \rightarrow \mathbb{R}^{d_2}$ such that $\lipext{f}$ is $L$-Lipschitz and $\lipext{f}(x) = f(x)$ for all $x \in U$.
\end{lemma}

While $\lipext{f}$ is not unique in general, it admits various explicit formulations \cite{mcshane1934extension,sukharev1978optimal}, and the particular formulation is not important in this work.

\subsection{Mean-field Games and \texorpdfstring{$\alpha, \beta$}{a,b}-Symmetric Games}

Next, using the definitions from the previous section, we show how the FH-DG can be extended to an MFG.
We formalize the finite-horizon MFG (FH-MFG), which will be the main approximation tool.

\begin{definition}[Finite-horizon mean-field game]\label{def:mfg}
    A finite-horizon mean-field game (FH-MFG) is a tuple $(\setS, \setA, \initpop, H, P, R)$ where $\setS, \setA$ are discrete, $\initpop \in \Delta_\setS$, $H\in\mathbb{N}_{>0}$, the transition dynamics $P$ is a function $P : \setS\times\setA\times \Delta_{\setS\times\setA} \rightarrow \Delta_\setS$, and 
 the reward $R$ is a function $R : \setS\times\setA\times \Delta_{\setS\times\setA} \rightarrow [0,1]$.
\end{definition}

Compared to Definition~\ref{def:dg}, Definition~\ref{def:mfg} introduces two conceptual changes under the premise of exact symmetry: (1) the dependency of dynamics to the states and actions of other agents have been reduced to a dependency on a population distribution on $\Delta_{\setS\times\setA}$, and (2) $N$ agents have been implicitly replaced by a single representative agent.
We next extend the definition NE to MFGs.

\begin{definition}[Induced population, MFG-NE]\label{def:mfg_ne}
    For a FH-MFG defined by the tuple $(\setS, \setA, \initpop, H, P, R)$, we define the population update operators $\gpop, \lpop$ as
\begin{align}
    \gpop(\mu, \pi)(s',a') &:= \sum_{s \in \setS, a \in \setA} \mu(s,a) P(s'|s,a,\mu) \pi(a'|s') \\
    \lpop(\pi) &:= \big\{\Gamma(\cdots \Gamma(\Gamma(\initpop \cdot \pi_0, \pi_1) \cdots, \pi_{h-1})) \big\}_{h=0}^{H-1}.
\end{align}
For $\pi  \in \Pi$ and $\vecmu = \{ \mu_h \}_{h=0}^{H-1} \in \Delta_{\setS\times\setA}^{H}$, the expected reward is defined as
    \begin{align}
        \Vfin \left( \vecmu, \pi \right) & := \Exop \left[ \sum_{h=0}^{H-1} R(s_h, a_h, \mu_h) \middle| \substack{s_0 \sim \initpop , \quad a_h \sim \pi_h(s_h)\\ s_{h+1} \sim P(s_h, a_h, \mu_h)} \right].
    \end{align}
    We define MFG exploitability as $\Expfin (\pi) := \max_{\pi' \in \Pi} \Vfin( \lpop (\pi), \pi') - \Vfin ( \lpop ( \pi), \pi )$ and FH-MFG-NE as:
    \begin{align}
        \textit{Policy } &\pi^* = \{ \pi^*_h\}_{h=0}^{H-1} \in \Pi \text{ such that } \quad \Expfin (\pi^*) = 0. \tag{MFG-NE}
    \end{align}
\end{definition}

Intuitively, the above definition of MFG-NE requires that the policy $\pi$ is optimal against the population flow it induces.
Questions of existence \cite{cardaliaguet2010notes,bensoussan2013mean, huang2023statistical} and approximation of the FH-DG under exact symmetry \cite{saldi2018markov} have been thoroughly studied in the literature.
That is, if an $N$-player game exhibits exact symmetry, then the MFG-NE exists and is an approximate NE of the FH-DG.

Taking a constructive approach, we show that the FH-MFG-NE of an \emph{appropriately constructed} MFG is also an approximate NE for a given FH-DG without a prior model.
The definition below of an ``induced MFG'' demonstrates how arbitrary non-symmetric dynamics can be extended to an MFG.

\begin{definition}[Induced FH-MFG]\label{def:inducedmfg}
    Let $\setG = (\setS, \setA, \initpop, N, H, \{ P^i\}_{i=1}^N, \{R^i\}_{i=1}^N)$ be a FH-DG.
    The MFG induced by $\setG$, denoted $\inducedmfg{\setG}$, is defined to be the $(\setS, \setA, \initpop, H, P, R)$, where $P: \setS\times\setA\times\Delta_{\setS\times\setA} \rightarrow \Delta_\setS$ and $R: \setS\times\setA\times\Delta_{\setS\times\setA} \rightarrow [0,1]$ are defined for all $s\in\setS,a\in\setA,\mu\in\Delta_{\setS\times\setA}$ as:
    {
        \begin{align*}
            P(s,a, \mu) :=  \sum_{i=1}^N \frac{\lipext{\symmbar{P^i(s,a, \cdot)}} (\mu)}{N}, \quad
            R(s,a,\mu) := \sum_{i=1}^N \frac{\lipext{\symmbar{R^i(s,a, \cdot)}} (\mu)}{N}.
        \end{align*}
    }
\end{definition}

$\inducedmfg{\setG}$ is well-defined due to Lemma~\ref{lemma:kirszbraun}.
In words, the definition of $\inducedmfg{\setG}$ consists of two main operations: (1) symmetrize ($\symmbar{\cdot}$) and extend ($\lipext{\cdot}$) $P^i, R^i$ to $\Delta_{\setS\times\setA}$,
and (2) average symmetrized dynamics and rewards for all players.
Furthermore, in the special case $P^i=P^j$ and $R^i=R^j$ for all $i\neq j$ and $P^i(s,a,\cdot), R^i(s,a,\cdot)$ are symmetric, the $\inducedmfg{\setG}$ has dynamics and rewards $\lipext{\widebar{P}^1}, \lipext{\widebar{R}^1}$, which are simply the Lipschitz extensions of $P^1, R^1$ to the continuum.

\begin{remark}
    Even for exact symmetric games, Definition~\ref{def:inducedmfg} is relevant.
    The availability of an MFG model is typically taken for granted, however, 
    since real-world algorithms might only be able to access finite-agent dynamics without a known MFG model, it is a valid research question when and how a game can be meaningfully extended to infinite players (answered by Definition~\ref{def:inducedmfg}).
\end{remark}

Finally, we provide the definition of approximate or $\alpha,\beta$-symmetry.
\begin{definition}[$\alpha,\beta$-Symmetric DG]\label{def:alphabeta}
    Let $\setG =  (\setS, \setA, \initpop, N, H, \{ P^i\}_{i=1}^N, \{R^i\}_{i=1}^N)$ be an $N$-player FH-DG, inducing $\inducedmfg{\setG} = (\setS, \setA, \initpop, H, P, R)$.
    If it holds for $\alpha,\beta \geq 0$ that
    \begin{align*}
        \max_{ \substack{i\in[N], \, s,a\in\setS\times\setA \\ \mu\in\Delta_{\setS\times\setA}}} \, \max_{ \substack{\vecrho \in (\setS\times\setA)^{N-1} \\ \sigma(\vecrho) = \mu} } \| P^i(s,a,\vecrho) - P(s,a,\mu) \|_1 \leq \alpha, \\
        \max_{ \substack{i\in[N], \, s,a\in\setS\times\setA \\ \mu\in\Delta_{\setS\times\setA}}} \, \max_{ \substack{\vecrho \in (\setS\times\setA)^{N-1} \\ \sigma(\vecrho) = \mu} } | R^i(s,a,\vecrho) - R(s,a,\mu) | \leq \beta,
    \end{align*}
    then the FH-DG $\setG$ is called $\alpha,\beta$-symmetric.
\end{definition}
As expected, an exactly symmetric $N$-player game is also $0,0$-symmetric, and any dynamic game $\setG$ is $\alpha,\beta$-symmetric for some constants $\alpha\leq 2, \beta \leq 1$.
Hence, Definition~\ref{def:alphabeta} generalizes exact permutation invariance.
Games that exhibit near-exact symmetries will have very small constants $\alpha,\beta$, we will next provide approximation and learning guarantees for such finite-agent games.

\subsection{Approximation of NE under Approximate Symmetry}

In this section, we will prove that a NE of the induced $\inducedmfg{\setG}$ is also an approximate NE of the finite-agent game $\setG$.
We will provide an explicit bound on the approximation, motivating the use of MFGs for solving FH-DG.

We first introduce the notion of $\kappa$-sparse dynamics.
In words, with $\kappa$-sparse dynamics an agent at state $s$ playing action $a$ is impacted only by other agents occupying a sparse set of ``neighboring'' state-actions $\setN_{s,a} \subset \setS\times\setA$ where $|\setN_{s,a}| \leq \kappa$.
For a subset $\setU \subset \setX$, we define the function $p_\setU : \setX \rightarrow \setX \cup \{\perp \}$ as $p_\setU (x) = x$ if $x\in\setU$ and $p_\setU(x) = \perp$ otherwise, where $\perp$ is treated as a placeholder element such that $\perp \notin \setU$. 

\begin{definition}[$\kappa$-sparse dynamics/rewards]
    A function $f: \setX^{M} \rightarrow \setY$ is called \emph{$\kappa$-sparse} (on some $\setU \subset \setX$) if $|\setU| \leq \kappa$ and $f(\vecx) = f(\vecy)$ whenever $p_{\setU}(x_i) = p_\setU(y_i)$ for all $i=1,\ldots,M$.
    Dynamics $\{ P^i \}_{i=1}^N$ (resp. rewards $\{R^i\}_{i=1}^N$) are called \emph{$\kappa$-sparse} if all $P^i(s,a,\cdot)$ (resp. $R^i(s,a,\cdot)$) are $\kappa$-sparse on some $\setU_{s,a} \subset \setS\times\setA$ for all $s\in\setS, a\in\setA$ (resp. $\setU_{s,a}\subset \setS\times\setA$ for all $s\in\setS, a\in\setA$).
\end{definition}

Sparsity is common in FH-DG, particularly when there is spatial structure.
Many standard MFG problems such as the beach-bar problem \cite{perrin2020fictitious} and crowd modeling \cite{zaman2023oracle} are in fact $(\kappa=1)$-sparse, as agents are only affected by the population distribution at their current state.

Using sparsity, we provide an upper bound of the Lipschitz constants of maps $P(s,a,\cdot), R(s,a,\cdot)$ of the induced MFG on the continuum $\Delta_{\setS\times\setA}$, demonstrating that unless the FH-DG exhibits dominant players, $P,R$ have bounded Lipschitz moduli independent of $N$.

\begin{lemma}[Lipschitz extension bound] \label{lemma:extension_bound}
    Let $\setG$ be an FH-DG with dynamics and rewards $\{ P^i\}_{i=1}^N, \{R^i\}_{i=1}^N$ admitting the induced mean-field game $\inducedmfg{\setG}$ with dynamics and rewards $P,R$.
    Assume that $\{ P^i\}_{i=1}^N, \{R^i\}_{i=1}^N$ are $\kappa$-sparse and it holds that
    \begin{align*}
        \| P^i(s,a, \vecrho) - P^i(s,a, ((s', a'), \vecrho^{-j}))\|_1 \leq C_1, \quad
        | R^i(s,a, \vecrho) - R^i(s,a, ((s', a'), \vecrho^{-j}))| \leq C_2,
    \end{align*}
    for any $i,j \in[N], i\neq j$, $s,s'\in\setS, a, a'\in\setA$ and $\vecrho\in(\setS\times\setA)^{N-1}$ for some constants $C_1, C_2$.
    Then, the induced $P, R$ have Lipschitz modulus at most 
    $ C_1 N \kappa$ and $ C_2 N \sqrt{\kappa}$ respectively,
    that is, 
    \begin{align*}
        \| P(s,a,\mu) - P(s,a,\mu')\|_2 \leq C_1 N \kappa \| \mu - \mu' \|_2, \quad
        | R(s,a,\mu) - R(s,a,\mu')| \leq C_2 N \sqrt{ \kappa } \| \mu - \mu' \|_2,
    \end{align*}
    for any $s\in\setS,a\in\setA,\mu,\mu'\in\Delta_{\setS\times\setA}$.
\end{lemma}

The above theorem characterizes a condition on the original FH-DG for the induced FH-MFG to have smooth (Lipschitz) dynamics.
The theorem suggests that the game must have \emph{no dominant players} so that the effect of each agent on others is upper bounded of order $\mathcal{O}(\sfrac{1}{N})$.
Furthermore, by standard results in MFG literature, if the ``no dominant players'' condition of Lemma~\ref{lemma:extension_bound} holds, the population update $\Gamma$ is also Lipschitz continuous with some modulus $\lpopmu$ that is independent of $N$.

Finally, we state the main approximation result, which quantifies how closely the true $N$-player game Nash equilibrium can be approximated by the mean-field Nash equilibrium of the symmetrized game.

\begin{theorem}\label{theorem:main_approximation}
    Let $\setG = (\setS, \setA, \initpop, N, H, \{ P^i\}_{i=1}^N, \{R^i\}_{i=1}^N)$ be an $N$-player FH-DG and $\inducedmfg{\setG} = (\setS, \setA, \initpop, H, P, R)$.
    Let the Lipschitz modulus of the population update operator $\gpop$ in $\mu$ be $\lpopmu$.
    If $\pi^* \in \Pi$ is a MFG-NE of $\inducedmfg{\setG}$, then $(\pi^*, \ldots, \pi^*) \in \Pi^N$ is an $\epsilon$-NE of the FH-DG, where
    \begin{align*}
        \epsilon = \mathcal{O}\left( \frac{H^2 (1-\lpopmu^H)}{(1-\lpopmu)\sqrt{N}} + \alpha H^2 \frac{1-\lpopmu^H}{1-\lpopmu} + \beta H \right).
    \end{align*}
\end{theorem}
\begin{proof}(sketch)
     We show that (1) the empirical distribution of agent state-actions over $\setS\times\setA$ approximates the induced mean-field  $\lpop(\pi^*)$, (2) marginal distributions of states of an agent $\Prob[s_h^i = \cdot]$ in FH-DG are also approximated by the mean-field, and (3) explicitly bounding the difference between $\Vfin$ and $\JfinNi{i}$.
    The formal proof and explicit upper bound are presented in Section~\ref{sec:proof_main_approx}.
\end{proof}

Most critically, the approximation bound proves that the MFG approximation is robust to small heterogeneity: when $\alpha,\beta$ are small, the induced MFG-NE approximates the true NE well.
Furthermore, the upper bound suggests three different asymptotic regimes depending on $\gpop$ being non-expansive, contractive, or expansive.
If $\lpopmu \leq 1$, the bound above is polynomial.
If $\lpopmu > 1$, $\alpha > 0$ might incur an exponential dependency on $H$, whereas the error due to $\beta > 0$ only scales linearly with $\cO(\beta H)$.
However, the exponential worst-case dependence of the bias on $H$ is generally unavoidable even under perfect symmetry, as matching lower bounds are known \cite{yardim2024meanfield}.
Theorem~\ref{theorem:main_approximation} also recovers the bounds known for exactly symmetric FH-DG (i.e. $\alpha=\beta=0$, see \cite{yardim2024meanfield}).

Finally, we emphasize that Theorem~\ref{theorem:main_approximation} does not assume any particular structure on the FH-DG: the results apply for any values of $\alpha,\beta$, although the quality of approximation will vary.
Furthermore, it is known that for $N>2$, finding an $\epsilon$-NE for the FH-DG is \textsc{PPAD}-complete even for a certain \emph{absolute constant} $\epsilon$ \cite{goldberg2011survey}.
Hence, even when $\alpha,\beta$ are not close to $0$, the result will be useful in approaching the \textsc{PPAD}-complete limit via mean-field approximation.

The results so far already suggest a learning algorithm: 
one can estimate (e.g. via neural networks) the induced $P,R$ and solve the MFG directly with standard algorithms (e.g. \cite{perrin2020fictitious, lauriere2022scalable}).
However, this method can be prohibitively expensive as it involves learning functions to and from $\Delta_{\setS\times\setA}$.
The remainder of the paper will be dedicated to formulating more efficient methods.

\subsection{Policy Evaluation with \texorpdfstring{$\alpha, \beta$}{a,b}-Symmetry}

In this section, we analyze TD learning for $\alpha,\beta$-symmetric FH-DG.
While Definition~\ref{def:inducedmfg} provides an explicit construction of an MFG, we show that this construction is not needed for policy evaluation.
Namely, using TD learning, a policy $\pi$ can be evaluated with respect to the (induced) mean-field $\lpop(\pi)$ only through sampling trajectories of the FH-DG $\setG$.
We first define Q functions on the MFG.

\begin{definition}[Mean-field Q values]\label{def:qvalue}
    For the MFG $(\setS, \setA, \initpop, H, P, R)$, for $\tau\geq 0$, $h = 0,\ldots,H-1$, we define (entropy regularized) Q-values for each $h = 0,\ldots,H-1$ and $s\in\setS, a\in\setA$ as
    \begin{align*}
        \Qpi{h} (s,a ) := \Exop \left[ \sum_{h'=h}^{H-1} R(s_{h'}, a_{h'}, \mu_{h'}) + \tau \entropy( \pi_{h'}(\cdot|s_{h'})) \middle| \substack{s_h = s, \, a_h = a , \, s_{h'+1} \sim P(s_{h'}, a_{h'}, \mu_{h'}), \\   a_{h'} \sim \pi_{h'+1}(s_{h'+1}), \, \mu_{h'} := \lpop(\pi)_{h'}, \forall h' \geq h} \right].
    \end{align*}
\end{definition}

In other words, the Q-values of a policy $\pi$ are computed with respect to the MDP induced by the population distributions $\lpop(\pi)$ in the MFG.
We note that the above definition does not match the typical definition of Q-values in a multi-agent setting, and rather is defined concerning an abstract MFG.
We note that we will occasionally treat $\Qpi{h}$ as an element of the vector space $\mathbb{R}^{\setS\times\setA}$.

For the finite-horizon problem, we will analyze TD learning, which is a standard method for policy evaluation with established guarantees beyond MFGs \cite{tsitsiklis1996analysis}.
We formulate Algorithm~\ref{alg:td}, presented for simplicity as performing TD learning on agent $1$.

\begin{algorithm}
\caption{TD Learning for $\alpha,\beta$-symmetric games}\label{alg:td}
\begin{algorithmic}[1]
\Require FH-DG $\setG$, epochs $M$, learning rates $\{\lr{m}{h}\}_{m}$, policy $\pi \in \Pi$, entropy regularization $\tau \geq 0$.
\State $\Qest{0}{h}(s,a) \leftarrow 0, \quad \forall h \in \{ 0, \ldots H-1\}, s\in\setS,a\in\setA$
\For{$m \in 0, 1, \ldots, M-1$} 
    \State Using $\pi$ for all agents, sample path from $\setG$: $\{\vecrho_{m, h}, \vecr_{m, h}\}_{h=0}^{H-1} := \{s^i_{m,h}, a^i_{m,h}, r^i_{m, h}\}_{i,h}$.
    \State Perform TD update: 
        \begin{align*}
            \Qest{m+1}{h} &\leftarrow \Qest{m}{h} + \lr{m}{h}\big(\Qest{m}{h+1}(s_{m,h+1}^1, a_{m,h+1}^1) + r_{m,h}^1 + \tau\entropy( \pi_h(\cdot|s_{m,h}^1)) \\
            &\quad \quad - \Qest{m}{h}(s_{m,h}^1, a_{m,h}^1)\big) \vece_{s_{m,h}^1, a_{m,h}^1}, \, \forall h < H-1
           \\
            \Qest{m+1}{H-1} &\leftarrow \Qest{m}{H-1} + \lr{m}{h} (\tau\entropy( \pi_{H-1}(\cdot|s_{m,H-1}^1)) + r_{m,H-1}^1 - \Qest{m}{H-1}(s_{m,h}^1, a_{m,h}^1)) \vece_{s_{m,H-1}^1, a_{m,H-1}^1}
        \end{align*}
\EndFor
\State Return $\{\Qest{M}{h}\}_{h=0}^{H-1}$.
\end{algorithmic}
\end{algorithm}

\begin{theorem}[TD learning for $\alpha,\beta$-Symmetric Games]\label{theorem:td}
Let $\setG$ be an $N$-player FH-DG and $\inducedmfg{\setG}$ be its induced MFG.
Let $\pi \in \Pi$ be a policy such that $\lpop(\pi) = \vecmu = \{\mu_h \}_h$ and $\delta := \inf_{ \substack{ h,s,a \text{ s.t. } \Prob[s^1_{h} = s, a^1_{h} = s] > 0 }} \Prob[s^1_{h} = s, a^1_{h} = s].$
Assume Algorithm~\ref{alg:td} is run with $\pi$ for $M > 0$ epochs for with learning rates $\lr{m}{h} := \frac{2\delta^{-1}}{m + 2\delta^{-1}}$.
Then, the output $\{\Qest{M}{h}\}_h$ of Algorithm~\ref{alg:td} satisfies $\Exop\left[ \sum_{h=0}^{H-1} \| \Qest{M}{h} - \Qpi{h} \|^2_{\mu_h} \right] \leq \cO \left( \frac{1}{M} + \frac{1}{N} + \alpha^2 + \beta^2 \right),$
where $\|\cdot\|_{p}$ is defined for $p\in\Delta_{\setS\times\setA}$ as $\| q \|_p := \sqrt{\sum_{s,a} p(s,a) q(s,a)^2}$.
\end{theorem}

Theorem~\ref{theorem:td} provides a finite-sample guarantee for TD learning, a building block of many MARL and MFG algorithms.
Most importantly, it suggests that in order to use mean-field game theory to approximate NE of an FH-DG $\setG$, there is no need to explicitly build a model of $\inducedmfg{\setG}$.
Instead, TD learning in the original $N$-player game when all the agents pursue policy $\pi$ allows the evaluation of the mean-field $Q$-values of $\pi$ efficiently.
The rate of convergence suggested by Theorem~\ref{theorem:td} matches the optimal known rates for TD-learning in a single-agent setting \cite{kotsalis2022simpleII}.
In practice, one can use the trajectories of all $N$ agents to further improve efficiency, instead of only using that of agent $i=1$.

\subsection{Learning NE under \texorpdfstring{$\alpha, \beta$}{a,b}-Symmetry}

We complete our framework by providing our key theoretical result: any $\alpha,\beta$-symmetric DG can be solved approximately only using samples from the $N$-player DG, under monotonicity assumptions.
Our algorithm uses TD learning as a building block, with stochastic policy evaluations used for policy mirror descent updates \cite{lan2023policy, yardim2023policy, zhang2023learning}.

\begin{definition}[Monotone MFG \cite{perrin2020fictitious, perolat2022scaling}]\label{def:monotone_mfg}
    A MFG with dynamics $P$ and rewards $R$ is called \emph{monotone} if $P$ is independent of $\mu$, and for all $\mu,\mu'$ it holds that $\sum_{s,a} (R(s,a,\mu) - R(s,a,\mu'))(\mu(s,a) - \mu'(s,a)) < 0$.
    A DG $\setG$ is called \emph{monotone extendable} if $\inducedmfg{\setG}$ is monotone.
\end{definition}

To motivate this definition, we provide a large class of DGs that are relevant and monotone-extendable.

\begin{example}[Asymmetric dynamic congestion games]
    For any $i\in[N]$, let $h_i:\setS\times\setA\times [N] \rightarrow [0,1], r^i:\setS\times\setA \rightarrow [0,1] $ be arbitrary functions so that $h_i(s,a,\cdot)$ is non-increasing for any $s,a$.
    Assume $P^i(\cdot|s,a,\vecrho^{-i})$ does not depend on $\vecrho^{-i}$ for any $s,a$, and $R^i(s,a,\vecrho^{-i})$ be $1$-sparse so that
    $R^i(s,a,\vecrho^{-i}) = h_i(s,a,\sum_{j=1}^N \ind{\rho_j = (s,a)}) + r_i(s,a)$.
    Such games can be seen as generalizations of congestion games \cite{rosenthal1973class} and congestion games with player-specific incentives \cite{milchtaich1996congestion}, for which an efficient solution is unknown.
    We prove monotone extendability and characterize the values of $\alpha,\beta$ and Lipschitz constants for such games in Section~\ref{sec:appendix_monotone_congestion}.
\end{example}

\begin{algorithm}
\caption{Policy mirror descent for $\alpha,\beta$-symmetric games (Symm-PMD)}\label{alg:pmd}
\begin{algorithmic}[1]
\Require FH-DG $\setG$, epochs $T$, TD learning epochs $M$, learning rates $\{\plr{t}\}_{t}$, entropy $\tau$.
\State Initialize uniform policy: $\pi_{0,h}(a|s) = \sfrac{1}{|\setA|}, \quad \forall h \in \{ 0, \ldots H-1\}, s\in\setS,a\in\setA$
\For{$t \in 0, 1, \ldots, T-1$} \Comment{\emph{Run for $T$ epochs}}
    \State Run Algorithm~\ref{alg:td} for policy $\pi_{t}$, $M$ epochs, entropy $\tau$, $\{\lr{m}{h}\}_{m}$ as in Theorem~\ref{theorem:td}
    \State Obtain $\{\Qest{t}{h}\}_{h=0}^{H-1}$, \quad set $\qest{t}{h}(s,a) := \Qest{t}{h}(s,a) - \tau\entropy(\pi_{t,h}(\cdot|s))$.
    \State Perform PMD update: ($\forall s\in\setS, h = 0,\ldots,H-1$)
        \begin{align*}
            \widehat{\pi}_{t+1,h}(\cdot|s) := \argmax_{u\in\Delta_\setA} \frac{\plr{t}}{1 - \tau \plr{t}} \Bigg[ \qest{t}{h}(s, \cdot)^\top u + \tau\setH(u) \Bigg] - \kl(u|\pi_{t,h}(\cdot|s)).
        \end{align*}
    \State Update policies: $\pi_{t+1,h}(\cdot|s) := \left(1 - \frac{1}{t+1}\right)\widehat{\pi}_{t+1,h}(\cdot|s) + \frac{1}{t+1}\operatorname{Unif(\cdot)}, \quad \forall s\in\setS$.
\EndFor
\State Return $\bar{\pi} := \{\frac{1}{T+1}\sum_{t=0}^T \pi_{t,h}\}_{h=0}^{H-1}$.
\end{algorithmic}
\end{algorithm}

\begin{theorem}[Convergence of PMD]\label{theorem:pmd}
    Let $\setG$ be a monotone extendable $\alpha,\beta$-symmetric game.
    Assume Symm-PMD (Algorithm~\ref{alg:pmd}) runs with learning rates $\plr{t} = \frac{1}{\sqrt{t+1}}$, entropy regularization $\tau\in(0,\sfrac{1}{2})$, with $M \geq \setO(\varepsilon^{-2})$ TD iterations for $T\geq\widetilde{\setO}(\varepsilon^{-4})$ epochs.
    Then, the output policy $\bar{\pi}$ is a $\setO(\varepsilon\tau^{-1} + \alpha \tau^{-1} + \beta\tau^{-1} + \sfrac{\tau^{-1}}{\sqrt{N}} + \tau)$-Nash equilibrium of $\setG$ in expectation.
\end{theorem}
\begin{proof}
    The proof is based on \cite{zhang2023learning} with the added complications of finitely many agents, approximate symmetry, and stochastic TD learning.
    Full proof is presented in Section~\ref{sec:proof_pmd}.
\end{proof}

Theorem~\ref{theorem:pmd} suggests a sample complexity of $\widetilde{\setO}(\varepsilon^{-6})$ trajectories from the $N$-agent FH-DG in order to compute a $\varepsilon$-NE (up to symmetrization bias).
In fact, it is (to the best of our knowledge) the first finite-sample guarantee for computing approximate NE for a large class of dynamic games with many agents.
Most importantly, the number of agents $N$ does not appear in the complexity: hence, the curse of many agents can be provably circumvented for $\alpha,\beta$-symmetric games.
Even in the exactly symmetric case ($\alpha=\beta=0$), Theorem~\ref{theorem:pmd} is the first guarantee to the best of our knowledge for learning FH-MFG-NE only observing trajectories of the $N$-agent game.

%% file: experiments.tex
\section{Experimental Results}

We support our theory by deploying Symm-PMD (Algorithm~\ref{alg:pmd}) on several large-scale $\alpha,\beta$-symmetric games.
For evaluations, we modify the well-known benchmarks from MFG literature (see \cite{cui2021approximately}) to propose three games with asymmetric incentives: A-RPS, A-SIS, and A-Taxi.
\textbf{A-RPS} is an adaptation of RPS \cite{cui2021approximately} to incorporate asymmetric rewards for agents.
\textbf{A-SIS} models disease propagation in a large population individually choosing to self-isolate or go out, incorporating asymmetric agents with individual susceptibility/healing rates and unique aversions to isolation.
Finally, \textbf{A-Taxi} simulates a large population of taxis serving clients in a grid, with individual preferences for regions and crowd aversion.
In our experiments, we use $N=1000$ and $N=2000$ agents demonstrating the ability of our framework to handle large MARL games.
A-Taxi incorporates $|\setS| > 2^{30}, H=128$, hence necessitates neural parameterization.
Our setup is thoroughly described in Section~\ref{sec:experiment_details}.

We deploy Symm-PMD on two different DGs, with $\alpha=0, \beta\approx 0.1, N=2000, H=10$ on A-RPS and $\alpha\approx\beta\approx 0.1, N=1000, H=20$ for A-SIS.
We compare the symmetrized approach of Symm-PMD to its asymmetric counterpart independent PMD (IPMD), where a separate policy is learned for each agent.
The training curves, pictured in Figures~\ref{figure:main}-(b,c) characterize the exploitabilty of the learned policies throughout training.
In both cases, while IPMD has no approximation bias in principle, it struggles to converge presumably suffering from the curse-of-many-agents.
Symm-PMD, on the other hand, rapidly converges to a policy profile with low exploitability and is much more sample-efficient.
In both cases, Symm-PMD converges to a solution with low bias.

We demonstrate the scalability of our approach with neural policies.
In the A-Taxi environment, we use PPO \cite{schulman2017proximal} with symmetrized neural policies and compare to the settings the policy has access to agent identities (either one-hot encoded, in OH-NN, or as an integer, in ID-NN).
Symmetrized policies outperform either benchmark by converging faster and to a better solution.
Learning independent neural policies for each of \num{1000} agents (Ind-NN) is extremely expensive in this setting: this approach performs the worst and is orders of magnitude computationally more expensive.

\textbf{Computational efficiency.} 
We also emphasize the computational efficiency of symmetrization: since our algorithm need not learn separate policies for each agent, it is drastically more computationally efficient compared to independent PMD.
In A-SIS and A-RPS benchmarks, learning is ~\num{60}\% faster, whereas symmetrized neural PPO in A-Taxi is >95\% faster than its independent counterpart.

\begin{figure}[h]
\begin{tabular}{ccc}
  \includegraphics[width=0.305\linewidth]{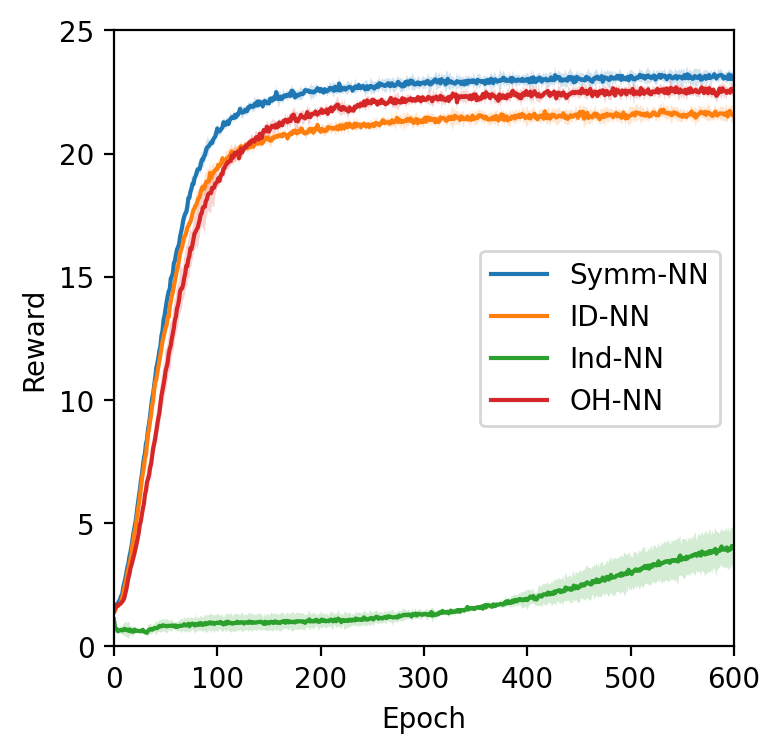} & 
  \includegraphics[width=0.31\linewidth]{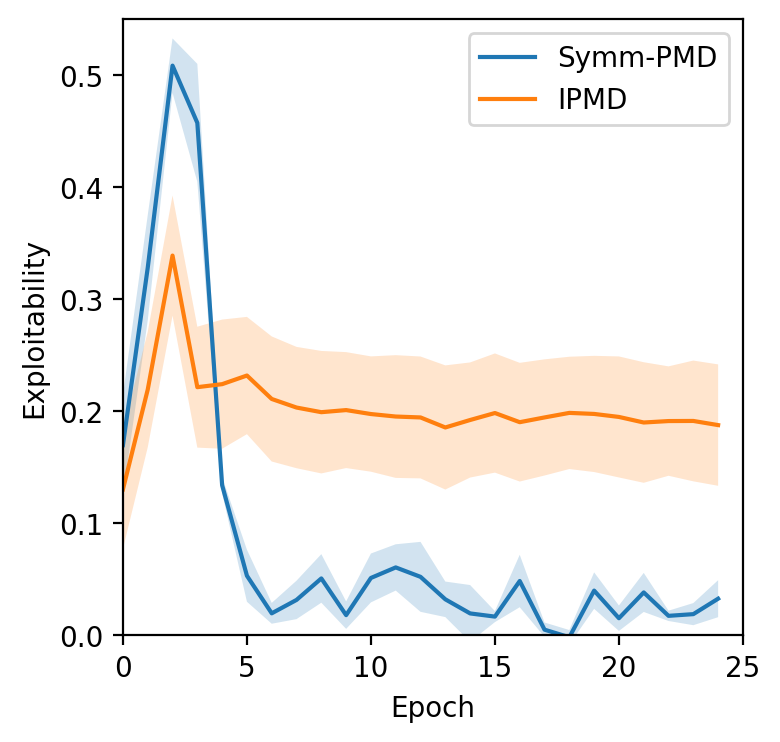} &   
  \includegraphics[width=0.31\linewidth]{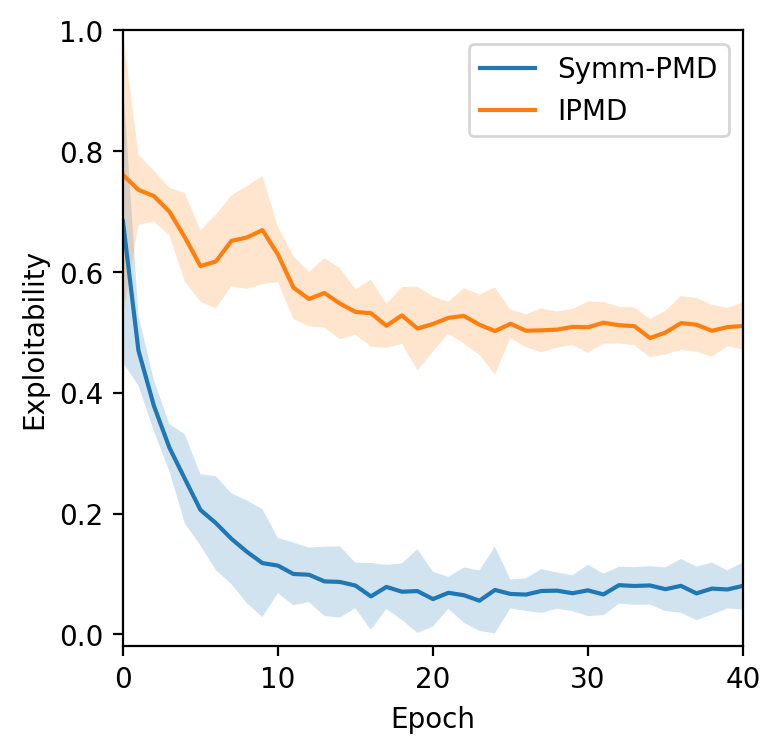} \\
(a) & (b) & (c)
\end{tabular}
\caption{
(a) The mean rewards throughout training of symmetric policies (Sym-NN), policies with onehot encoding for $i$ (OH-NN), policies with numerical encoding for $i$ (Ind-NN) and independent policies (Ind-NN) in A-Taxi.
(b, c) The exploitability throughout multiple epochs of Symm-PMD (Algorithm~\ref{alg:pmd}) and IPMD, for A-RPS with $\beta=0.1$ in (b) and A-SIS with $\alpha=\beta=0.1$ in (c).
}
\label{figure:main}
\end{figure}

%% file: discussion.tex
\section{Discussion and Conclusion}

We formulated a new class of competitive MARL problems ($\alpha,\beta$-symmetric games) that can be tractably solved.
We constructively showed that every $\alpha,\beta$-symmetric FH-DG can be efficiently approximated by an induced MFG.
We provided theoretical guarantees for TD learning, and under monotonicity, for PMD to approximate NE up to symmetrization bias.
These results provide a complete theory of learning under approximate symmetry, supported by numerical experiments.

%% file: appendix/preliminaries.tex
\section{Preliminaries}\label{sec:appendix_prelim}

Firstly, we present several basic facts regarding symmetrization and symmetric functions.

\begin{lemma}
    For any $f:\setX^K\rightarrow\setY$, $\symmetrization{f}$ is a symmetric function.
\end{lemma}
\begin{proof}
    For any $g'\in\Sym_K$, we have
\begin{align*}
    \symmetrization{f}(g'(\vecx)) = &\frac{1}{K!} \sum_{g \in \Sym_K} f(g(g'(\vecx))) 
    = \frac{1}{K!} \sum_{g \in \Sym_K} f(g(\vecx)) 
    = \symmetrization{f}(\vecx) ,
\end{align*}
since composition by $g'$ defines a bijection from $\Sym_K$ onto itself.
\end{proof}

\begin{lemma}
    For any symmetric function $f:\setX^K\rightarrow\setY$, $\symmetrization{f} = f$.
\end{lemma}
\begin{proof}
    By simple computation:
\begin{align*}
    \symmetrization{f}(\vecx) = &\frac{1}{K!} \sum_{g \in \Sym_K} f(g(\vecx)) 
    = \frac{1}{K!} \sum_{g \in \Sym_K} f(\vecx) 
    = f(\vecx) .
\end{align*}
\end{proof}

\textbf{Normed policy space.}
In the proofs, we equip the policy space $\Pi$ with the norm $\|\cdot\|_1$ defined as
\begin{align*}
    \|\pi - \pi'\|_1 := \sup_{s\in\setS} \| \pi(s) - \pi'(s)\|_1,
\end{align*}
for any $\pi, \pi'\in\Pi$.
We present several useful results.

\begin{lemma}\label{lemma:cdotl1bound}
    Let $\pi,\pi' \in \Pi$ and $\mu, \mu' \in \Delta_\setS$ be arbitrary.
    Then,
    \begin{align*}
        \| \mu \cdot \pi - \mu' \cdot \pi' \|_1 \leq \| \mu - \mu' \|_1 + \|\pi - \pi' \|_1.
    \end{align*}
\end{lemma}
\begin{proof}
    The lemma follows from the two inequalities
    \begin{align*}
        \| \mu \cdot \pi - \mu \cdot \pi ' \|_1 
        &\leq \sum_{s,a} | \mu(s)\pi(a|s) - \mu(s)\pi'(a|s)| \\
        &\leq \sum_{s} \mu(s) \sum_a |\pi(a|s) - \pi'(a|s)| \leq \| \pi - \pi' \|_1 ,
    \end{align*}
    and similarly:
    \begin{align*}
        \| \mu \cdot \pi - \mu' \cdot \pi \|_1 
        &\leq \sum_{s,a} | \mu(s)\pi(a|s) - \mu'(s)\pi(a|s)| \\
        &\leq \sum_{s} | \mu(s) - \mu'(s)|\sum_a \pi(a|s) = \| \mu - \mu' \|_1 .
    \end{align*}
\end{proof}

\begin{lemma}[Lemma B.2 of \cite{yardim2023policy}]
\label{lemma:expvector_inequality}
Assume $E$ a finite set, $g: E \rightarrow \mathbb{R}^p$ a vector value function, and $\nu, \mu$ two probability measures on $E$.
Then, 
\begin{align*}
    \left\|\sum_{e} g(e) \mu(e)-\sum_{e} g(e) \nu(e)\right\|_1 \leq \frac{\lambda_g}{2} \|\mu-\nu\|_{1},
\end{align*}
where $\lambda_g := \sup_{e, e'} \|g(e) - g(e')\|_1$.
\end{lemma}

To establish explicit upper bounds on the approximation rate, we will use standard concentration tools.

\begin{lemma}
\label{lemma:concentration_bound_iid}
    Let $x^1, \ldots, x^N$ be $N$ independent Bernoulli random variables taking values in the discrete set $\setX$, with $x^i$ taking the value $1$ with some probability $p^i \in [0,1]$ for all $i\in[N]$.
    It holds that
    \begin{align*}
        \Exop\left[ \left| \frac{1}{N}\sum_{i=1}^N (p^i - x^i) \right| \right] \leq \frac{1}{\sqrt{N}},
        \qquad \Exop\left[ \left( \frac{1}{N}\sum_{i=1}^N (p^i - x^i) \right)^2 \right] \leq \frac{1}{N}.
    \end{align*}
\end{lemma}
\begin{proof}
Observing that $(x^i - p^i)$ are independent random variables,
\begin{align*}
    \Exop\left[ \left( \frac{1}{N}\sum_{i=1}^N (p^i - x^i) \right)^2 \right] = \frac{1}{N^2} \Exop\left[  \sum_{i=1}^N \left(p^i - x^i \right)^2 \right] \leq \frac{1}{N}.
\end{align*}
Furthermore, using Jensen's inequality,
\begin{align*}
    \Exop\left[ \left| \frac{1}{N}\sum_{i=1}^N (p^i - x^i) \right| \right] \leq &\Exop\left[ \sqrt{\left( \frac{1}{N}\sum_{i=1}^N (p^i - x^i) \right)^2} \right] \\
    \leq & \sqrt{\Exop\left[ \left( \frac{1}{N}\sum_{i=1}^N (p^i - x^i) \right)^2 \right] } \leq \frac{1}{\sqrt{N}}.
\end{align*}
\end{proof}

For controlling errors under stochasticity, the following simple lemma will be useful.

\begin{lemma}[Harmonic partial sum bound]\label{lemma:harmonic}
For any integers $s,\bar{s}$ such that $1 \leq \bar{s} < s$ and $p \neq -1$, it holds that
\begin{align*}
   \log s - \log \bar{s} + \frac{1}{s} &\leq \sum_{n = \bar{s}}^{s} \frac{1}{n} \leq \frac{1}{\bar{s}} + \log s - \log \bar{s}, \\
   \frac{s^{p+1}}{p+1} - \frac{\bar{s}^{p+1}}{(p+1)} + \bar{s}^p &\leq \sum_{n = \bar{s}}^{s} n^p \leq \frac{s^{p+1}}{p+1} - \frac{\bar{s}^{p+1}}{p+1} + s^p, \text{ if } p \geq 0 \\
   \frac{s^{p+1}}{p+1} - \frac{\bar{s}^{p+1}}{p+1} + s^p &\leq \sum_{n = \bar{s}}^{s} n^p \leq \frac{s^{p+1}}{p+1} - \frac{\bar{s}^{p+1}}{p+1} + \bar{s}^p, \text{ if } p \leq 0
\end{align*}
\end{lemma}
\begin{proof}
    The proof follows from the the basic fact that if $f:[1,\infty) \rightarrow \mathbb{R}_{\geq 0}$ is a non-increasing function, then
    \begin{align*}
        \int_{\bar{s}}^s f(x) \mathrm{d}x + f(s)  \leq \sum_{n = \bar{s}}^s f(n) \leq \int_{\bar{s}}^s f(x) \mathrm{d}x + f(\bar{s}) ,
    \end{align*}
    and likewise for a non-decreasing function $f:[1,\infty) \rightarrow \mathbb{R}_{\geq 0}$, it holds that 
    \begin{align*}
        \int_{\bar{s}}^s f(x) \mathrm{d}x + f(\bar{s})  \leq \sum_{n = \bar{s}}^s f(n) \leq \int_{\bar{s}}^s f(x) \mathrm{d}x + f(s) .
    \end{align*}
\end{proof}

Finally, we slightly generalize the definition of MFG-NE (Definition~\ref{def:delta_mfg_ne}), as our approximation theorems are somewhat more general than what is stated in the main body of the paper: we consider approximate MFG-NE rather than only exact MFG-NE.

\begin{definition}[$\delta$-MFG-NE]\label{def:delta_mfg_ne}
    A policy sequence $\vecpi^* \in \Pi _ H$ is called a $\delta$-MFG-NE of the MFG $(\setS, \setA, \initpop, H, P, R)$ if it holds that
    \begin{align}
        & \Expfin (\{\pi^*_h\}_{h=0}^{H-1}) \leq \delta. \tag{$\delta$-MFG-NE}
    \end{align}
\end{definition}

\textbf{A remark on extension Lemma~\ref{lemma:kirszbraun} and $\Delta_{\setS}$.} 
For given $N > 0$ and map $\bar{P}:\Delta_{\setS\times\setA, N} \rightarrow \Delta_{\setS} $ with Lipschitz modulus $L$ on $\Delta_{\setS\times\setA, N}$, the Kirszbraun-Valentine Lemma (Lemma~\ref{lemma:kirszbraun}) only guarantees an $L$-Lipschitz extension $\lipext{\bar{P}} : \Delta_{\setS\times\setA, N} \rightarrow \mathbb{R}^{\setS}$.
However, we can trivially side-step this issue with a modified application of Kirszbraun-Valentine.
Let $\operatorname{Proj}_{\Delta_{\setS}}: \mathbb{R}^{\setS} \rightarrow \Delta_{\setS}$ be the projection operator to the convex set $\Delta_{\setS}$.
For any extension $\lipext{\bar{P}}$, $\operatorname{Proj}_{\Delta_{\setS}} \circ \lipext{\bar{P}}$ is also a valid $L$-Lipschitz extension that preserves $\bar{P}$ on the set $\Delta_{\setS\times\setA, N}$ as $\operatorname{Proj}_{\Delta_{\setS}}$ is non-expansive.
Moreover, $\operatorname{Proj}_{\Delta_{\setS}} \circ \lipext{\bar{P}}$ has image set contained in $\Delta_{\setS}$ as required.

%% file: appendix/approximation.tex
\section{Extended Proofs on Approximation}

\subsection{Proof of Lemma~\ref{lemma:extension_bound}}

The proof relies on the properties of symmetrization and Lemma~\ref{lemma:expvector_inequality}.
Forn convenience we denote $K := N-1$ in this proof.

Firstly, we show that for any $i,j \in [N], s\in\setS, a\in\setA$, the functions $\symmetrization{P^i}(s,a, \cdot)$ and $\symmetrization{R^i}(s,a,\cdot)$ also satisfy the bounded variation and sparsity assumptions.
Assume that $\vecrho \in (\setS\times\setA)^K, (s',a')\in \setS\times\setA$ arbitrary,
Then, by definition,
\begin{align*}
    &\| \symmetrization{P^i}(s,a, \vecrho) - \symmetrization{P^i}(s,a, ((s',a'), \vecrho^{-j}))\|_2 \\
    &\leq \frac{1}{K!} \left\| \sum_{g \in \Sym_{K}} P^i(s,a, g(\vecrho)) -  \sum_{g \in \Sym_{K}} P^i(s,a,g((s',a'), \vecrho^{-j}))\right\|_2 \\
    &\leq \frac{1}{K!} \sum_{g \in \Sym_{K}} \left\| P^i(s,a, g(\vecrho)) - P^i(s,a,g((s',a'), \vecrho^{-j}))\right\|_2 \\
    &\leq \frac{1}{K!} \sum_{g \in \Sym_{K}} C_1 \leq C_1.
\end{align*}
Furthermore, assume $P^i(s,a, \cdot)$ is $\kappa$-sparse on some set $\setU_{s,a}\subset\setS\times\setA$ where $|\setU_{s,a}|\leq \kappa$.
Let $\vecrho,\vecrho'\in (\setS\times\setA)^{K}$ be two vectors agreeing in their entries in $\setU_{s,a}$, i.e., $p_{\setU_{s,a}}(\vecrho) = p_{\setU_{s,a}}(\vecrho')$.
Then,
\begin{align*}
    \symmetrization{P^i}(s,a, \vecrho) & = \frac{1}{K!}  \sum_{g \in \Sym_{K}} P^i(s,a, g(\vecrho)) \\
        & = \frac{1}{K!}  \sum_{g \in \Sym_{K}} P^i(s,a, g(\vecrho)) \\
        & = \symmetrization{P^i}(s,a, \vecrho'),
\end{align*}
since $g(\vecrho')$ agrees with $g(\vecrho)$ on its elements in $\setU_{s,a}$ as well, as $p_{\setU_{s,a}}(g(\vecrho)) = p_{\setU_{s,a}}(g(\vecrho'))$.
Therefore we conclude $\symmetrization{P^i}$ is also $\kappa$-sparse on $\setU_{s,a}$.
By similar computation,
\begin{align*}
    \left| \symmetrization{R^i}(s,a, \vecrho) - \symmetrization{R^i}(s,a, ((s',a'), \vecrho^{-j}))\right| \leq C_2,
\end{align*}
and $\symmetrization{R^i}$ is also $\kappa$-sparse.

Next, we establish that the lifted functions $\symmbar{P^i}(s,a,\cdot)$ only depend on $\mu(s)$ for $s\in \setU_{s,a}$, that is, we show that if $\mu, \mu'$ are such that $\mu(s',a') = \mu'(s',a')$ for all $(s',a')\in\setU_{s,a} $, then $\symmbar{P^i}(s,a,\mu) = \symmbar{P^i}(s,a,\mu')$.
Let $\mu, \mu'\in\Delta_{\setS\times\setA, K}$ be such that $\mu(s',a') = \mu'(s',a')$ for all $(s',a')\in\setU_{s,a} $.
Take arbitrary $\vecrho,\vecrho'$ such that $\empc{\vecrho} = \mu, \empc{\vecrho'} = \mu' $.
It holds that for some permutation $g' \in \Sym_{K}$ that $g'(\vecrho')$ agrees with $\vecrho$ on all entries taking values in $\setU_{s,a}$, as $\vecrho'$ and $\vecrho$ have the same count of elements in $\setU_{s,a}$.
Then
\begin{align*}
    \symmbar{P^i}(s,a,\mu) = & \frac{1}{K!}  \sum_{g \in \Sym_{K}} P^i(s,a, g(\vecrho)) \\
        = & \frac{1}{K!}  \sum_{g \in \Sym_{K}} P^i(s,a, g(g'(\vecrho'))) \\
        = & \frac{1}{K!}  \sum_{g \in \Sym_{K}} P^i(s,a, g(\vecrho')) =\symmbar{P^i}(s,a,\mu').
\end{align*}
A similar argument works for $\symmbar{R^i}(s,a,\cdot)$, allowing us to conclude that $\symmbar{P^i}(s,a,\cdot),\symmbar{R^i}(s,a,\cdot)$ only depend on $\mu(s',a')$ if $(s',a')\in\setU_{s,a} $.

Finally, we analyze the Lipschitz modulus of the lifted functions $\sfrac{1}{N}\sum_{i} \symmbar{P^i}, \sfrac{1}{N}\sum_{i} \symmbar{R^i}$.
Let $\mu_1,\mu_2\in\Delta_{\setS\times\setA, K}$ and $\vecrho_1 = \{\rho_1^i\}_{i=1}^K, \vecrho_2 = \{\rho_2^i\}_{i=1}^K \in (\setS\times\setA)^{K}$ be such that $\empc{\vecrho_1}=\mu_1, \empc{\vecrho_2}=\mu_2$.
Then,
\begin{align*}
        \|\symmbar{P^i}(s,a, \vecrho_1) - \symmbar{P^i}(s,a, \vecrho_2)\|_1 \leq C_1 \sum_{i\in [K]} \ind{\rho_1^i \neq \rho^i_2}.
\end{align*}
Taking the minimum over such $\vecrho_1, \vecrho_2$, we have that
\begin{align*}
    &\| \symmbar{P^i}(s,a, \mu_1) - \symmbar{P^i}(s,a, \mu_2) \|_1 \\
    &\leq \min_{\substack{\vecrho_1,\vecrho_2\in(\setS\times\setA)^K \\ \empc{\vecrho_1}=\mu_1, \, \empc{\vecrho_2}=\mu_2}} C_1 \sum_{i\in [K]} \ind{\rho_1^i \neq \rho^i_2} \leq C_1 K \| \mu_1 - \mu_2 \|_1,
\end{align*}
as $\vecrho_1, \vecrho_2$ can differ at a minimum at $K \| \mu_1 - \mu_2 \|_1$ coordinates.
Finally, as concluded from the arguments above, since $\symmbar{P^i}(s,a, \mu_1)$ only depends on $\mu_1(s',a')$ for $s',a'\in\setU_{s,a}$, one can choose $\bar{\mu}_1, \bar{\mu}_2 \in \Delta_{\setS\times\setA}$ such that $\bar{\mu}_1(s',a') = \mu_1(s',a')$ and $\bar{\mu}_2(s',a') = \mu_2(s',a')$ for all $(s',a')\in\setU_{s,a}$ and $\bar{\mu}_1(s'',a'') = \bar{\mu}_2(s'',a'')$ whenever $(s'',a'')\neq \setU_{s,a}$.
Then,
\begin{align*}
    \| \symmbar{P^i}(s,a, \mu_1) - \symmbar{P^i}(s,a, \mu_2) \|_1 = &\| \symmbar{P^i}(s,a, \bar{\mu}_1) - \symmbar{P^i}(s,a, \bar{\mu}_2) \|_1 \\
    = &C_1 K \|  \bar{\mu}_1 -  \bar{\mu}_2 \|_1 \\ 
    = & C_1 K \sum_{s',a' \in\setU_{s,a}} |\bar{\mu}_1(s',a') -  \bar{\mu}_2(s',a') | \\
    \leq & C_1 K \sqrt{\sum_{s',a' \in\setU_{s,a}} |\bar{\mu}_1(s',a') -  \bar{\mu}_2(s',a') |^2 } \sqrt{\kappa} \\
    \leq & C_1 K \sqrt{\kappa} \| \bar{\mu}_1 - \bar{\mu}_2\|_2 \leq  C_1 K \sqrt{\kappa} \| \mu_1 - \mu_2\|_2 ,
\end{align*}
thus proving Lipschitz bound on the set $\Delta_{\setS\times\setA, K}$.
By an identical argument, it holds that
\begin{align*}
    | &\symmbar{R^i}(s,a, \mu_1) - \symmbar{R^i}(s,a, \mu_2) | \leq C_2 (N-1) \sqrt{\kappa} \| \mu_1 - \mu_2 \|_2.
\end{align*}

The result follows from an application of Lemma~\ref{lemma:kirszbraun} to extend $\sfrac{1}{N}\sum_i \symmbar{R^i}(s,a,\cdot)$ and $\sfrac{1}{N}\sum_i \symmbar{P^i}(s,a,\cdot)$ from $\Delta_{\setS\times\setA, N}$ to $\Delta_{\setS\times\setA}$, as the norm equivalence $\|\cdot\|_2\leq\|\cdot\|_1$ holds.

\subsection{Population Flows are Lipschitz Continuous}

\begin{lemma}[Lipschitz continuity of $\Gamma$]
    Let $P: \setS\times\setA\times\Delta_{\setS\times\setA}\rightarrow \Delta_{\setS}$ be such that $P(s,a,\mu)$ is Lipschitz continuous in $\|\cdot\|_1$ norm with modulus $K_\mu >0$ and
    \begin{align*}
        K_s := \sup_{\substack{s,s'\\ a,\mu}}
        \left\|
        P(s,a,\mu)-P(s',a,\mu)
        \right\|_1, \qquad
        K_a := \sup_{\substack{a,a'\\ s,\mu}}
        \left\|
        P(s,a,\mu)-P(s,a',\mu)
        \right\|_1.
    \end{align*}
    Then it holds for all $\mu, \mu' \in \Delta_{\setS\times\setA}, \pi,\pi' \in \Pi$ that:
    \label{lemma:lipschitz_gpop}
    \[
    \|\gpop(\mu,\pi)
    -
    \gpop(\mu',\pi)\|_1
    \leq 
    \left(\frac{K_s + K_a}{2} + K_\mu\right) \|\mu-\mu'\|_1,
    \]
    for all $\pi\in\Pi$, $\mu,\mu' \in\Delta_{\setS\times\setA}$.
\end{lemma}
\begin{proof}

The proof is inspired by \cite{yardim2023policy}, apart from the fact that in our case the population update operator is defined differently as:
\begin{align*}
    \gpop(\mu, \pi)(s',a') &:= \sum_{s \in \setS, a \in \setA} \mu(s,a) P(s'|s,a,\mu) \pi(a'|s')
\end{align*}
We will prove a slightly more general statement, that
\begin{align*}
    &\left\|\gpop(\mu, \pi) - \gpop(\mu', \pi')\right\|_1 \leq \|\mu - \mu' \|_1 \left(\frac{K_s + K_a}{2} + K_\mu\right) + \|\pi-\pi'\|_1.
\end{align*}

Firstly, we upper bound the Lipschitz modulus of the function with respect to $\mu$.
For any $\mu,\mu'\in\Delta_{\setS\times\setA}$, it holds that:
\begin{align*}
    &\left\|\gpop(\mu, \pi) - \gpop(\mu', \pi)\right\|_1 \\
    &\leq \sum_{s',a'} \left| \sum_{s,a} (\mu(s,a) P(s'|s,a,\mu) - \mu'(s,a) P(s'|s,a,\mu')) \pi(a'|s')\right| \\
    &\leq \sum_{s',a'} \left| \sum_{s,a} (\mu(s,a)  - \mu'(s,a) ) P(s'|s,a,\mu) \pi(a'|s')\right| \\
        & \quad + \sum_{s',a'} \left| \sum_{s \in \setS, a \in \setA} \mu'(s,a) (P(s'|s,a,\mu) -  P(s'|s,a,\mu')) \pi(a'|s')\right| \\
    &\leq   \left\| \sum_{s ,a} (\mu(s,a)  - \mu'(s,a)) (P(s,a,\mu) \cdot \pi) \right\|_1 \\
        & \quad +   \sum_{s , a} \mu'(s,a) \sum_{s'} |P(s'|s,a,\mu) -  P(s'|s,a,\mu')| \sum_{a'} \pi(a'|s') \\
    &\leq  \|\mu - \mu' \|_1 \frac{\max \|P(s,a,\mu) \cdot \pi - P(\widebar{s},\widebar{a},\mu) \cdot \pi\|_1}{2} + K_\mu \|\mu - \mu' \|_1 \\
    &\leq  \|\mu - \mu' \|_1 \frac{\max \|P(s,a,\mu) - P(\widebar{s},\widebar{a},\mu) \|_1}{2} + K_\mu \|\mu - \mu' \|_1.
\end{align*}
where the last two lines follow from Lemma~\ref{lemma:expvector_inequality} and Lemma~\ref{lemma:cdotl1bound}.
Since 
\begin{align*}
    \|P(s,a,\mu) - P(\widebar{s},\widebar{a},\mu) \|_1 \leq K_s + K_a,
\end{align*}
we have the claimed inequality
\begin{align*}
    \left\|\gpop(\mu, \pi) - \gpop(\mu', \pi)\right\|_1 \leq \|\mu - \mu' \|_1 \left(\frac{K_s + K_a}{2} + K_\mu\right).
\end{align*}

Finally, the Lipschitz constant for the policy $\pi$ is computed by:
\begin{align*}
    \left\|\gpop(\mu, \pi) - \gpop(\mu, \pi')\right\|_1 &\leq \sum_{s',a'} \left| \sum_{s,a} \mu(s,a) P(s'|s,a,\mu) \left(\pi(a'|s') - \pi'(a'|s')\right)\right| \\
    &\leq \sum_{s,a,s'}\mu(s,a) P(s'|s,a,\mu) \sum_{a'}  \left| \pi(a'|s') - \pi'(a'|s')\right| \\
    &\leq \| \pi - \pi' \|_1.
\end{align*}
\end{proof}

\subsection{Proof of Theorem~\ref{theorem:main_approximation}}\label{sec:proof_main_approx}

The main ideas of the approximation proof are similar to some arguments from MFG literature (e.g. see \cite{saldi2018markov}) with two major differences: 
(1) the dynamics of the finite player game are not exactly symmetric, and
(2) unlike some standard works the dynamics and rewards depend on the distribution of agents over state-action pairs, not just states.

For given $R$ and $P$ define the following constants:
    \begin{align*}
        L_s &= \sup_{s,s',a,\mu}
        \left|
        R(s,a,\mu)-R(s',a,\mu)
        \right|, \qquad
        L_a = \sup_{s,a,a',\mu}
        \left|
        R(s,a,\mu)-R(s,a',\mu)
        \right|, \\
        K_s &= \sup_{s,s',a,\mu}
        \left\|
        P(\cdot|s,a,\mu)-P(\cdot|s',a,\mu)
        \right\|_1, \qquad
        K_a = \sup_{s,a,a',\mu}
        \left\|
        P(\cdot|s,a,\mu)-P(\cdot|s,a',\mu)
        \right\|_1.
    \end{align*} 
We also introduce the shorthand notation for any $s\in \setS, u\in \Delta_\setA, \mu \in \Delta_{\setS\times\setA}$:
\begin{align*}
    P(\cdot|s,u,\mu)&= \sum_{a\in\setA}u(a)P(\cdot|s,a,\mu), \qquad
    R(s,u,\mu) = \sum_{a\in\setA}u(a)R(s,a,\mu).
\end{align*}
By \cite[Lemma~C.1]{yardim2023policy}, it holds that
\begin{align}
\|P(\cdot|s,u,\mu)
-
P(\cdot|s',u',\mu')\|_1
\leq &
K_\mu\|\mu-\mu'\|_1+K_s d(s,s') +\frac{K_a}{2}\|u-u'\|_1, \notag \\
|R(s,u,\mu)
-
R(s',u',\mu')|
\leq &L_\mu \|\mu-\mu'\|_1
+
L_s d(s,s') +\frac{L_a}{2}\|u-u'\|_1. \label{eq:lipschitz_pr_policy_avg}
\end{align}

We will define a new operator for tracking the evolution of the population distribution over finite time horizons for a time-varying policy $\forall \pi = \{ \pi_h \}_{h=0}^{H-1} \in \Pi$:
\begin{align*}
   \Gamma^h(\mu, \pi)&:= 
    \underbrace{\Gamma(\ldots \Gamma(\Gamma(\mu, \pi_{0}), \pi_{1}) \ldots, \pi_{h-1})}_{h\text{ times}}
\end{align*}
so that $\Gamma^0(\mu, \pi) := \mu$.
Lemma~\ref{lemma:lipschitz_gpop} yields the Lipschitz condition:
\begin{align}
        \|
        &\Gamma^n(\mu, \{\pi_{i}\}_{i=0}^{n-1} ) -\Gamma^n(\mu', \{\pi_{i}\}_{i=0}^{n-1} )
        \|_1 \notag\\
        &\leq
        \lpopmu 
        \|
        \Gamma^{n-1}(\mu, \{\pi_{i}\}_{i=0}^{n-2})-\Gamma^{n-1}(\mu', \{\pi'_{i}\}_{i=0}^{n-2})
        \|_1 +
        \|\pi_{n-1}-\pi_{n-1}'\|_1 \notag\\
        &\leq 
        \lpopmu^{n} \|\mu-\mu'\|_1
        +
        \sum_{i=0}^{n-1}\lpopmu^{n-1-i}\|\pi_{i}-\pi_{i}'\|_1,  \label{eq:gpop_composition_lip}
\end{align}
where $\lpopmu$ is the Lipschitz constant of $\Gamma$ in $\mu$.

We also define a useful function $\Xi: (\setS\times\setA)^N \times \Pi^N \rightarrow \Delta_{\setS\times\setA}$ such that for any $\vecrho = \{(s^i, a^i)\}_{i=1}^N$,
\begin{align*}
    \Xi( \vecrho, \refpol ) = \frac{1}{N}\sum_{i=1}^N  P(\cdot|s^i,a^i,\empc{\vecrho^{-i}}) \cdot \refpol.
\end{align*}
In other words, $\Xi$ is the average population flow expected under symmetrized dynamics and reference policy $\refpol$.

The proof will proceed in four steps:
\begin{itemize}
    \item \textbf{Step 1.} Bounding the expected deviation of the empirical population distribution from the mean-field distribution $\EE{\|\empdist{h}-\mu_h\|_1}$ for any given policy $\pi$.
    \item \textbf{Step 2. } Bounding total variation distance (or equivalently $\ell_1$ distance) between the marginal distributions $\Prob[s_h^1 = \cdot]$ in the $N$-player game and $\Prob[s_h=\cdot]$ in the mean-field game,
    \item \textbf{Step 3. } Bounding difference of $N$ agent value function $\JfinNi{1}$ and the infinite player value function $\Vfin$, when all the players except the first one play the same policy,
    \item \textbf{Step 4. } Bounding the exploitability of an agent when each of $N$ agents are playing the FH-MFG-NE policy.
\end{itemize}

\textbf{Step 1: Empirical distribution bound.}
Due to its relevance for a general connection between the FH-MFG and the $N$-player game, we state this result in the form of an explicit bound.
In this step, we will assume $N$ players of the FH-DG pursue policies $\{ \pi^i \}_i = \{\pi^i_h\}_{i,h} \in \Pi^N$, and random variables $f$
Furthermore, assume $\refpolvec = \{\refpol_h\}_h \in \Pi$ arbitrary, and induces population $\vecmu = \lpop(\refpolvec) = \{ \mu_h \}_h$.
We also define the quantity $\Delta_h := \frac{1}{N}\sum_{i=1}^N \| \pi^i_h - \refpol_h \|_1$ and $\widebar{\Delta} := \max_{h\in [H]} \Delta_h$.

The proof will proceed inductively over $h$.
First, for time $h=0$, we have
\begin{align*}
     \EE{
    \|\empdist{0}-\mu_0\|_1
    }
    \leq & \EE{\left\|\empdist{0}-\frac{1}{N}\sum_{i=1}^N \initpop \cdot \pi^i_0\right\|_1} + \EE{\left\|\frac{1}{N}\sum_{i=1}^N \initpop \cdot \pi^i_0-\mu_0\right\|_1} \\
    \leq &
    \sum_{ \substack{s\in\setS\\ a\in\setA}}
    \EE{
    \left|\frac{1}{N}\sum_{i=1}^N (\ind{s_0^i=s, a_0^i = a}-\initpop(s)\pi^i_0(a|s))\right|} \\
        & + \frac{1}{N} \EE{ \sum_{i=1}^N \| \pi^i_0 - \refpol_0 \|_1 } \\
    \leq &
    |\setS| |\setA| \frac{1}{\sqrt{N}} + \Delta_h,
\end{align*}
where the last line is due to \Cref{lemma:concentration_bound_iid} and the fact that $\ind{s_0^i=s, a_0^i = a}$ are independent, bounded random variables, and that we have $\EE{\ind{s_0^i=s, a_0^i = a}}= \initpop (s) \pi(s,a) = \mu_0(s,a)$.

Next, denoting the $\sigma$-algebra induced by the random variables $(\{s^i_{h}, a^i_{h} \})_{i, h' \leq h}$ as $\mathcal{F}_h$, we have that:
\begin{align}
     \EEc{\|
        \empdist{h+1}-\mu_{h+1}
        \|_1}{\cF_h} 
         \leq  & \underbrace{\EEc{\|
        \empdist{h+1}-\EEc{\empdist{h+1}}{\cF_h}
        \|_1}{\cF_h}}_{(\triangle)} \notag \\
        &+ 
        \underbrace{
            \EEc{\| \EEc{\empdist{h+1}}{\cF_h}- \Xi(\vecrho_h, \refpol_{h+1}) \|_1}{\cF_h}
        }_{(\square)} \notag \\
        &+ \underbrace{
            \EEc{\|
                \Xi(\vecrho_h, \refpol_{h+1})-\gpop(\empdist{h}, \refpol_h)
                \|_1}{\cF_h}
        }_{(\star)} \notag\\
         & + \underbrace{\EEc{\|
        \gpop(\empdist{h}, \widebar{\pi}_h)-\mu_{h+1}
        \|_1}{\cF_h} }_{(\heartsuit)} \label{eq:main_inductive_bound_lemma_dist}
\end{align}
We upper bound the four terms separately.
For $(\triangle)$, it holds that
\begin{align*}
    (\triangle) = &\EEc{
    \|
    \empdist{h+1}-
    \EEc{\empdist{h+1}}{\cF_h}
    \|_1
    }{\cF_h} \\    
    =&
    \sum_{s\in\setS, a\in\setA}
    \EEc{
    |
    \empdist{h+1}(s, a)-
    \EEc{\empdist{h+1}(s, a)}{\cF_h}
    |
    }{\cF_h} \\
    \leq & |\setS||\setA|\frac{1}{\sqrt{N}},
\end{align*}
since each $\empdist{h+1}(s)$ is an average of $N$ independent random variables (specifically $N$ independent Bernoulli random variables) given $\cF_h$, using Lemma~\ref{lemma:concentration_bound_iid}. 

Next, for the term $(\square)$,
\begin{align*}
    (\square) &= \EEc{\| \EEc{\empdist{h+1}}{\cF_h}- \Xi(\vecrho_h, \refpol_{h+1}) \|_1}{\cF_h} \\
    &\begin{aligned}
        = \frac{1}{N} \Exop\bigg[ &\bigg\|  \sum_{i=1}^N P^i(\cdot|s^i_h,a^i_h,\vecrho^{-i}_h) \cdot \pi^i_{h+1} -  \sum_{i=1}^N  P(\cdot|s^i_h,a^i_h,\empc{\vecrho^{-i}_h}) \cdot \refpol_{h+1} \bigg\|_1 \bigg| \cF_h \bigg]
    \end{aligned}\\
    &\leq \frac{1}{N} \Exop\bigg[ \sum_{i=1}^N \|P^i(\cdot|s^i_h,a^i_h,\vecrho^{-i}_h) - P(\cdot|s^i_h,a^i_h,\empc{\vecrho^{-i}_h})\|_1 \bigg| \cF_h \bigg] \\
    &\quad + \frac{1}{N} \sum_{i=1}^N \| \pi^i_{h+1} - \refpol_{h+1} \|_1.
\end{align*}
By the $\alpha$-symmetry condition, it follows that $(\square) \leq \alpha + \Delta_{h+1}$.
    
For $(\star)=\|\Xi(\vecrho_h, \refpol_h \}_i)-\gpop(\empdist{h}, \widebar{\pi}_h)\|_1$,
\begin{align*}
    (\star)
    &= \EEc{ \|\Xi(\vecrho_h, \refpol_{h+1})-\gpop(\empdist{h}, \refpol_{h+1})\|_1}{\cF_h} \\
    &\begin{aligned}
        = \Exop \bigg[ \bigg\| &\frac{1}{N}\sum_{i=1}^N P(\cdot|s^i_h,a^i_h,\empc{\vecrho^{-i}_h}) \cdot \refpol_{h+1} -
        \sum_{s',a'}\empdist{h}(s',a') P(\cdot|s',a',\empdist{h}) \cdot \refpol_{h+1} \bigg\|_1 \bigg| \cF_h \bigg].
    \end{aligned} \\
    &\begin{aligned}
        = \frac{1}{N} \Exop \bigg[ \bigg\| &\sum_{i=1}^N P(\cdot|s^i_h,a^i_h,\empc{\vecrho^{-i}_h}) \cdot \refpol_{h+1} -
        \sum_{i=1}^N P(\cdot|s^i_h,a^1_h,\empdist{h}) \cdot \refpol_{h+1} \bigg\|_1 \bigg| \cF_h \bigg].
    \end{aligned}
\end{align*}
The vectors $(N-1)\empc{\vecrho^{-i}_h}, N \empdist{h}$ can differ by only 1 in one component due to the $i$-th agent being excluded from the former, it holds that
\begin{align*}
    &\|N \empc{\vecrho^{-i}_h} - N \empdist{h} \|_1 \leq \|(N-1) \empc{\vecrho^{-i}_h} - N \empdist{h} \|_1 + \| \empc{\vecrho^{-i}_h} \|_1 \leq 3,
\end{align*}
therefore for any $s,a$,
\begin{align*}
    \|P(\cdot|s,a,\empc{\vecrho^{-i}_h} ) - P(\cdot|s,a,\empdist{h} ) \|_1 \leq \frac{3K_\mu}{N}
\end{align*}
almost surely and $(\star)$ is further upper bounded by $(\star) \leq \frac{3 K_\mu}{N}.$
    
Finally, the last term $(\heartsuit)$ can be bounded using:
\begin{align*}
    (\heartsuit) = &\EEc{\|
        \gpop(\empdist{h}, \widebar{\pi}_h)-
        \gpop(\mu_{h}^{\widebar{\pi}}, \widebar{\pi}_h)
        \|_1}{\cF_h} 
        \leq \lpopmu \|\empdist{h} - \mu_{h} \|_1.
\end{align*}
To conclude, merging the bounds on the three terms in Inequality~\eqref{eq:main_inductive_bound_lemma_dist} and taking the expectation on both sides, by the law of iterated expectations we obtain:
\begin{align*}
    \EE{\| \empdist{h+1} - \mu_{h+1}\|_1} \leq &\lpopmu \EE{\| \empdist{h} - \mu_{h}\|_1} + |\setS||\setA|\frac{1}{\sqrt{N}} + \frac{3K_\mu}{N} + \Delta_{h+1} + \alpha.
\end{align*}

Induction on $h$ yields the bound for all $h$:
\begin{align}
    &\EE{\| \empdist{h} - \mu_{h}\|_1} \notag \\
    &\leq \sum_{h'=0}^{h} \lpopmu^{h-h'} \left(\frac{|\setS||\setA|}{\sqrt{N}} + \frac{3K_\mu}{N} + \Delta_{h+1} + \alpha\right) \notag \\
    &\leq \frac{1 - \lpopmu^{h+1}}{1 - \lpopmu}\left(\frac{|\setS||\setA|}{\sqrt{N}} + \widebar{\Delta} + \frac{3K_\mu}{N} + \alpha\right), \label{ineq:popdistl1}
\end{align}
where we adopt the convenient shorthand $\frac{1 - \lpopmu^{h+1}}{1 - \lpopmu} := h$ if $\lpopmu = 1$.

\textbf{Step 2: Marginal state-action distributions.}
In this step, we analyze the distributions $\Prob[s_h^i = \cdot, a_h^i=\cdot]$.
For simplicity, assume each player $i\neq 1$ follows policy $\pi \in \Pi$,
player $1$ follows an arbitrary policy $\refpolvec \in \Pi$: we denote the induced random variables in the $N$-player game $\setG$ as $s_h^i, a_h^i, \empdist{h}$.
Assume that in the mean-field game $\inducedmfg{\setG}$, the representative player in $\inducedmfg{\setG}$ also follows policy $\refpolvec$, evaluated against distribution $\vecmu := \lpop(\pi)$: denote the induced random variables as $s_h, a_h$.
In this setting, we will inductively upper bound the quantity
\begin{align*}
    \| \Prob[s_h = \cdot] - \Prob[s_h^1 = \cdot] \|_1 .
\end{align*}

Firstly, for $h=0$, it holds by definition that
\begin{align*}
    \Prob[s_0 = \cdot] = \Prob[s^1_0 = \cdot] = \initpop,
\end{align*}
hence $\| \Prob[s_0 = \cdot] - \Prob[s_0^1 = \cdot] \|_1 = 0$.

Next, for the time step $h+1$,
\begin{align}
    \| &\Prob[s_{h+1} = \cdot] - \Prob[s^1_{h+1} = \cdot] \|_1 \notag \\
        \leq & \bigg\| \sum_{s,\vecrho} P^1(s, \refpol_h(s), \vecrho) \Prob[s_h^1=s, \vecrho_h^{-i}=\vecrho] - \sum_{s} P(s, \refpol_h(s), \mu_h) \Prob[s_h=s] \bigg \|_1 \notag \\
        \leq & \bigg\| \sum_{s,\vecrho} P(s, \refpol_h(s), \empc{\vecrho}) \Prob[s_h^1=s, \vecrho_h^{-i}=\vecrho] - \sum_{s} P(s, \refpol_h(s), \mu_h) \Prob[s_h=s] \bigg \|_1 \notag \\
            & + \bigg\|\sum_{s,\vecrho} [P^1(s, \refpol_h(s), \vecrho) - P(s, \refpol_h(s), \empc{\vecrho})] \Prob[s_h^1=s, \vecrho_h^{-i}=\vecrho]\bigg\|_1 \notag \\
        \leq & \bigg\| \sum_{s,\vecrho} P(s, \refpol_h(s), \empc{\vecrho}) \Prob[s_h^1=s, \vecrho_h^{-i}=\vecrho] - \sum_{s} P(s, \refpol_h(s), \mu_h) \Prob[s_h=s] \bigg \|_1 + \alpha, \label{ineq:approxproofsteptwoind}
\end{align}
where the last line follows from the $\alpha$-symmetry condition.
For the remaining term, we first observe the inequality:
\begin{align*}
    &\bigg\| \sum_{s,\vecrho} [P(s, \refpol_h(s), \empc{\vecrho}) - P(s, \refpol_h(s), \mu_h)] \Prob[s_h^1=s, \vecrho_h^{-i}=\vecrho]\bigg\|_1 \\
    &\leq \sum_{s,\vecrho} K_\mu \| \empc{\vecrho} - \mu_h \| _1 \Prob[s_h^1=s, \vecrho_h^{-i}=\vecrho] \\
     &\leq K_\mu \Exop[\| \empc{\vecrho}_h^{-i} - \mu_h \| _1] \\
    &\leq K_\mu \Exop[\| \empdist{h} - \mu_h \| _1] + \frac{3K_\mu}{N}
\end{align*}
as again $\|\empdist{h} - \empc{\vecrho^{-i}_h} \|_1 \leq \sfrac{3}{N}$ almost surely, as $N\empdist{h}$ and $(N-1)\empc{\vecrho^{-i}_h}$ differ by one in one coordinate only.
Then, applying the triangle inequality and marginalizing over $\vecrho$ in Inequality~\eqref{ineq:approxproofsteptwoind},
\begin{align*}
    \| &\Prob[s_{h+1} = \cdot] - \Prob[s^1_{h+1} = \cdot] \|_1 \\
    \leq & \bigg\| \sum_{s} P(s, \refpol_h(s), \mu_h) \Prob[s_h^1=s] - \sum_{s} P(s, \refpol_h(s), \mu_h) \Prob[s_h=s] \bigg \|_1 \\
        &+ \alpha + \frac{3}{N} + K_\mu \Exop[\| \empdist{h} - \mu_h \| _1] \\
    \leq & \| \Prob[s_h^1=s] - \Prob[s_h=s] \|_1 + \alpha + \frac{3K_\mu}{N} + K_\mu \Exop[\| \empdist{h} - \mu_h \| _1]
\end{align*}
where in the last line we used the fact that Markov kernels are non-expansive in $\ell_1$ norm, where in this case the Markov kernel is given by $P(\cdot|s, \refpol_h(s), \mu_h)$ for all $s\in\setS$.

Inductively applying the recursive bound above we obtain the inequality for all $h$:
\begin{align}
    &\| \Prob[s_h^1=\cdot] - \Prob[s_h=\cdot] \|_1 \leq (h-1)\left(\alpha + \frac{3K_\mu}{N}\right) + K_\mu \sum_{h=0}^{h-1} \Exop[\| \empdist{h} - \mu_h \| _1]. \label{ineq:marg_diff_l1_notsub}
\end{align}
By the result in Step 1 (Inequality~\ref{ineq:popdistl1}), since $\widebar{\Delta} \leq \sfrac{2}{N}$ in this case it holds that
\begin{align*}
    &\EE{\| \empdist{h} - \mu_{h}\|_1} \leq \frac{1 - \lpopmu^{h+1}}{1 - \lpopmu}\left(\frac{|\setS||\setA|}{\sqrt{N}} + \frac{5K_\mu}{N} + \alpha\right),
\end{align*}
for all $h=0,\ldots,N-1$.
Merging the two inequalities:
\begin{align}
    &\| \Prob[s_h^1=\cdot] - \Prob[s_h=\cdot]  \|_1 \notag \\
    &\leq \lpopmu \sum_{h'=0}^{h-1} \frac{1 - \lpopmu^{h'+1}}{1 - \lpopmu}\left(\frac{|\setS||\setA|}{N} + \frac{5K_\mu}{N} + \alpha\right) + (h-1)\left(\alpha + \frac{3K_\mu}{N}\right) \label{ineq:marg_diff_l1}.
\end{align}

\textbf{Step 3. Bounding value function.}
In this step, we bound the difference between the expected returns of a player in the $N$ player game and the induced MFG (denoted by functions  $J,V$ respectively).
Namely, we will upper bound the deviation
\begin{align*}
    |\JfinNi{1}(\refpolvec, \pi, \ldots, \pi) - \Vfin(\lpop(\pi), \refpolvec) |
\end{align*}
for any two policies $\refpolvec, \pi \in \Pi$.

As in step 1, assume each player $i\neq 1$ follows policy $\pi \in \Pi$, and
player $1$ follows policy $\refpolvec$: we denote the induced random variables in the $N$-player game $\setG$ as $s_h^i, a_h^i, \empdist{h}$.
Assume that in the mean-field game $\inducedmfg{\setG}$, the representative player in $\inducedmfg{\setG}$ also follows policy $\refpolvec$, evaluated against distribution $\vecmu := \lpop(\pi)$: denote the induced random variables as $s_h, a_h$.

By the results in Step 1,2 shown in inequalities~\eqref{ineq:popdistl1}, \eqref{ineq:marg_diff_l1_notsub}, since $\widebar{\Delta} \leq \sfrac{2}{N}$ in this case it holds that:
\begin{align*}
    &\EE{\| \empdist{h} - \mu_{h}\|_1} \leq \frac{1 - \lpopmu^{h+1}}{1 - \lpopmu}\left(\frac{|\setS||\setA|}{\sqrt{N}} + \frac{5K_\mu}{N} + \alpha\right) , \\
    &\| \Prob[s_h^1=\cdot] - \Prob[s_h=\cdot] \|_1  
     \leq (h-1)\left(\alpha + \frac{3K_\mu}{N}\right) + K_\mu \sum_{h=0}^{h-1} \Exop[\| \empdist{h} - \mu_h \| _1] ,
\end{align*}
for all $h=0,\ldots,N-1$.

At a fixed time step $h$, the expected one-step reward differences can be decomposed into four terms:
\begin{align*}
    & | \Exop[R(s_h, a_h, \mu_h)] - \Exop[R^1(s^1_h, a^1_h, 
    \vecrho^{-i})] | \\
    &\leq | \Exop[R(s_h, a_h, \mu_h)] - \Exop[R(s^1_h, a^1_h, 
    \mu_h)] | \\
    &\quad + | \Exop[R(s^1_h, a^1_h, \mu_h)] - \Exop[R(s^1_h, a^1_h, 
    \empdist{h})] | \\
    &\quad + | \Exop[R(s^1_h, a^1_h, 
    \empdist{h})] - \Exop[R(s^1_h, a^1_h, 
    \empc{\vecrho^{-i}})] | \\
    &\quad + | \Exop[R(s^1_h, a^1_h, 
    \empc{\vecrho^{-i}})] - \Exop[R^1(s^1_h, a^1_h, 
    \vecrho^{-i})] |,
\end{align*}
enumerating these terms as (I), (II), (III), and (IV), we upper bound each as follows:
\begin{align*}
    \text{(I)} \leq & \frac{L_a}{2} \| \Prob[s_h = \cdot, a_h = \cdot] - \Prob[s^1_h = \cdot, a^1_h = \cdot]\|_1  \\
        \leq & \frac{L_a}{2}  \| \Prob[s_h = \cdot] - \Prob[s^1_h = \cdot]\|_1 \\
    \text{(II)} \leq & L_\mu \Exop [ \|\empdist{h} - \mu_h\|_1 ] \\
    \text{(III)} \leq & L_\mu \Exop [ \|\empdist{h} - \empc{\vecrho^{-i}}\|_1 ] \leq L_\mu \frac{3}{N} \\
    \text{(IV)} \leq & \beta
\end{align*}
where the last line uses the $\alpha,\beta$ symmetry condition.

Then, summing up over the entire time horizon, we obtain the result
\begin{align}
    &|\JfinNi{1}(\refpolvec, \pi, \ldots, \pi) - \Vfin(\lpop(\pi), \refpolvec) | \notag \\
    &\leq \sum_{h=0}^{H-1} \left( L_\mu \Exop [ \|\empdist{h} - \mu_h\|_1 ]+\frac{L_a}{2} \| \Prob[s_h = \cdot] - \Prob[s^1_h = \cdot]\|_1\right) + \beta H + \frac{3HL_\mu}{N} \label{eq:JVdiff_nosubs}
\end{align}
or substituting the upper bounds established before,
\begin{align}
    &|\JfinNi{1}(\refpolvec, \pi, \ldots, \pi) - \Vfin(\lpop(\pi), \refpolvec) | \notag \\
    &\leq \sum_{h=0}^{H-1} \left( L_\mu \Exop [ \|\empdist{h} - \mu_h\|_1 ]+\frac{L_a K_\mu }{2}  \sum_{h'=0}^{h-1} \Exop [ \|\empdist{h'} - \mu_{h'}\|_1 ] \right) \notag\\
        &\quad + \beta H + \frac{3HL_\mu}{N} + H^2\left( \alpha + \frac{ 3K_\mu }{N} \right) \notag \\
    &\leq \sum_{h=0}^{H-1} \left( L_\mu E^{(h)} +\frac{L_a K_\mu }{2}  \sum_{h'=0}^{h-1} E^{(h')} \right) 
         + \beta H + \frac{3HL_\mu}{N} + H^2\left( \alpha + \frac{ 3K_\mu }{N} \right)
        \label{eq:JVdiff}
\end{align}
where we define the quantity
\begin{align*}
    E^{(h)} := \frac{1 - \lpopmu^{h+1}}{1 - \lpopmu}\left(\frac{|\setS||\setA|}{\sqrt{N}} + \frac{5K_\mu}{N} + \alpha\right).
\end{align*}

\textbf{Step 4. Bounding exploitability function.}
Finally, we use the results from the previous steps to upper bound the exploitability of the MFG-NE policy $\pi^*$ in the FH-DG.
Let $\vecmu^* = \lpop (\pi^*)$.
Let $\pi'$ be arbitrary.
\begin{align*}
    \JfinNi{1}( \pi', \pi^*, \ldots, \pi^* ) - \JfinNi{1} ( \pi^*, \ldots, \pi^* ) \leq &\Vfin( \lpop (\pi^*), \pi' ) - \Vfin ( \lpop ( \pi^*), \pi^* ) \\
        & + |\JfinNi{1}(\pi^*, \pi^*, \ldots, \pi^*) - \Vfin(\lpop(\pi^*), \pi^*) | \\
        & + |\JfinNi{1}(\pi', \pi^*, \ldots, \pi^*) - \Vfin(\lpop(\pi^*), \pi') | .
\end{align*}
The last two terms in this inequality can be bounded by Inequality~\eqref{eq:JVdiff} by choosing $\refpolvec = \pi^*$ and $\refpolvec = \pi'$ respectively, and the first term satisfies
\begin{align*}
    \Vfin( \lpop (\pi^*), \pi' ) - \Vfin ( \lpop ( \pi^*), \pi^* ) \leq \delta
\end{align*}
as $\pi$ is assumed to be a $\delta$-MFG-NE.

Then, the main statement of the theorem is obtained by observing
\begin{align*}
    \ExpfinNi{i} (\pi^*) &= \max_{\pi' \in \Pi} \JfinNi{i} \left(\pi', \pi^{*, -i} \right) - \JfinNi{i} \left( \pi^*, \ldots, \pi^* \right) \\
    &\leq \delta + 2\sum_{h=0}^{H-1} \left( L_\mu E^{(h)} +\frac{L_a K_\mu }{2}  \sum_{h'=0}^{h-1} E^{(h')} \right) 
         + 2\beta H + \frac{6HL_\mu}{N} + 2H^2\left( \alpha + \frac{ 3K_\mu }{N} \right),
\end{align*}
where once again
\begin{align*}
    E^{(h)} := \frac{1 - \lpopmu^{h+1}}{1 - \lpopmu}\left(\frac{|\setS||\setA|}{\sqrt{N}} + \frac{5K_\mu}{N} + \alpha\right).
\end{align*}

Namely, if $\lpopmu = 1$, then
\begin{align*}
    E^{(h)} = \cO\left( h \left(\alpha + \frac{1}{\sqrt{N}} \right) \right),
\end{align*}
if $\lpopmu < 1$, then
\begin{align*}
    E^{(h)} = \cO\left(\alpha + \frac{1}{\sqrt{N}} \right),
\end{align*}
and finally if $\lpopmu > 1,$ then
\begin{align*}
    E^{(h)} = \cO\left( \lpopmu^h \left(\alpha + \frac{1}{\sqrt{N}} \right) \right).
\end{align*}

%% file: appendix/learning.tex
\section{Extended Results on TD Learning}

We will make use of the following technical lemma.

\begin{lemma}\label{lemma:pop_error_squared}
    Let $\setG = (\setS, \setA, \initpop, N, H, \{ P^i\}_{i=1}^N, \{R^i\}_{i=1}^N)$ be a FH-DG which induces $\inducedmfg{\setG}=(\setS, \setA, N, H, P, R)$.
    Furthermore, assume that the population flow operator $\gpop$ of the induced MFG satisfies the Lipschitz condition
    \begin{align*}
        \| \gpop ( \mu, \pi) - \gpop ( \mu', \pi)\|_2 \leq \lpopmu \| \mu - \mu'\|_2,
    \end{align*}
    for all policies $\pi$ and $\mu,\mu'\in\Delta_{\setS\times\setA}$.
    Then, if all players in the FH-DG play policy $\pi \in \Pi$, it holds that
    \begin{align*}
        \EE{\| \empdist{h} - \mu_{h}\|_2^2} \leq C \frac{1 - \lpopmu^{2(h+1)}}{1 - \lpopmu^2}  \left(\frac{|\setS||\setA|}{N} + \frac{18 K_\mu^2 (H^2 + 2)}{N^2} + 2(H^2 + 2)\alpha^2\right), 
    \end{align*}
    for some absolute constant $C> 0$.
\end{lemma}
\begin{proof}
Due to its relevance for a general connection between the FH-MFG and the $N$-player game, we state this result in the form of an explicit bound.
In this step, we will assume $N$ players of the FH-DG pursue policies $\{ \pi^i \}_i = \{\pi^i_h\}_{i,h} \in \Pi^N$ such that $\pi^i = \pi$ for some $\pi\in\Pi$, and random variables $\{s^i_{h}, a^i_{h} \}_{i, h' \leq h}$, $\vecrho _ h \in (\setS\times\setA)^{N}$ are generated according to the finite player dynamics.

The proof will proceed inductively over $h$.
First, for time $h=0$, we have
\begin{align*}
     \EE{
    \|\empdist{0}-\mu_0\|_2^2
    }
    \leq & \EE{\left\|\empdist{0}-\initpop \cdot \pi_0\right\|_2^2}  \\
    \leq & \sum_{s,a} \EE{ \left(\frac{1}{N}\sum_{i=1}^N \ind{s_0^i=s, a_0^i = a} - (\initpop \cdot \pi_0)(s,a) \right)^2}  \\
    \leq & \sum_{s,a} \frac{1}{N^2}\EE{\sum_{i=1}^N \left(\ind{s_0^i=s, a_0^i = a} - (\initpop \cdot \pi_0)(s,a) \right)^2}  \\
    \leq &
     \frac{|\setS| |\setA|}{N},
\end{align*}
due to the fact that $\ind{s_0^i=s, a_0^i = a}$ are independent, bounded random variables, and that we have $\EE{\ind{s_0^i=s, a_0^i = a}}= \initpop (s) \pi(s,a) = \mu_0(s,a)$.

Next, denoting the $\sigma$-algebra induced by the random variables $\{s^i_{h}, a^i_{h} \}_{i, h' \leq h}$ as $\mathcal{F}_h$, we have that:
\begin{align}
     &\EEc{\|
        \empdist{h+1}-\mu_{h+1}
        \|_2^2}{\cF_h} \\
        &\leq \EEc{\|
        \empdist{h+1}-\EEc{\empdist{h+1}}{\cF_h}
        \|_2^2}{\cF_h} + \EEc{\|
        \mu_{h+1}-\EEc{\empdist{h+1}}{\cF_h}
        \|_2^2}{\cF_h} \\
        &\leq \EEc{\|
        \empdist{h+1}-\EEc{\empdist{h+1}}{\cF_h}
        \|_2^2}{\cF_h} \notag \\
            &\quad + (1 + \delta_h^{-1}) \EEc{\|
        \gpop(\empdist{h}, \pi_h)-\EEc{\empdist{h+1}}{\cF_h}
        \|_2^2}{\cF_h} + (1 + \delta_h) \EEc{\|
        \gpop(\empdist{h}, \pi_h)-\mu_{h+1}
        \|_2^2}{\cF_h} \label{ineq:young_squared_pop}
        \\
         \leq  & \underbrace{\EEc{\|
        \empdist{h+1}-\EEc{\empdist{h+1}}{\cF_h}
        \|_2^2}{\cF_h}}_{(\triangle)} + 2(1 + \delta_h^{-1})
        \underbrace{
            \EEc{\| \EEc{\empdist{h+1}}{\cF_h}- \Xi(\vecrho_h, \pi_h) \|_2^2}{\cF_h}
        }_{(\square)} \notag \\
        &+ 2(1 + \delta_h^{-1}) \underbrace{
            \EEc{\|
                \Xi(\vecrho_h, \pi_h)-\gpop(\empdist{h}, \pi_h)
                \|_2^2}{\cF_h}
        }_{(\star)} + (1 + \delta_h) \underbrace{\EEc{\|
        \gpop(\empdist{h}, \pi_h)-\mu_{h+1}
        \|_2^2}{\cF_h} }_{(\heartsuit)} \label{eq:squared_pop_main_inductive_bound_lemma_dist}
\end{align}
where inequalities~\ref{ineq:young_squared_pop} and \ref{eq:squared_pop_main_inductive_bound_lemma_dist} follow from applications of Young's inequality, where $\delta_h > 0$ is a positive value to be determined later.
We upper bound the four terms separately as in the proof of Theorem~\ref{theorem:main_approximation}.
For $(\triangle)$, it holds that
\begin{align*}
    (\triangle) = &\EEc{
    \|
    \empdist{h+1}-
    \EEc{\empdist{h+1}}{\cF_h}
    \|_2^2
    }{\cF_h} \leq \frac{|\setS||\setA|}{N},
\end{align*}
since each $\empdist{h+1}(s)$ is an average of independent subgaussian random variables given $\cF_h$. 
Specifically, each indicator is bounded $\ind{s^i_{h+1}= s, a^i_{h+1}= a}\in [0,1]$ almost surely. 

Next, for the term $(\square)$,
\begin{align*}
    (\square) &= \EEc{\| \EEc{\empdist{h+1}}{\cF_h}- \Xi(\vecrho_h, \refpol_{h+1}) \|_2^2}{\cF_h} \\
    &\begin{aligned}
        \leq \frac{1}{N^2} \Exop\bigg[ &\bigg\|  \sum_{i=1}^N P^i(\cdot|s^i_h,a^i_h,\vecrho^{-i}) \cdot \pi_{h+1} -  \sum_{i=1}^N  P(\cdot|s^i_h,a^i_h,\empc{\vecrho^{-i}}) \cdot \pi_{h+1} \bigg\|_1^2 \bigg| \cF_h \bigg]
    \end{aligned}\\
    &\leq \frac{1}{N^2} \Exop\bigg[ \left(\sum_{i=1}^N \|P^i(\cdot|s^i_h,a^i_h,\vecrho^{-i}) - P(\cdot|s^i_h,a^i_h,\empc{\vecrho^{-i}})\|_1 \right)^2 \bigg| \cF_h \bigg]
\end{align*}
By the $\alpha$-symmetry condition, it follows that $(\square) \leq \alpha^2$.
    
For $(\star)=\Exop[\|\Xi(\vecrho_h, \pi_h \}_i)-\gpop(\empdist{h}, \pi_h)\|_2^2|\cF_h]$,
\begin{align*}
    (\star)
    &\leq \EEc{ \|\Xi(\vecrho_h, \refpol_{h+1})-\gpop(\empdist{h}, \refpol_{h+1})\|_1^2}{\cF_h} \\
    &\begin{aligned}
        = \Exop \bigg[ \bigg\| &\frac{1}{N}\sum_{i=1}^N P(\cdot|s^i_h,a^i_h,\empc{\vecrho^{-i}_h}) \cdot \refpol_{h+1} -
        \sum_{s',a'}\empdist{h}(s',a') P(\cdot|s',a',\empdist{h}) \cdot \refpol_{h+1} \bigg\|_1^2 \bigg| \cF_h \bigg]
    \end{aligned} \\
    &\begin{aligned}
        = \frac{1}{N^2} \Exop \bigg[ \bigg\| &\sum_{i=1}^N P(\cdot|s^i_h,a^i_h,\empc{\vecrho^{-i}_h}) \cdot \refpol_{h+1} -
        \sum_{i=1}^N P(\cdot|s^i_h,a^1_h,\empdist{h}) \cdot \refpol_{h+1} \bigg\|_1 ^2 \bigg| \cF_h \bigg].
    \end{aligned}
\end{align*}
The vectors $(N-1)\empc{\vecrho^{-i}_h}, N \empdist{h}$ can differ by only 1 in one component due to the $i$-th agent being excluded from the former, it holds that
\begin{align*}
    &\|N \empc{\vecrho^{-i}_h} - N \empdist{h} \|_1 \leq \|(N-1) \empc{\vecrho^{-i}_h} - N \empdist{h} \|_1 + \| \empc{\vecrho^{-i}_h} \|_1 \leq 3,
\end{align*}
therefore for any $s,a$,
\begin{align*}
    \|P(\cdot|s,a,\empc{\vecrho^{-i}_h} ) - P(\cdot|s,a,\empdist{h} ) \|_1 \leq \frac{3K_\mu}{N}
\end{align*}
almost surely and $(\star)$ is further upper bounded by $(\star) \leq \frac{9 K_\mu^2}{N^2}.$
    
Finally, the last term $(\heartsuit)$ can be upper bounded using the Lipschitz condition on $\gpop$, namely:
\begin{align*}
    (\heartsuit) = &\EEc{\|
        \gpop(\empdist{h}, \pi_h)-
        \gpop(\mu_{h}^{\pi}, \pi_h)
        \|_2^2}{\cF_h} 
        \leq \lpopmu^2 \|\empdist{h} - \mu_{h} \|_2^2.
\end{align*}
To conclude, merging the bounds on the three terms in Inequality~\eqref{eq:main_inductive_bound_lemma_dist} and taking the expectations we obtain:
\begin{align*}
    \EE{\| \empdist{h+1} - \mu_{h+1}\|_2^2} \leq & (1 + \delta_h) \lpopmu^2 \EE{\| \empdist{h} - \mu_{h}\|_2^2} + \frac{|\setS||\setA|}{N} \\ 
    &+ 2(1 + \delta_h^{-1}) \frac{9K_\mu^2}{N^2} + 2(1 + \delta_h^{-1}) \alpha^2.
\end{align*}

Induction on $h$ yields the bound for all $h$:
\begin{align}
    &\EE{\| \empdist{h} - \mu_{h}\|_2^2} \notag \\
    &\leq \sum_{h'=0}^{h} \lpopmu^{2(h-h')} \left( \prod_{h'' = h'+1}^{h} (1 + \delta_{h''}) \right) \left(\frac{|\setS||\setA|}{N} + 2(1 + \delta_{h'}^{-1})\frac{9K_\mu^2}{N^2} + 2(1 + \delta_{h'}^{-1})\alpha^2\right). \notag 
\end{align}
Here, we specify $\delta_h := (1 + h^2)^{-1}$, for which it holds that 
\begin{align*}
    \prod_{h = 0}^{\infty} (1 + \delta_{h}) = \prod_{h = 0}^{\infty} (1 + (1+h^2)^{-1}) \leq C,
\end{align*}
for an absolute constant $C < 10$.
Then,
\begin{align}
    &\EE{\| \empdist{h} - \mu_{h}\|_2^2} \notag \\
    &\leq \sum_{h'=0}^{h} \lpopmu^{2(h-h')} C \left(\frac{|\setS||\setA|}{N} + 2(H^2 + 2)\frac{9K_\mu^2}{N^2} + 2(H^2 + 2)\alpha^2\right). \notag 
\end{align}
which yields the stated upper bound of the lemma 
\begin{align}
    &\EE{\| \empdist{h} - \mu_{h}\|_2^2} \leq \frac{1 - \lpopmu^{2(h+1)}}{1 - \lpopmu^2} C \left(\frac{|\setS||\setA|}{N} + \frac{18 K_\mu^2 (H^2 + 2)}{N^2} + 2(H^2 + 2)\alpha^2\right),  
\end{align}
where we adopt the convenient shorthand $\frac{1 - \lpopmu^{2(h+1)}}{1 - \lpopmu^2} := h$ if $\lpopmu = 1$.
\end{proof}

Note that while the Lipschitz modulus used in Lemma~\ref{lemma:pop_error_squared} is with respect to the $\|\cdot\|_2$ norm, Lemma~\ref{lemma:extension_bound} readily guarantees that this will hold.

\subsection{Extended Proof of Theorem~\ref{theorem:td}}\label{sec:appendix_td}

Let $\vecmu = \{ \mu_h \}_h:= \lpop(\pi)$, and note that by the proof of Theorem~\ref{theorem:main_approximation}, it holds that (Inequality~\ref{ineq:marg_diff_l1})
\begin{align*}
    A_h & := \| \Prob[s_h^1=\cdot, a_h^1=\cdot] - \Prob[s_h=\cdot, a_h=\cdot]  \|_1 \notag \\
    &\leq \| \Prob[s_h^1=\cdot] - \Prob[s_h=\cdot]  \|_1 \notag \\
    &\leq \lpopmu \sum_{h'=0}^{h-1} \frac{1 - \lpopmu^{h'+1}}{1 - \lpopmu}\left(\frac{|\setS||\setA|}{N} + \frac{5K_\mu}{N} + \alpha\right) + (h-1)\left(\alpha + \frac{3K_\mu}{N}\right)  .
\end{align*}
Likewise, by Lemma~\ref{lemma:pop_error_squared},
\begin{align*}
    B_h := &\Exop[ \| \mu_h - \empc{\vecrho^{-1}_{m,h}}\|_1^2] \\
    := &\Exop[ \| \mu_h - \empc{\vecrho^{-1}_{m,h}}\|_2^2] |\setS||\setA| \\
    \leq & \frac{1 - \lpopmu^{2(h+1)}}{1 - \lpopmu^2} 2C \left(\frac{|\setS||\setA|}{N} + \frac{18 K_\mu^2 (H^2 + 2)}{N^2} + 2(H^2 + 2)\alpha^2\right)|\setS||\setA| + \frac{18 |\setS||\setA|}{N^2}.
\end{align*}    
Will will commonly utilize the bounds $\Qest{m}{h} \in [0,\Qmax ], \Qpi{h} \in [0,\Qmax]$ almost surely for $\Qmax := H (1 + \log |\setA|)$, as the one-step rewards are bounded in range $[0,1]$ and the policy entropy has trivial upper bound $\log |\setA|$.
Denote the marginal probabilities of $s^1_{m,h}, a^1_{m,h}$ (which is i.i.d. for all $m$) as $p_h \in \Delta_{\setS\times\setA}$, which clearly does not depend on epoch $m$ as the same policies are deployed at each TD learning round.
We also define the quantity
\begin{align*}
    \qpi{h}(s,a) := \Qpi{h} (s,a) - \tau \entropy(\pi_h(\cdot|s)),
\end{align*}
which corresponds to the more standard entropy regularized value function in some works.

We outline the proof strategy into different steps as follows:
\begin{itemize}
    \item \textbf{Step 1. } Analyze the algorithm for the $Q$-values at time step $H-1$, that is, show that in expectation $\| \Qpi{H-1} - \Qest{m}{H-1} \|_{p_{H-1}}^2$ decreases with $\cO(\sfrac{1}{m})$ over epochs, up to a small bias term.
    \item \textbf{Step 2. } Assuming that the error at some time $h$ decreases with rate $\cO(\sfrac{1}{m})$, show that the error $\| \Qpi{h-1} - \Qest{m}{h-1} \|_{p_{H-1}}^2$ also decreases with rate $\cO(\sfrac{1}{m})$, showing that the magnification in the constants are not too large.
    \item \textbf{Step 3. } Conclude the statement of the theorem by induction.
\end{itemize}

\textbf{Step 1. }
We will first analyze the evolution of $\Qest{m}{H-1}$.
By definition, it holds that 
\begin{align*}
    \Qpi{H-1}(s, a) = R(s,a,\mu_{H-1}) + \tau \entropy(\pi_{H-1}(s)).
\end{align*}
In other words, there is no bootstrapping and the stochastic error does not have a dependence on future biased estimates.
Firstly, if $s^1_{m,H-1}=s, a^1_{m,H-1}=a, \vecrho_{m,H-1} = \vecrho$, then it holds almost surely that
\begin{align*}
    &\Qpi{H-1}(s,a) - \Qest{m+1}{H-1} (s,a) \\
    &=\Qpi{H-1}(s,a) - \Qest{m}{H-1}(s,a) - \lr{m}{h} ( r_{m,H-1}^i + \tau\entropy(\pi_{H-1}(s)) - \Qest{m}{h}(s,a) )  \\
    &= (1 - \lr{m}{H-1}  ) (\Qpi{H-1}(s,a) - \Qest{m}{H-1}(s,a)) - \lr{m}{h} ( r_{m,H-1}^i + \tau\entropy(\pi_{H-1}(\cdot|s)) - \Qpi{H-1}(s,a) ) \\ 
    &= (1 - \lr{m}{H-1}  ) (\Qpi{H-1}(s,a) - \Qest{m}{H-1}(s,a)) - \lr{m}{h} ( ( R^1(s,a,\vecrho_{m,H-1}^{-1}) - R(s,a,\mu_{H-1}) ) ,
\end{align*}
as the entropy term $\entropy(\pi_{H-1}(\cdot|s))$ cancels out.
Denote the $\sigma$-algebra 
\begin{align*}
    \cF^m_{s,a} := \cF \{ \{s_{m',h}, a_{m',h}\}_{m'<m}, s_{m,H-1}=s, a_{m,H-1}=a \}
\end{align*}
for any fixed $s,a$.
Then, noting that $\Qest{m}{H-1} (s,a)$ is $\cF^m_{s,a}$-measurable, we have the inequalities
\begin{align*}
    &\Exop[ ( \Qpi{H-1}(s, a) - \Qest{m+1}{H-1} (s,a) )^2 | \cF^m_{s,a} ]\\
    &= \Exop\bigg[ \bigg( (1-\lr{m}{H-1}) ( \Qpi{H-1}(s, a) - \Qest{m}{H-1} (s,a) ) - \lr{m}{H-1} ( R^1(s,a,\vecrho_{m,H-1}^{-1}) - R(s,a,\mu_{H-1}) ) \bigg)^2 \bigg| \cF^m_{s,a} \bigg]\\
    &= (1-\lr{m}{H-1})  ^ 2  ( \Qpi{H-1}(s, a) - \Qest{m}{H-1} (s,a) )^2   \\
        &\qquad + 2 (1-\lr{m}{H-1}) \lr{m}{H-1} ( \Qpi{H-1}(s, a) - \Qest{m}{H-1} (s,a) ) \Exop[   R^1(s,a,\vecrho_{m,H-1}^{-1}) - R(s,a,\mu_{H-1})   | \cF^m_{s,a} ] \\
        &\qquad + \lr{m}{H-1}^2 \Exop[ ( R^1(s,a,\vecrho_{m,H-1}^{-1}) - R(s,a,\mu_{H-1}) )^2 | \cF^m_{s,a} ] 
\end{align*}
We use Young's inequality and using the fact that rewards are bounded in $[0,1]$, 
\begin{align*}
    &\Exop[ ( \Qpi{H-1}(s, a) - \Qest{m+1}{H-1} (s,a) )^2 | \cF^m_{s,a} ]\\
    &\leq (1-2\lr{m}{H-1})   ( \Qpi{H-1}(s, a) - \Qest{m}{H-1} (s,a) )^2  + \lr{m}{H-1}  ( \Qpi{H-1}(s, a) - \Qest{m}{H-1} (s,a) )^2 \\
        &\qquad+\lr{m}{H-1}\Exop[  ( R^1(s,a,\vecrho_{m,H-1}^{-1}) - R(s,a,\mu_{H-1}) )^2   | \cF^m _{s,a} ] + (1 + \Qmax^2)\lr{m}{H-1}^2 \\
    &\leq (1-\lr{m}{H-1})  ( \Qpi{H-1}(s, a) - \Qest{m}{H-1} (s,a) )^2 +\lr{m}{H-1}\Exop[ 2\beta^2 + 2 L_\mu^2\|\empc{\vecrho_{m,H-1}^{-1}} -\mu_{H-1}\|_2^2 | \cF^m_{s,a} ] \\
        & \qquad + (1 + \Qmax^2)\lr{m}{H-1}^2 
\end{align*}
We then take expectations and use the law of total expectation to obtain the bound:
\begin{align}
    &\Exop[ ( \Qpi{H-1}(s,a) - \Qest{m+1}{H-1} (s,a) )^2  ] \notag \\
    &\leq (1-\lr{m}{H-1}) p_{H-1}(s,a) \Exop[ ( \Qpi{H-1}(s,a) - \Qest{m}{H-1} (s,a))^2 ]  \notag \\
        &\qquad + 2\lr{m}{H-1} p_{H-1}(s,a) \Exop[  \beta^2 +  L_\mu^2\|\empc{\vecrho_{m,H-1}^{-1}} -\mu_{H-1}\|_2^2  | s_{m,H-1}=s, a_{m,H-1}=a ] \notag \\
        &\qquad + p_{H-1}(s,a) (1 + \Qmax^2) \lr{m}{H-1}^2  \notag \\
        &\qquad + (1-p_{H-1}(s,a)) \Exop[ ( \Qpi{H-1}(s,a) - \Qest{m+1}{H-1} (s,a) )^2 ] \notag \\ 
    &\leq (1-p_{H-1}(s,a)\lr{m}{H-1})  \Exop[ ( \Qpi{H-1}(s,a) - \Qest{m}{H-1} (s,a))^2 ]  \notag \\
        &\qquad + \lr{m}{H-1} p_{H-1}(s,a) \Exop[  2\beta^2 + 2 L_\mu^2\|\empc{\vecrho_{m,H-1}^{-1}} -\mu_{H-1}\|_2^2  | s_{m,H-1}=s, a_{m,H-1}=a ] \notag \\
        &\qquad + p_{H-1}(s,a) (1 + \Qmax^2) \lr{m}{H-1}^2  \notag
\end{align}
Summing this inequality over all state-action pairs with weight $p_{H-1}$, we obtain
\begin{align*}
    &\Exop[ \| \Qpi{H-1} - \Qest{m+1}{H-1} \|^2_{p_{H-1}}  ] \\
        &\leq (1-\delta\lr{m}{H-1})  \Exop[ \| \Qpi{H-1} - \Qest{m}{H-1}\|_{p_{H-1}}^2 ] \\
            &\qquad + \sum_{s,a} \lr{m}{H-1} p_{H-1}(s,a) \Exop[  2\beta^2 + 2 L_\mu^2\|\empc{\vecrho_{m,H-1}^{-1}} -\mu_{H-1}\|_2^2    | s_{m,H-1}=s, a_{m,H-1}=a ] \\
            &\qquad + (1 + \Qmax^2) \lr{m}{H-1}^2 \\
        &= (1-\delta\lr{m}{H-1})  \Exop[ \| \Qpi{H-1} - \Qest{m}{H-1}\|_{p_{H-1}}^2 ] + (1 + \Qmax^2) \lr{m}{H-1}^2  \\
            &\qquad +  2\lr{m}{H-1} \beta^2 +  2\lr{m}{H-1} L_\mu^2 \Exop[ \| \mu_{H-1} - \empc{\vecrho_{m,H-1}^{-1}}\|_2^2   ] \\
        &\leq (1-\delta\lr{m}{H-1})  \Exop[ \| \Qpi{H-1} - \Qest{m}{H-1}\|_{p_{H-1}}^2 ] + (1 + \Qmax^2) \lr{m}{H-1}^2 +  2\lr{m}{H-1} (\beta^2 + L_\mu^2 B_{H-1}).
\end{align*}

We expand this recursive inequality as follows. 
Define the shorthand notation $\Pi_{m}^{m'} := \prod_{k=m}^{m'} (1 - \delta \lr{k}{H-1})$.
Then, for any $M>0$,
\begin{align}
    &\Exop[\| \Qpi{H-1} - \Qest{M}{H-1} \|^2_{p_{H-1}} ] \notag \\
    &\leq \Pi_{0}^{M-1} + (1 + 2\Qmax^2) \sum_{m=0}^{M-1} \lr{m}{H-1}^2 \Pi_{m+1}^{M-1} + 2\sum_{m=0}^{M-1} \lr{m}{H-1} (\beta^2 + L_\mu^2 B_{H-1}) \Pi_{m+1}^{M-1} \notag \\
    &\leq \Pi_{0}^{M-1} + (1 + 2\Qmax^2)\lr{M-1}{H-1}^2 + 2\lr{M-1}{H-1} (\beta^2 + L_\mu^2 B_{H-1}) \notag \\
        &\qquad + (1 + 2\Qmax^2) \sum_{m=0}^{M-2} \lr{m}{H-1}^2 \Pi_{m+1}^{M-1} 
        + 2\sum_{m=0}^{M-2} \lr{m}{H-1} (\beta^2 + L_\mu^2 B_{H-1}) \Pi_{m+1}^{M-1} .
        \label{ineq:Mexpanded}
\end{align}
We bound the multiplicative terms $\Pi_{m}^{m'}$.
Assuming that $\lr{m}{H-1}$ is of the form $\lr{m}{H-1} = \frac{u}{v + m}$, for any $m \leq m'$, we have that
\begin{align*}
    \Pi_{m}^{m'} = &\prod_{k=m}^{m'} (1 -  \delta \lr{k}{H-1}) \leq \exp\{ - \delta \sum_{k=m}^{m'} \lr{k}{H-1}) \} \\
        \leq & \exp\{ -\delta\sum_{k=m}^{m'} \frac{u}{v + k} ) \} \leq  \exp\{ - \delta u \log \frac{m' + v }{ m + v - 1}  \} \\
        \leq &\left( \frac{m + v - 1  }{ m' + v } \right)^{\delta u} 
\end{align*}
using Lemma~\ref{lemma:harmonic}.
Taking the values $u=v=2\delta^{-1}$, this reduces to
\begin{align*}
    \Pi_{m}^{m'} \leq \left( \frac{m + u - 1  }{ m' + u } \right)^2,
\end{align*}
Placing this in Inequality~\ref{ineq:Mexpanded} for the two terms appearing $\Pi_{m+1}^{M-1}, \Pi_{0}^{M-1}$, we obtain
\begin{align*}
     &\Exop[ \| \Qpi{H-1} - \Qest{M}{H-1} \|^2_{p_{H+1}}  ]  \\
    &\leq \left( \frac{u - 1  }{ M + u - 1 } \right)^2 + (1 + 2\Qmax^2)\left(\frac{u}{M+u-1}\right)^2 + 2\left(\frac{u}{M+u-1}\right) (\beta^2 + L_\mu^2 B_{H-1})  \\
        &\qquad + (1 + 2\Qmax^2)\sum_{m=0}^{M-2} \left( \frac{u}{m+u}\right)^2 \left( \frac{m + u }{ M + u - 1} \right)^2 \\
        &\qquad + 2(\beta^2 + L_\mu^2 B_{H-1}) \sum_{m=0}^{M-2} \left( \frac{u}{m+v}\right) \left( \frac{m +u }{ M + u - 1} \right)^2 \\
    &\leq \frac{C_1 u}{(M + u - 1)} + C_2 u (\beta^2 + L_\mu^2 B_{H-1}),
\end{align*}
for some absolute constants $C_1, C_2$.
This inequality concludes the convergence result for the $Q$ values at time step $H-1$, showing a rate of convergence $\cO(\sfrac{1}{M})$ over $M$ epochs up to a bias of $\cO(\beta^2 + \alpha^2 + \sfrac{1}{N})$.

\textbf{Step 2. }
Next, we analyze the case $h < H-1$.
Again under the observation that if $s^1_{m,h}=s, a^1_{m,h}=a, s^1_{m,h+1}=s', a^1_{m,h+1}=a'$, then it holds almost surely that
\begin{align*}
    &\Qpi{h}(s,a) - \Qest{m+1}{h} (s,a) \\
    &=\Qpi{h}(s,a) - \Qest{m}{h}(s,a) - \lr{m}{h} ( \Qest{m}{h+1}(s',a') + r_{m,h}^1 + \tau\entropy( \pi_h(\cdot|s)) - \Qest{m}{h}(s,a) )  \\
    &= (1 - \lr{m}{h}  ) (\Qpi{h}(s,a) - \Qest{m}{h}(s,a)) - \lr{m}{h} ( \Qest{m}{h+1}(s',a')  + r_{m,h}^1 + \tau\entropy( \pi_h(\cdot|s)) - \Qpi{h}(s,a) ), \\
    &= (1 - \lr{m}{h}  ) (\Qpi{h}(s,a) - \Qest{m}{h}(s,a)) - \lr{m}{h} ( \Qest{m}{h+1}(s',a')  + r_{m,h}^1  - \qpi{h}(s,a) ), 
\end{align*}
since again as the entropy terms $\entropy(\pi_h (\cdot|s))$ cancel.
Defining the induced $\sigma$-algebra
\begin{align*}
    \cF^m_h := \cF \{ \{s_{m',h}, a_{m',h}\}_{m'<m}, s_{m,h}=s, a_{m,h}=a\}.
\end{align*}
Note that with respect to $\cF^m_h$, $\Qest{m}{h'}$ is measurable for any $h'$ as $\Qest{m}{h'}$ only depends on episodes previous.
we have the upper bound:
\begin{align}
    &\Exop[ ( \Qpi{h}(s, a) - \Qest{m+1}{h} (s,a) )^2 | \cF^m_h ] \notag \\
    &= (1 - \lr{m}{h}  )^2 \Exop[ (\Qpi{h}(s,a) - \Qest{m}{h}(s,a))^2 | \cF^m_h ] \notag \\
        & \qquad + (\lr{m}{h})^2 \Exop[ ( \Qest{m}{h+1}(s',a')  + r_{m,h}^1 - \qpi{h}(s,a) )^2| \cF^m_h ] \notag \\
        &\qquad - 2 \lr{m}{h}  (\Qpi{h}(s,a) - \Qest{m}{h}(s,a))  \Exop[ \underbrace{\Qest{m}{h+1}(s^1_{m,h+1}, a^1_{m,h+1})  + r_{m,h}^1 - \qpi{h}(s,a) }_{\triangle} | \cF^m_h ]  \notag \\
    &= (1 - \lr{m}{h}  )^2 (\Qpi{h}(s,a) - \Qest{m}{h}(s,a))^2  + \lr{m}{h}^2 \Qmax^2 \notag \\
        &\qquad - 2 \lr{m}{h} (\Qpi{h}(s,a) - \Qest{m}{h}(s,a)) \Exop[ \triangle | \cF^m_h ] , \notag \\
    &\leq (1 - 2\lr{m}{h}  ) (\Qpi{h}(s,a) - \Qest{m}{h}(s,a))^2  + 2\lr{m}{h}^2 \Qmax^2 \notag \\
        &\qquad + \lr{m}{h} (\Qpi{h}(s,a) - \Qest{m}{h}(s,a))^2 +  \lr{m}{h}\Exop[ \triangle | \cF^m_h ]^2, \notag \\
    &\leq (1 - \lr{m}{h}  ) (\Qpi{h}(s,a) - \Qest{m}{h}(s,a))^2  + 2\lr{m}{h}^2 \Qmax^2 +  \lr{m}{h}\Exop[ \triangle | \cF^m_h ]^2, \label{ineq:qpimh_inter}
\end{align}
as $\Qest{m}{h}$ is $\cF_h^m$ measurable, using Young's inequality in the last line.
The last bias term due to bootstrapping and finite population bias we bound separately.
We decompose $(\triangle)$ as follows using Young's inequality.
\begin{align*}
\Exop[ \triangle | \cF^m_h ]^2 \leq &\left(1+\frac{1}{(H-h+1)^2}\right) \Exop[ \Qest{m}{h+1}(s',a') - \Qpi{h+1}(s',a') \, | \cF^m_h ]^2  \\
        & + 2 (H-h+1)^2 \big| \Exop[ \Qpi{h+1}(s',a')  \, | \cF^m_h ] - \sum_{\widebar{s},\widebar{a}}P(\widebar{s}, \widebar{a}, \mu_h) \Qpi{h+1}(\widebar{s}, \widebar{a}) \big|^2 \\
        & + 2 (H-h+1)^2 | \Exop[ r_{m,h}^1 - R(s,a,\mu_h)  \, | \cF^m_h ] |^2
\end{align*}
The three terms are upper-bounded by the inequalities in expectation:
\begin{align*}
     &| \Exop[ \Qest{m}{h+1}(s^1_{m,h+1}, a^1_{m,h+1}) - \Qpi{h+1}(s^1_{m,h+1}, a^1_{m,h+1}) \, | \cF^m_h ] |^2  \\
        &\qquad \leq  \Exop[| \Qest{m}{h+1}(s^1_{m,h+1}, a^1_{m,h+1}) - \Qpi{h+1}(s^1_{m,h+1}, a^1_{m,h+1})|^2 | \cF^m_h  ] \\
        &\qquad \leq  \| \Qest{m}{h+1} - \Qpi{h+1} \|_{p_{h+1}(\cdot|s,a)}^2  \\
     &  \bigg|\Exop[\Qpi{h+1}(s^1_{m,h+1}, a^1_{m,h+1})  \, | \cF^m_h ] -  \sum_{\widebar{s},\widebar{a}}P(\widebar{s}, \widebar{a}|s,a ,\mu_h) \Qpi{h+1}(\widebar{s}, \widebar{a}) \bigg|^2  \\
        & \qquad \leq \frac{\Qmax^2}{4} \Exop[ 2\alpha^2 + 2K_\mu ^2 \| \mu_h - \empc{\vecrho^{-1}_{m,h}} \|_1^2  | \cF^m_h ]  \\
    &  \bigg| \Exop[ r_{m,h}^1 - R(s,a,\mu)  \, | \cF^m_h ] \bigg|^2  \\
        & \qquad = 2| \Exop[  R^1(s,a,\vecrho^{-1}_{m,h}) - R(s,a,\empc{\vecrho^{-1}_{m,h}}) \, | \cF^m_h ] |^2 + 2| \Exop[  R(s,a,\empc{\vecrho^{-1}_{m,h}}) - R(s,a,\mu_h)  \, | \cF^m_h ] |^2 \\
        &\qquad \leq 2\beta^2 + 2L_\mu^2 \Exop[ \| \mu_h - \empc{\vecrho^{-1}_{m,h}} \|_1^2 | \cF^m_h] 
\end{align*}
Therefore, we conclude by an application of Young's inequality that almost surely,
\begin{align*}
    \Exop[ \triangle | \cF^m_h ]^2 \leq &\left(1+\frac{1}{(H-h+1)^2}\right) \| \Qest{m}{h+1} - \Qpi{h+1} \|_{p_{h+1}(\cdot|s,a)}^2  \\
        & + (H-h+1)^2 \Qmax^2 [\alpha^2 + \beta^2 + (2K_\mu^2 + 4L_\mu^2) \Exop[ \| \mu_h - \empc{\vecrho^{-1}_{m,h}} \|_1^2 | \cF^m_h]  ].
\end{align*}

We place this result in Inequality~\ref{ineq:qpimh_inter} to obtain
\begin{align*}
    \Exop[ &( \Qpi{h}(s, a) - \Qest{m+1}{h} (s,a) )^2 | \cF^m_h ] \\
        \leq &(1 - \lr{m}{h}  ) (\Qpi{h}(s,a) - \Qest{m}{h}(s,a))^2  + 2\lr{m}{h}^2 \Qmax^2 \\
        & +  \lr{m}{h} \left(1+\frac{1}{(H-h+1)^2}\right) \| \Qest{m}{h+1} - \Qpi{h+1} \|_{p_{h+1}(\cdot|s,a)}^2  \\
        & + \lr{m}{h} (H-h+1)^2 \Qmax^2 [\alpha^2 + \beta^2 + (2K_\mu^2 + 4L_\mu^2) \Exop[ \| \mu_h - \empc{\vecrho^{-1}_{m,h}} \|_1^2 | s_{m,h}=s, a_{m,h}=a]  ],
\end{align*}
as $ \| \mu_h - \empc{\vecrho^{-1}_{m,h}} \|_1^2 $ is independent of the estimates.
Then, taking expectations on both sides for any $p_h(s,a) > 0$, and noting that $p_h(s,a) \geq \delta$,
\begin{align*}
    \Exop&\left[ \left( \Qpi{h}(s,a) - \Qest{m+1}{h} (s,a) \right)^2  \right] \notag \\
        \leq & (1 - \delta \lr{m}{h}  ) \Exop[\Qpi{h}(s,a) - \Qest{m}{h}(s,a))^2 ] + 2\lr{m}{h}^2 \Qmax^2 \\
            & + \lr{m}{h} \left(1+\frac{1}{(H-h+1)^2}\right) \Exop [ \| \Qest{m}{h+1} - \Qpi{h+1} \|_{p_{h+1}}^2 ] \\
            & + \lr{m}{h} (H-h+1)^2 \Qmax^2 (\alpha^2 + \beta^2  ) \\
            & + \lr{m}{h} (H-h+1)^2 \Qmax^2 (2K_\mu^2 + 4L_\mu^2)\Exop[ \| \mu_h - \empc{\vecrho^{-1}_{m,h}} \|_1^2 | s_{m,h}=s, a_{m,h}=a ]
\end{align*}
Finally, taking a weighted sum of both sides of the inequality over $s,a$ with weights $p_{h}$,
\begin{align*}
    \Exop\left[ \| \Qpi{h} - \Qest{m+1}{h} \|_{p_{h}}^2  \right] \leq & (1 - \delta \lr{m}{h}  ) \Exop[\| \Qpi{h} - \Qest{m+1}{h} \|_{p_{h}}^2 ] + 2\lr{m}{h}^2 \Qmax^2 \\
            & + \lr{m}{h} \left(1+\frac{1}{(H-h+1)^2}\right) \Exop [ \| \Qest{m}{h+1} - \Qpi{h+1} \|_{p_{h+1}}^2 ] \\
            & + \lr{m}{h} (H-h+1)^2 \Qmax^2 [\alpha^2 + \beta^2 + (2K_\mu^2 + 4L_\mu^2) B_h  ]
\end{align*}
Expanding this recursive inequality and using the same notation as in Step 1 for the multiplicative terms, also taking the inductive assumption that $\Exop [ \| \Qest{m}{h+1} - \Qpi{h+1} \|_{p_{h+1}}^2 ] \leq \frac{G_1}{2\delta^{-1} + m - 1} + G_2$ 
for some $G_1,G_2$ which depends on problem parameters but not on $m$, we have the final inequality:
\begin{align*}
    \Exop\left[ \| \Qpi{h} - \Qest{m+1}{h} \|_{p_{h}}^2  \right] \leq & \Pi_{m=0}^{M-1} \Exop[\| \Qpi{h} - \Qest{0}{h} \|_{p_{h}}^2 ] + \frac{C_3 u}{M-2 + v} \\
            & + 2\sum_{m=0}^{M-2}(\lr{m}{h})^2 \Qmax^2 \Pi_{m+1}^{M-1} \\
            & + \sum_{m=0}^{M-2}\lr{m}{h}(H-h+1)^2 \Qmax^2 [\alpha^2 + \beta^2 + (2K_\mu^2 + 4L_\mu^2) B_h  ] \Pi_{m+1}^{M-1} \\
            & + \sum_{m=0}^{M-2} \lr{m}{h} \left(1+\frac{1}{(H-h+1)^2}\right) \left(\frac{G_1}{2\delta^{-1} + m - 1} + G_2\right) \Pi_{m+1}^{M-1}
\end{align*}
Once again as in Step 1, fixing the values $u=v=2\delta^{-1}$,
\begin{align*}
    \Exop\left[ \| \Qpi{h} - \Qest{m+1}{h} \|_{p_{h}}^2  \right] \leq & \Qmax^2 \left( \frac{ u - 1  }{ M + u - 1 } \right)^2 + \frac{C_3 u}{M-2 + v} + \sum_{m=0}^{M-2}\Qmax^2 \frac{ 8\delta^{-2} }{ (M + u - 1)^2}  \\
            & + \sum_{m=0}^{M-2} (H-h+1)^2 \Qmax^2 [\alpha^2 + \beta^2 + (2K_\mu^2 + 4L_\mu^2) B_h  ]  \frac{ 2(m +u) \delta^{-1} }{ (M + u - 1)^2} \\
            & + \sum_{m=0}^{M-2} \left(1+\frac{1}{(H-h+1)^2}\right) \left(\frac{G_1}{2\delta^{-1} + m - 1} + G_2\right)  \frac{2(m +u) \delta^{-1} }{ (M + u - 1)^2 } 
\end{align*}
Using the fact that $\sum_{m=1}^M (m+u) \approx M^2 $ and $\sum_{m=1}^M c = cM$, for some absolute constants we have that
\begin{align}
    \Exop\left[ \| \Qpi{h} - \Qest{m+1}{h} \|_{p_{h}}^2  \right] \leq &  \left( 2\frac{ \Qmax^2 \delta^{-1}  }{ M + 2\delta^{-1} - 1 } \right)^2 + \frac{C_3 \delta^{-1}}{M - 2 + 2\delta^{-1}} + \frac{  C_4 \Qmax^2 \delta^{-2} }{ M + \delta^{-1} - 1} \notag \\
            & + C_5 (H-h+1)^2 \Qmax^2 [\alpha^2 + \beta^2 + (2K_\mu^2 + 4L_\mu^2) B_h  ]  \delta^{-1} \notag \\
            & + \left(1+\frac{1}{(H-h+1)^2}\right) \left(\frac{C_6 \delta^{-1} G_1}{2\delta^{-1} + m - 1} + C_7 G_2 \delta^{-1} \right) \label{ineq:inductive_result}
\end{align}

\textbf{Step 3. } 
Finally, we conclude with the proof using Steps 1 and 2.
By using Inequality~\ref{ineq:inductive_result}, it readily follows that 
\begin{align*}
    \Exop\left[ \| \Qpi{h} - \Qest{m+1}{h} \|_{p_{h}}^2 \right] = \cO\left(\frac{1}{M} + \alpha^2 + \beta^2 + \frac{1}{N}\right),
\end{align*}
for all $h$, as the bound in Step 1 established the rate for time step $H-1$.
We comment on the constants: Inequality~\ref{ineq:inductive_result} shows that in the worst case, there might be an exponential dependence on $H$, which might be fundamental.

%% file: appendix/monotonicity.tex
\section{Extended Results on Monotonicity and Learning NE}

\subsection{Example: Asymmetric Congestion Games}\label{sec:appendix_monotone_congestion}

Note that by symmetry in arguments, it follows that $\symmetrization{R^i(s,a,\cdot)} = R^i(s,a,\cdot)$ for any $s,a$, as
\begin{align*}
    \symmetrization{R^i(s,a,\cdot)}(\vecrho) = &\frac{1}{(N-1)!} \sum_{f\in\Sym_{N-1}} R^i(s,a,g(\vecrho)) \\
    = &\frac{1}{(N-1)!} \sum_{f\in\Sym_{N-1}} \left( h_i(s,a,\sum_{j=1}^N \ind{g(\vecrho)_j = (s,a)}) + r_i(s,a) \right)\\
    = & h_i\left(s,a,\sum_{j=1}^N \ind{\rho_j = (s,a)} \right) + r_i(s,a) \\
    = &R^i(s,a,\vecrho)
\end{align*}
By simple computation, the population lifted rewards $\symmbar{R^i(s,a\cdot)}$ are given by
\begin{align*}
    \symmbar{R^i(s,a\cdot)}(\mu) = h_i(s,a,N\mu(s,a)) + r_i(s,a), \quad \forall \mu\in \Delta_{\setS\times\setA, N-1}.
\end{align*}
We provide an extension to the continuum $\Delta_{\setS\times\setA}$ via linear interpolation in this case, while many other alternatives are possible.
Take the function $\widetilde{h}_i:\setS\times\setA\times[0,1] \rightarrow [0,1]$ such that
\begin{align*}
    \widetilde{h}_i(s,a,u) := (Nu - \lfloor N u \rfloor) h_i(s,a,\lfloor Nu \rfloor) + (\lceil Nu \rceil - Nu) h_i(s,a,\lceil Nu \rceil ). 
\end{align*}
The function is clearly monotonically decreasing in $u$.
Furthermore, it is also Lipschitz continuous in $u$, as for any $u_1 > u_2$,
\begin{align*}
    \widetilde{h}_i(s,a,u_1) - \widetilde{h}_i(s,a,u_2) \leq |u_1 - u_2|.
\end{align*}

Finally, the asymmetry due to rewards can be upper bounded by
\begin{align*}
    \beta \leq \sup_{s,a} \sup_{i,j} \sup_{k\in[N]} |h_i(s,a, k) - h_j(s,a, k)|.
\end{align*}

\subsection{Preliminaries for Learning Regularized Monotone MFG}

We present several results required to establish convergence under monotonicity.
We define the (entropy regularized) MFG value functions for an arbitrary population flow $\vecmu \in \Delta_{\setS\times\setA}$ and policy $\pi\in\Pi$ as
\begin{align*}
    \Vfin^\tau_h  \left( s | \vecmu, \pi \right) & := \Exop \left[ \sum_{h'=h}^{H-1} R(s_{h'}, a_{h'}, \mu_{h'}) + \tau  \entropy (\pi_{h'}(\cdot|s_{h'})) \middle| \substack{s_0 = s , \quad a_{h'} \sim \pi_{h'}(s_{h'})\\ s_{h'+1} \sim P(s_{h'}, a_{h'}, \mu_{h'})} \right] \\
    Q^\tau_h (s,a |\vecmu, \pi) &:= \Exop \left[ \sum_{h'=h}^{H-1} R(s_{h'}, a_{h'}, \mu_{h'}) + \tau \entropy( \pi_{h'}(\cdot|s_{h'})) \middle| \substack{s_h = s, \, a_h = a , \, s_{h'+1} \sim P(s_{h'}, a_{h'}, \mu_{h'}), \\   a_{h'} \sim \pi_{h'+1}(s_{h'+1}), \forall h' \geq h} \right].
\end{align*}
We define the regularized value function of the game similarly:
\begin{align*}
    \Vfin^\tau \left( \vecmu, \pi \right) & := \Exop \left[ \sum_{h=0}^{H-1} R(s_h, a_h, \mu_h) + \tau \entropy( \pi_{h}(\cdot|s_{h}))\middle| \substack{s_0 \sim \initpop , \quad a_h \sim \pi_h(s_h)\\ s_{h+1} \sim P(s_h, a_h, \mu_h)} \right].
\end{align*}
As expected, with these definitions it holds that
\begin{align*}
    \Vfin^\tau_h  \left( s | \vecmu, \pi \right) &= \sum_{a\in\setA} \pi_h(a|s) Q^\tau_h(s,a|\vecmu, \pi), \\
    \Vfin^\tau \left( \vecmu, \pi \right) &= \sum_{s\in\setS} \initpop(s) \Vfin^\tau_0  \left( s | \vecmu, \pi \right).
\end{align*}
We also define the quantity
\begin{align*}
    q^\tau_{h}(s,a|\vecmu,\pi) := Q^\tau_h (s,a |\vecmu, \pi) - \tau \entropy(\pi_h(\cdot|s)),
\end{align*}
which corresponds to the more standard entropy regularized value function.
Firstly, we provide several useful lemmas and definitions.

\begin{definition}[Entropy regularized MFG-NE]\label{def:entropy_reg_mfg_ne}
    For a given MFG $(\setS, \setA, \initpop, H, P, R)$, a policy $\pi_\tau^*\in\Pi$ is called the $\tau$-entropy regularized MFG-NE if it holds that
    \begin{align}
        \max_{\pi' \in \Pi} \Vfin^\tau( \lpop (\pi_\tau^*), \pi') - \Vfin^\tau ( \lpop ( \pi_\tau^*), \pi_\tau^* ). \tag{Regularized MFG-NE}
        \label{eq:def_tau_mfg_ne}
    \end{align}
\end{definition}

While entropy regularization will enable the convergence of our algorithm, it will also introduce a bias in terms of the original (unregularized) MFG.
The next lemma quantifies this bias.

\begin{lemma}[Regularization bias]\label{lemma:reg_bias}    Let $\vecmu \in \Delta_{\setS\times\setA}$ and policy $\pi\in\Pi$ be arbitrary.
Then, it holds that
\begin{align*}
    |\Vfin^\tau \left( \vecmu, \pi \right) - \Vfin(\vecmu, \pi)| \leq \tau H \log |\setA|.
\end{align*}
Furthermore, if $\pi^*_\tau$ is a $\tau$-entropy regularized MFG-NE, then it is a $2\tau H \log |\setA|$-MFG-NE, that is,
\begin{align*}
    \Expfin (\pi^*_\tau) \leq 2 \tau H \log |\setA|.
\end{align*}
\end{lemma}
\begin{proof}
    \begin{align*}
        |\Vfin^\tau \left( \vecmu, \pi \right) - \Vfin(\vecmu, \pi)| = &\left| \Exop \left[ \sum_{h=0}^{H-1} \tau \entropy(\pi(\cdot|s_h)) \middle| \substack{s_0 \sim \initpop , \quad a_h \sim \pi_h(s_h)\\ s_{h+1} \sim P(s_h, a_h, \mu_h)} \right] \right| \\
        \leq & \Exop \left[ \sum_{h=0}^{H-1} \tau\left| \entropy(\pi(\cdot|s_h)) \right| \middle| \substack{s_0 \sim \initpop , \quad a_h \sim \pi_h(s_h)\\ s_{h+1} \sim P(s_h, a_h, \mu_h)} \right] 
        \leq & H \tau \log |\setA|,
    \end{align*}
    since entropy is upper bounded by $|\entropy(\pi(\cdot|s_h))| \leq \log |\setA|$.
    The bound for exploitability follows from:
    \begin{align*}
        \Expfin (\pi^*_\tau) = &\max_{\pi' \in \Pi} \Vfin( \lpop (\pi^*_\tau), \pi') - \Vfin ( \lpop ( \pi^*_\tau), \pi^*_\tau ) \\
        = &\max_{\pi' \in \Pi} \Vfin( \lpop (\pi^*_\tau), \pi') - \Vfin^\tau( \lpop (\pi^*_\tau), \pi') + \Vfin^\tau( \lpop (\pi^*_\tau), \pi') - \Vfin ( \lpop ( \pi^*_\tau), \pi^*_\tau ) \\
        \leq & \tau H \log|\setA| + \max_{\pi' \in \Pi} \Vfin^\tau( \lpop (\pi^*_\tau), \pi') - \Vfin ( \lpop ( \pi^*_\tau), \pi^*_\tau ) \\
        \leq & \tau H \log|\setA| + \max_{\pi' \in \Pi} \Vfin^\tau( \lpop (\pi^*_\tau), \pi') - \Vfin^\tau( \lpop (\pi^*_\tau), \pi^*_\tau) + \Vfin^\tau( \lpop (\pi^*_\tau), \pi^*_\tau) - \Vfin ( \lpop ( \pi^*_\tau), \pi^*_\tau ) \\
        \leq & 2\tau H \log|\setA| + \max_{\pi' \in \Pi} \Vfin^\tau( \lpop (\pi^*_\tau), \pi^*_\tau) - \Vfin^\tau ( \lpop ( \pi^*_\tau), \pi^*_\tau ) \\
        = & 2\tau H \log|\setA|.
    \end{align*}
\end{proof}

We note that in our setting, a monotone MFG exhibits a unique MFG-NE, in fact, a unique regularized MFG-NE for any value of $\tau$ \cite{zhang2023learning}.
The above lemma shows that the bias introduced due to entropy regularization is of the order $\mathcal{O}(\tau)$ as expected.

Finally, we state several standard facts about monotone MFG, adapted from various works \cite{zhang2023learning, perrin2020fictitious, perolat2022scaling}.

\begin{lemma}[Monotone improvement lemma]\label{ref:monotone_improvement_lemma}
    Let $\vecmu, \widetilde{\vecmu} \in \Delta_{\setS\times\setA} \in \Delta_{\setS,\setA}^H$ which are induced by policies $\pi, \widetilde{\pi} \in \Pi$ respectively, that is $\lpop(\widetilde{\pi}) = \widetilde{\vecmu}$ and $\lpop({\pi}) = {\vecmu}$.
    If the MFG is monotone, that is, if Definition~\ref{def:monotone_mfg} is satisfied, then it holds that:
    \begin{align*}
        \Vfin^\tau\left( \vecmu, \pi \right)+\Vfin^\tau \left(\widetilde{\vecmu}, \widetilde{\pi} \right)-\Vfin^\tau\left(\vecmu, \widetilde{\pi}\right)-\Vfin^\tau\left( \widetilde{\vecmu}, \pi \right) \leq 0.
    \end{align*}
\end{lemma}
\begin{proof}
Let $\vecmu = \{ \mu_h \}_{h=0}^{H-1}, \widetilde{\vecmu} = \{ \widetilde{\mu}_h \}_{h=0}^{H-1}$.
The proof follows \cite{zhang2023learning}, except for the absence of a graphon.
For monotone MFG, due to the assumption that $P$ does not depend on $\mu$, it holds that
\begin{align*}
    \Vfin^\tau\left( \vecmu, \pi \right) - \Vfin^\tau\left( \widetilde{\vecmu}, \pi \right) = \Vfin\left( \vecmu, \pi \right) - \Vfin\left( \widetilde{\vecmu}, \pi \right) , \\
    \Vfin^\tau\left( \widetilde{\vecmu}, \widetilde{\pi} \right) - \Vfin^\tau\left( \vecmu, \widetilde{\pi} \right) = \Vfin\left( \widetilde{\vecmu}, \widetilde{\pi} \right) - \Vfin \left( \vecmu, \widetilde{\pi} \right) . 
\end{align*}
Furthermore, it holds that
\begin{align*}
    \Vfin\left( \widetilde{\vecmu}, \widetilde{\pi} \right) - \Vfin \left( \vecmu, \widetilde{\pi} \right) = \sum_h\sum_{s,a}\widetilde{\mu}_h(s,a) (R(s,a,\widetilde{\mu}_h) - R(s,a,\mu_h)), \\
    \Vfin\left( {\vecmu}, {\pi} \right) - \Vfin \left( \widetilde{\vecmu}, {\pi} \right) = \sum_h \sum_{s,a}\mu_h(s,a) (R(s,a,\mu_h) - R(s,a,\widetilde{\mu}_h)). 
\end{align*}
Then, using the monotonicity assumption on the rewards, it holds that
\begin{align*}
    &\Vfin^\tau\left( \vecmu, \pi \right)+\Vfin^\tau \left(\widetilde{\vecmu}, \widetilde{\pi} \right)-\Vfin^\tau\left(\vecmu, \widetilde{\pi}\right)-\Vfin^\tau\left( \widetilde{\vecmu}, \pi \right) \\
    &= \Vfin\left( \vecmu, \pi \right)+\Vfin \left(\widetilde{\vecmu}, \widetilde{\pi} \right)-\Vfin\left(\vecmu, \widetilde{\pi}\right)-\Vfin\left( \widetilde{\vecmu}, \pi \right) \\
    &= \sum_h\sum_{s,a}\widetilde{\mu}_h(s,a) (R(s,a,\widetilde{\mu}_h) - R(s,a,\mu_h)) + \sum_h \sum_{s,a}\mu_h(s,a) (R(s,a,\mu_h) - R(s,a,\widetilde{\mu}_h)) \\
    & = \sum_h \sum_{s,a}(\widetilde{\mu}_h(s,a) - \mu_h(s,a) )(R(s,a,\widetilde{\mu}_h) - R(s,a,\mu_h)) \leq 0.
\end{align*}
\end{proof}

The next lemma is simply an adaptation of the standard MFG performance difference lemma in single-agent RL to the MFG setting.

\begin{lemma}[Performance difference lemma]\label{lemma:perf_diff}
    For an arbitrary MFG, let $\pi, \widetilde{\pi}\in\Pi$ and $\vecmu = \lpop(\pi)$.
    \begin{align*}
        & V_0^{\tau}\left(s | \vecmu, \widetilde{\pi}\right)-V_0^{\tau}\left(s| \vecmu, \pi \right)+\tau \mathbb{E}_{\widetilde{\pi}, \vecmu} \left[\sum_{h=0}^{H-1} \kl\left(\widetilde{\pi}_h\left(\cdot \mid s_h\right) | \pi_h\left(\cdot \mid s_h\right)\right) \middle| s_0=s\right] \\
        & \quad=\mathbb{E}_{\tilde{\pi}, \vecmu}\left[\sum_{h=0}^{H-1}\left\langle q_h^{\tau}\left(s_h, \cdot | \vecmu, \pi\right)-\tau \log \pi_h\left(\cdot \mid s_h\right), \widetilde{\pi}_h\left(\cdot \mid s_h\right)-\pi_h\left(\cdot \mid s_h\right)\right\rangle \mid s_0=s\right].
    \end{align*}
\end{lemma}
\begin{proof}
    See the standard proof technique for the performance difference lemma, e.g. \cite{mei2020global, zhang2023learning}.
\end{proof}

Finally, we state two technical lemmas due to \cite{zhang2023learning}.

\begin{lemma}[Lemma I.3 of \cite{zhang2023learning}]\label{lemma:eps_mixing}
Let $p, p' \in \Delta_{\setA}$ be arbitrary, and $\widehat{p}=(1-\beta) p+\beta \operatorname{Unif}(\mathcal{A})$ for some $\beta\in (0,1)$.
Then, 
    \begin{align*}
\kl\left(p^* | \widehat{p}\right) & \leq \log \frac{|\mathcal{A}|}{\beta}, \\
\kl\left(p^* | \widehat{p}\right)-\kl\left(p^* | p\right) & \leq \frac{\beta}{1-\beta} .
    \end{align*}
\end{lemma}

\begin{lemma}[Lemma 3.3 in \cite{cai2020provably}]\label{lemma:technical_lemma_pol}
    Let $p, p^* \in \Delta_{\setA}$, $\alpha > 0$ and $g: \setA \rightarrow [0, G]$ be arbitrary, and $q\in \Delta_{\setA}$ be a distribution such that $q(\cdot) \, \propto \, p(\cdot) \exp\{ \alpha g(\cdot) \}$.
    Then,
    \begin{align*}
        \left\langle g(\cdot), p^*(\cdot)-p(\cdot)\right\rangle \leq \frac{\alpha G^2}{ 2}+\alpha^{-1}\left[\kl\left(p^* | p\right)-\kl\left(p^* | q\right)\right] .
    \end{align*}
\end{lemma}

\subsection{Extended Proof of Theorem~\ref{theorem:pmd}}\label{sec:proof_pmd}

As mentioned before, the proof is an adaptation of \cite{zhang2023learning} to setting where learning occurs with $N$ potentially asymmetric agents.
The main differences will be the absence of an explicit MFG and the fact that our algorithms are only allowed to use samples of finite agent trajectories.

Define the random variable $\vecmu_t := \{\mu_{t,h}\}_{h=0}^{H-1}:= \lpop(\pi_t)$, which is the mean-field population distribution induced by the policy at epoch $t$.
We denote the random variables due to estimation error of the q-functions at epoch $t$, time step $h$ and an arbitrary state $s$ as
\begin{align*}
    \setE^s_{t,h} := \left| \langle \qest{t}{h}(s,\cdot ) - q^\tau_{h}(s, \cdot|  \vecmu_{t}, \pi_{t} ), \pi^*_h(\cdot|s) -  \pi_{t,h}(\cdot|s) \rangle \right|.
\end{align*}
Furthermore, let $\pi^*$ be the unique $\tau$-regularized MFG-NE, and $\vecmu^* := \{\mu^*_h\}_{h=0}^{H-1} = \lpop(\pi^*)$.
We define the quantity
\begin{align*}
    \Delta_t := \sum_{h=0}^{H-1} \mathbb{E}_{s_h \sim \mu_h^{*}}\left[\kl\left(\pi_h^{*}\left(\cdot \mid s_h\right) \| \pi_{t, h}\left(\cdot \mid s_h\right)\right)\right],
\end{align*}
which will be the main quantity of error to be bounded using the techniques of \cite{zhang2023learning}.
We denote the mixing coefficients $\beta_t := \sfrac{1}{t+1}$ for generality.
Finally, we also define the distribution mismatch coefficients
\begin{align*}
   C_{\text{dist}} := \sup_{t, h} \sup_{\substack{s,a}} \frac{\mu^*_{h}(s,a)}{\mu_{t,h}(s,a)} , %
\end{align*}
which are always finite (and bounded) in our entropy-regularized setting (see for instance \cite{cayci2021linear}).
It is well-known that the PMD update with entropy regularization can be written as
\begin{align*}
    \widehat{\pi}_{t+1,h}(a|s) \propto \pi_{t,h}(a|s)^{1-\tau\plr{t}} \exp\{\plr{t} \qest{t}{h}(s,a) \},
\end{align*}
with appropriate normalization.

By Lemma~
\ref{lemma:eps_mixing}, it holds (almost surely) for any $s\in\setS$ that
\begin{align*}
    \kl(\pi^*_h(\cdot|s)| \pi_{t+1,h}(\cdot|s)) \leq \kl\left(\pi_h^{*}\left(\cdot \mid s\right) | \widehat{\pi}_{t+1, h}\left(\cdot \mid s\right)\right)+\frac{\beta_t}{1-\beta_t}.
\end{align*}
Using Lemma~\ref{lemma:technical_lemma_pol} and the fact that $\pi_t(a|s) \geq \beta_t$ for all $(s,a) \in \setS\times\setA$, 
\begin{align*}
    &\kl(\pi^*_h(\cdot|s)| \pi_{t+1,h}(\cdot|s)) \\
    &\leq -\plr{t}\left\langle\qest{t}{h}\left(s, \cdot \right)-\tau \log \pi_h\left(\cdot \mid s\right), \pi_h^{*}\left(\cdot \mid s\right)-\pi_{t, h}\left(\cdot \mid s\right)\right\rangle \\
        &\quad +\kl\left(\pi_h^{*}\left(\cdot \mid s\right) | \pi_{t, h}\left(\cdot \mid s\right)\right)+\frac{1}{2} \plr{t}^2\left(H+\tau H \log |\mathcal{A}|+\tau \log \frac{|\mathcal{A}|}{\beta_t}\right)^2+\frac{\beta_t}{1-\beta_t} \\
    &\leq -  \plr{t}\left\langle q_h^{\tau}\left(s, \cdot| \vecmu_t, \pi_t \right)-\tau \log \pi_h\left(\cdot \mid s\right), \pi_h^{*}\left(\cdot \mid s\right)-\pi_{t, h}\left(\cdot \mid s\right)\right\rangle \\
        &\quad +\kl\left(\pi_h^{*}\left(\cdot \mid s\right) | \pi_{t, h}\left(\cdot \mid s\right)\right)+\frac{1}{2} \plr{t}^2\left(H+\tau H \log |\mathcal{A}|+\tau \log \frac{|\mathcal{A}|}{\beta_t}\right)^2+\frac{\beta_t}{1-\beta_t}+\plr{t} \setE^s_{t,h},
\end{align*}
Then, using the performance difference lemma (Lemma~\ref{lemma:perf_diff}),
\begin{align*}
    &\Delta_{t+1} - \Delta_t \\
    &:= \sum_{h=0}^{H-1} \mathbb{E}_{\mu_h^{*}}\left[\kl\left(\pi_h^{*}\left(\cdot \mid s_h \right) | \pi_{t+1, h}\left(\cdot \mid s_h\right)\right)-\kl\left(\pi_h^{*}\left(\cdot \mid s_h\right) \| \pi_{t, h}\left(\cdot \mid s_h\right)\right)\right] \\
&\leq \plr{t} [V^{\tau}\left(\vecmu_t, \pi_t \right)-V^{\tau}\left( \vecmu_t, \pi^{*}\right)] -\tau \plr{t}  \sum_{h=0}^{H-1} \mathbb{E}_{\mu_h^{*}}\left[\kl\left(\pi_h^{*}\left(\cdot \mid s_h\right) \| \pi_{t, h}\left(\cdot \mid s_h\right)\right)\right]  \\
    & \quad +\frac{1}{2} \plr{t}^2 H\left(H+\tau H \log |\mathcal{A}|+\tau \log \frac{|\mathcal{A}|}{\beta_{t-1}}\right)^2+\frac{\beta_t}{1-\beta_t} H+2 \plr{t} \sum_{h=0}^{H-1} \mathbb{E}_{\mu_h^{*}}\left[\setE_{t,h}\right]  \\
&\leq -\tau \plr{t} \Delta_t +\frac{1}{2} \plr{t}^2 H\left(H+\tau H \log |\mathcal{A}|+\tau \log \frac{|\mathcal{A}|}{\beta_{t-1}}\right)^2+\frac{\beta_t}{1-\beta_t} H+2 \plr{t} \sum_{h=0}^{H-1} \mathbb{E}_{\mu_h^{*}}\left[\setE_{t,h}\right] ,
\end{align*}
where the last inequality follows from the monotone improvement lemma (Lemma~\ref{ref:monotone_improvement_lemma}) applied to the policy pair $\pi_t, \pi^*$.
Rearranging both sides,
\begin{align*}
    \Delta_t \leq &\frac{1}{\tau \plr{t}}\left(\Delta_t -\Delta_{t+1}\right)+\frac{\plr{t}}{2 \tau} H\left(H+\tau H \log |\mathcal{A}|+\tau \log \frac{|\mathcal{A}|}{\beta_{t-1}}\right)^2+\frac{\beta_t H}{(1-\beta_t) \tau \plr{t}} \\
& \quad+\frac{2}{\tau} \sum_{h=1}^H \mathbb{E}_{\mu_h^{*}}\left[\varepsilon_h\right]  .
\end{align*}
Summing this inequality from $t=1, \ldots, T$, we obtain
\begin{align*}
\frac{1}{T} \sum_{t=1}^T \Delta_t &\leq \frac{1}{T \tau \plr{t}} \Delta_1 +\frac{\plr{t}}{2 \tau} H\left(H+\tau H \log |\mathcal{A}|+\tau \log \frac{|\mathcal{A}|}{\beta_{t-1}}\right)^2+\frac{\beta_t H}{(1-\beta_t) \tau \plr{t}} \\
&\quad +\frac{2}{\tau} \sum_{h=1}^H \mathbb{E}_{\mu_h^{*}}\left[\varepsilon_h\right]  .
\end{align*}
Given that $\plr{t} = \sfrac{1}{\sqrt{t+1}}$ and $\beta_t = \sfrac{1}{t+1}$, we obtain the bounds
\begin{align*}
    \frac{1}{T} \sum_{t=1}^T \Delta_t = \frac{\tau^{-1} \Delta_1 + \tau^{-1} H^2 + \tau H\log|\setA| + \tau \log ^2 (T + 1)}{\sqrt{T}}+\frac{ \sum_{t=1}^T \sum_{h=0}^{H-1} \mathbb{E}_{\mu_h^{*}}\left[\setE_{t,h}^s\right] }{T\tau},
\end{align*}
and finally using Young's inequality on the last term, and an application of Pinsker's inequality,
\begin{align*}
    &\frac{1}{T}\sum_{t}\sum_h \Exop_{\mu_h^*}\left[\frac{1}{2}\| \pi_{t,h}(\cdot|s_h ) - \pi^*_{h}(\cdot|s_h )\|_1^2 \right] \leq\frac{1}{T} \sum_{t=1}^T \Delta_t \\
    &\leq\frac{\tau^{-1} \Delta_1 + \tau^{-1} H^2 + \tau H\log|\setA| + \tau \log ^2 (T + 1)}{\sqrt{T}}+\frac{ 2\sum_{t=1}^T \sum_{h=0}^{H-1} \mathbb{E}_{\mu_h^{*}}\left[\setE_{t,h}^s\right] }{T\tau}, \\
    &\leq\frac{\tau^{-1} \Delta_1 + \tau^{-1} H^2 + \tau H\log|\setA| + \tau \log ^2 (T + 1)}{\sqrt{T}} +\frac{ \sum_{t=1}^T \sum_{h=0}^{H-1} \mathbb{E}_{\mu_h^{*}}\left[8\| \qest{t}{h}(s_h,\cdot) - q^\tau_{h}(s_h,\cdot|\vecmu_t,\pi_t)\|_2^2 \right] }{T\tau^{2}}, \\
        &\quad +\frac{ \sum_{t=1}^T \sum_{h=0}^{H-1} \mathbb{E}_{\mu_h^{*}}\left[ \| \pi_{t,h}(\cdot|s_h) - \pi^*(\cdot|s_h)\|_2^2 \right] }{4T}.
\end{align*}
Rearranging the terms, 
\begin{align*}
    &\frac{1}{T}\sum_{t}\sum_h \Exop_{\mu_h^*}\left[\frac{1}{2}\| \pi_{t,h}(\cdot|s_h ) - \pi^*_{h}(\cdot|s_h )\|_1^2 \right] \\
    &\leq \frac{4\tau^{-1} \Delta_1 + 4\tau^{-1} H^2 + 4\tau H\log|\setA| + 4\tau \log ^2 (T + 1)}{\sqrt{T}} \\
    &\quad +\frac{ \sum_{t=1}^T \sum_{h=0}^{H-1} \mathbb{E}_{\mu_h^{*}}\left[32\| \qest{t}{h}(s_h,\cdot) - q^\tau_{h}(s_h,\cdot|\vecmu_t,\pi_t)\|_2^2 \right] }{T\tau^{2}}.
\end{align*}
Finally, noting that by Theorem~\ref{theorem:td}, after taking expectations on both sides it holds that $\Exop[\mathbb{E}_{\mu_h^{*}}\left[4\| \qest{t}{h}(s_h,\cdot) - q^\tau_{h}(s_h,\cdot)\|_2^2 \right] ] = \cO(\varepsilon^2 + \alpha^2 + \beta^2 +\sfrac{1}{N})$, if follows that $\Exop[\mathbb{E}_{\mu_h^{*}}\left[ \sum_h \| \widebar{\pi}_h(\cdot|s_h) - \pi^*_h(\cdot|s_h)\|_2^2 \right] ] \leq \cO(\tau^{-2}\varepsilon^2 + \tau^{-2}\alpha^2 + \tau^{-2}\beta^2 +\tau^{-2}\sfrac{1}{N}) $ .
Using the standard Lipschitz continuity of exploitability (see e.g. \cite{yardim2023policy}), the exploitability bound in expectation holds.
Using Lemma~\ref{lemma:reg_bias}, we obtain the upper bound in expectation on the exploitability of the output policy $\widebar{\pi}$ in terms of the original (unregularized) DG.

%% file: appendix/experiment_details.tex
\section{Details of Experiments}\label{sec:experiment_details}

\subsection{Hardware Setup for Experiments}

Except the A-Taxi benchmark, all our experiments are CPU-based.
We use a single AMD EPYC 7742 CPU, equipped with 128GB RAM.
For training policy and value neural networks with PPO in the A-Taxi benchmark, we use a single RTX 3090 GPU.
With this setup, running Symm-PMD and IPMD in the A-SIS and A-RPS benchmarks takes between 5-20 minutes, and running PPO on A-Taxi takes approximate 2 hours.
Evaluating exploitability for a given policy on A-SIS and A-RPS takes around 2 hours, as we employ a brute-force UCB-type bandit algorithm to accurately estimate best response in this setting.

\subsection{Extended Descriptions of the Experimental Setup}

For simplified notation, denote the state-action marginal densities
\begin{align*}
    \sigma_{\text{actions}}(\vecrho,a) &= \sum_{s'\in\setS}\empc{\vecrho}(s',a), \\
    \sigma_{\text{states}}(\vecrho,s) &= \sum_{s'\in\setA}\empc{\vecrho}(s,a').
\end{align*}

\textbf{Modified rock-paper-scissors (A-RPS). }
We formulate a modified population rock-paper-scissors game inspired by the formulation of \cite{cui2021approximately}.
Our version incorporates varying preferences between agents between possible moves as well as a crowdedness penalty.

A-RPS consists of three states $\setS := \{ R, P, S\}$ and three actions $\setA := \{ R, P, S\}$.
We use $N=2000$ players and a time horizon of $H=10$, though these can be increased or decreased arbitrarily.
We define the rewards as follows.
\begin{align*}
    R^i(s=R,a,\vecrho^{-i}) &= - c^i \sigma_{\text{actions}}(\vecrho,a) - u^i_R \sigma_{\text{states}}(\vecrho,P) + v^i_R \sigma_{\text{states}}(\vecrho,S), \\
    R^i(s=P,a,\vecrho^{-i}) &= - c^i \sigma_{\text{actions}}(\vecrho,a) - u^i_P \sigma_{\text{states}}(\vecrho,S) + v^i_P \sigma_{\text{states}}(\vecrho,R), \\
    R^i(s=S,a,\vecrho^{-i}) &= - c^i \sigma_{\text{actions}}(\vecrho,a) - u^i_S \sigma_{\text{states}}(\vecrho,R) + v^i_S \sigma_{\text{states}}(\vecrho,P). \\
\end{align*}
The coefficients $c^i,u^i, v^i$ are unique for each agent indicating their own utilities/rewards due to losing, winning, or individual penalty due to crowdedness.
The state transitions are deterministic and are given by:
\begin{align*}
    P^i(s'|s,a,\vecrho^{-i}) = \ind{s'=a}.
\end{align*}
We generate the fixed coefficients $u^i, v^i$ randomly by adding bounded noise to coefficients from \cite{cui2021approximately}, so that
\begin{align*}
    &u_R^i = 2 + \varepsilon^i_R, \quad v_R^i = 1 + \widebar{\varepsilon}^i_R \\
    &u_P^i = 4 + \varepsilon^i_P, \quad v_P^i = 2 + \widebar{\varepsilon}^i_P \\
    &u_S^i = 6 + \varepsilon^i_S, \quad v_S^i = 3 + \widebar{\varepsilon}^i_S.
\end{align*}
Therefore, the magnitudes of the player-specific additive terms determine $\beta$.
In the case of A-RPS, $\alpha=0$.

\textbf{Infection modeling with asymmetric agents (A-SIS).}
This benchmark, inspired by the SIS benchmark of \cite{cui2021approximately}, models a large population of infected or healthy agents that can choose to go out or remain in isolation.
Unlike the SIS benchmark, A-SIS is formulated as an $N$-player game and incorporates individual differences in natural susceptibilities, recovery rates, and aversion of isolation between agents.
We formalize the dynamic game as follows.
The game consists of the state space $\setS = \{ I, H \}$ ($I$ indicating infected, $H$ indicating healthy), and action space $\setA = \{ D, U \}$ ($D$ indicating social distancing, $U$ indicating going out).
The initial states $s_0^i$ are sampled i.i.d. from a uniform distribution over $\setS$.
Each agent $i\in N$ has a fixed \emph{susceptibility} parameter $\alpha_i \in [0,1]$, a fixed \emph{healing probability} $\theta_i \in [0,1]$ and a fixed \emph{aversion to isolation} parameter $\xi_i \in [0,1]$.
\begin{align*}
    P^i(I | H, D , \vecrho^{-i} ) &= 0 \\
    P^i(I | H, U , \vecrho^{-i} ) &= \alpha_i * \empc{\vecrho}(I, U),\\ 
    P^i(I | I, D , \vecrho^{-i} ) &= 1 - \theta_i,\\
    P^i(I | I, U , \vecrho^{-i} ) &= 1 - \theta_i.
\end{align*}
The probabilities of staying healthy are of course always defined by
\begin{align*}
    P^i(H | s,a, \vecrho^{-i} ) := 1 - P^i(I | s,a, \vecrho^{-i} ) .
\end{align*}
The rewards of each agent are give by the following which incorporates a penalty for illness and an agent specific penalty for isolation:
\begin{align*}
    R^i(s,a,\vecrho^{-i}) = - \ind{s=I} - \xi_i \ind{a=D}.
\end{align*}
The agent parameters $\alpha_i,\theta_i,\xi_i$ are as expected fixed throughout the game, and are sampled to be close.
In the case of A-SIS, we solve $N=1000$ agents with a time horizon of $H=20$.

\textbf{Asymmetric taxi (A-Taxi).}
Finally, as a more complicated benchmark we adapt the Taxi 
As in \cite{cui2021approximately}, we use the following layout of the city map, where $S$ indicates the starting cell of all agents, both $H$ and $S$ are impenetrable barriers and the rest of the city is divided into 2 zones.
\begin{align*}
    \left(\begin{array}{ccc}
1 & 1 & 1 \\
1 & 1 & 1 \\
1 & 1 & 1 \\
H & S & H \\
2 & 2 & 2 \\
2 & 2 & 2 \\
2 & 2 & 2
\end{array}\right)
\end{align*}
The action space of each agent is $\setA := \{U,D,L,R,W\}$, indicating actions to move in four directions (up, down, left, right respectively) and wait at the current location.
Customers can only be picked up while waiting, and delivered while waiting.
Each cell in the grid generates a new customer with probability $0.2$, which the agents can observe.
Upon picking up a customer, a random target coordinate is generate within the same zone.
Customers can only be left at their target cells.
Successful deliveries of customers in zone 1 generate a base reward of \num{1.1}, whereas successful deliveries of customers in zone 2 generate a lower reward of \num{1.0}.
Furthermore, each agent has a zone specific reward multiplier $\beta^1_i, \beta^2_i > 0$, so that agent $i$ by delivering a customer in zone $1$ gains reward $1.1 \beta_i^1$ and vice versa.
This models varying efficiencies of taxi drivers as well as individual preferences to various zones of the city.
Furthermore, we incorporate a crowdedness penalty: an agent $i\in[N]$ at state $s$ at time $h$ will not move (simulating a jamming effect) with probability $\min\{ \sum_{j} w_j \ind{s^j_h = s} , 0.7\}$, where $\{w_j\}_{j=1}^N$ are player specific weights indicating their contribution to traffic jams.
This intends to simulate unique contributions of each driver to traffic jams, presumably due to vehicle types, driving styles, etc.

The number of states in the game is on the order of $2^{30}$, making a neural network approximation fundamental.
For this reason, we use value and policy networks with two hidden layers with 128 neurons each, with a leaky ReLU nonlinearity.
We adopt the PPO implementation of CleanRL \cite{huang2022cleanrl} for our purposes.
The hyperparameters used are indicated in Table~\ref{tab:ppo_details}.

\begin{table*}[t]
\centering
  \begin{tabular}{lc} \toprule
    \textbf{Parameter} & \textbf{Value} \\
    \midrule
    Initial learning rate & \num{2.5e-4}\\
    Learning rate schedule & Linear\\
    $\gamma$ (discount factor) & \num{0.999} \\
    $\lambda_{\text{GAE}}$ 
 (see\cite{schulman2015high})& \num{0.95} \\
    Entropy regularization & \num{6e-3} \\ 
    Value loss coefficient & \num{0.9} \\
    Maximum gradient norm & \num{0.5} \\
    Clip coefficient & \num{0.2} \\
    Sample trajectories per epoch & \num{1} \\
    NN training passes per epoch & \num{4} \\
    Minibatch size & \num{16384} \\
    Advantage normalization & Yes \\
    \bottomrule
  \end{tabular}
  \caption{Hyperparameters of the PPO algorithm.}
  \label{tab:ppo_details}
\end{table*}

\subsection{Extended Experimental Results}

We report two additional sets of results regarding the sensitivity of our algorithms to $\alpha,\beta$ and the population distribution behvaiour in the A-Taxi environment.

In Figure~\ref{figure:appendix_exp}-(a), we report the sensitivity of the exact MFG-NE (computed via code provided in \cite{guo2023mfglib}) to heterogeneity parameters $\alpha,\beta$ in terms of exploitability in the $N=1000$ player game.
While keeping the other parameter constant at $0$, we sweep through various values of each of $\alpha,\beta$ in the range $(0,\sfrac{1}{4})$.
While around the \num{0.1} threshold, the exploitability rises as expected, for smaller values of $\alpha,\beta$ the bias introduced is very small, providing an empirical analysis of the approximation bound of Theorem~\ref{theorem:main_approximation}.

\begin{figure}[h]
\begin{tabular}{cc}
  \includegraphics[width=0.44\linewidth]{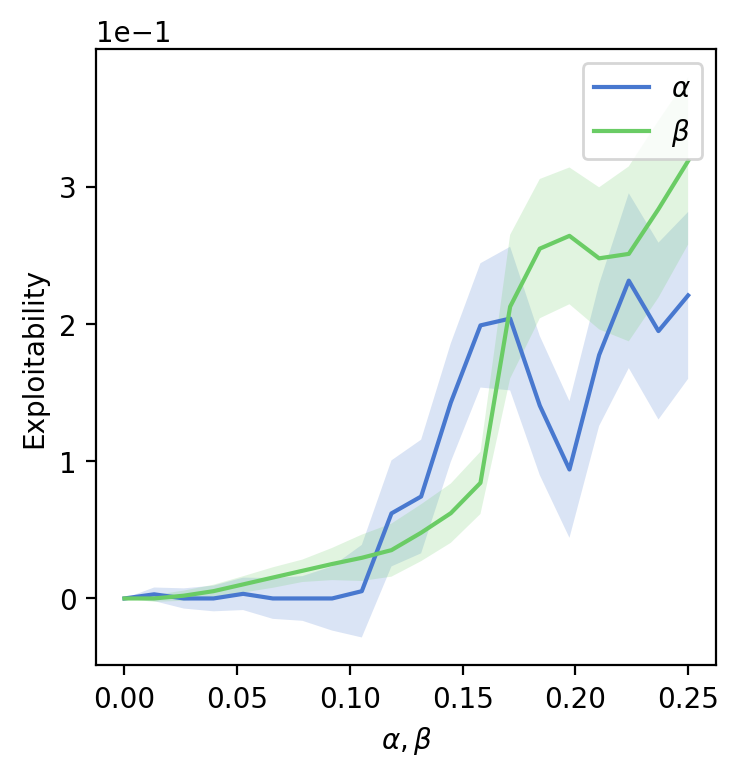} &
  \includegraphics[width=0.46\linewidth]{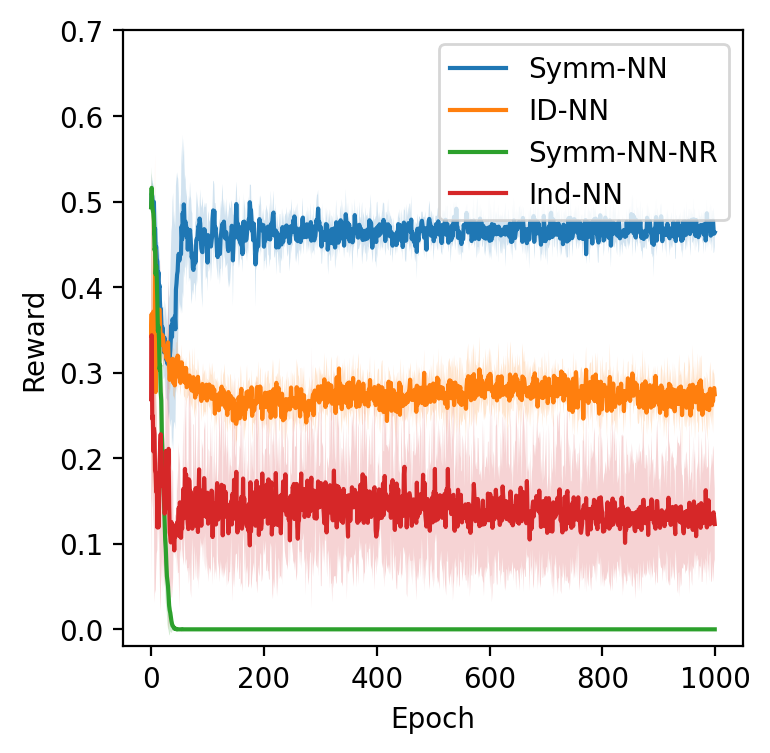 } \\
(a) & (b) 
\end{tabular}
\caption{
(a) The sensitivity of the MFG-NE to heterogeneity parameters $\alpha,\beta$ in the A-SIS environment, in terms of exploitability.
(b) Percentage of vehicles in Zone 1 in the A-Taxi environment throughout training epochs for 4 benchmark algorithms.
}
\label{figure:appendix_exp}
\end{figure}

In Figure~\ref{figure:appendix_exp}-(b), we keep track of number of taxis choosing to operate in Zone 1 in the A-Taxi environment throughout training.
While rewards in Zone 2 are higher in this environment, congestion effects require a mixed Nash equilibrium: agents must randomly choose at the very first step to serve either Zone 1 or 2.
The figure demonstrates the main advantage of policy-based methods for learning Nash: unlike most value-based methods, PPO can learn a mixed strategy instead of converging to a deterministic policy.
As an additional benchmark, in the figure we evaluate Symm-NN without any entropy regularization ($\tau=0$, shown by the line Symm-NN-NR).
In this case, the policy rapidly converges to a deterministic policy, indicating that a non-zero entropy regularizer might be necessary for learning a Nash equilibrium.